\documentclass[11pt,a4paper,reqno]{amsart}
\usepackage[hmargin={2.7cm,2.7cm},vmargin={2.5cm,2.5cm},centering]{geometry}
\usepackage{amssymb,amsmath,amsthm,amsfonts,amscd}
\usepackage{mathabx}
\usepackage{graphicx}
\usepackage[dvipsnames]{xcolor}
\usepackage{mathrsfs}
\usepackage[unicode,psdextra]{hyperref}
\hypersetup{pdfborder={0 0 0.06},citebordercolor=[rgb]{0.8196,0.2275,0.5098},linkbordercolor=[rgb]{0.1765,0.5490,0.8118},urlbordercolor=[rgb]{0.7059,0.5333,0.1137}}
\usepackage{tikz-cd}
\usepackage[english]{babel}
\usepackage{dsfont}
\usepackage{enumerate}
\usepackage{setspace}
\usepackage[font={scriptsize,bf}]{caption}
\usepackage{changepage}
\usepackage{booktabs}
\usepackage{tikz-cd}
\usepackage[sort&compress,capitalise,nameinlink]{cleveref}
\crefname{section}{\textsection}{\textsection}
\crefname{appendix}{\textsection}{\textsection}
\crefname{subsection}{\textsection}{\textsection}
\crefname{subsubsection}{\textsection}{\textsection}
\crefname{paragraph}{\textparagraph}{\textparagraph}
\crefname{thm}{Theorem}{Theorem}
\usepackage{mathtools}
\DeclareMathOperator{\tr}{tr}
\DeclareMathOperator{\Tr}{Tr}

\renewcommand{\Im}{\mathrm{Im}}
\renewcommand{\Re}{\mathrm{Re}}
          \newtheorem{thm}{Theorem}[section]
          \newtheorem{proposition}[thm]{Proposition}
          \newtheorem{lemma}[thm]{Lemma}
          \newtheorem{corollary}[thm]{Corollary}
          \newtheorem{definition}[thm]{Definition}
          
          \theoremstyle{definition}

          \newtheorem{remark}[thm]{Remark}

\usepackage{bm}
\renewcommand{\vec}[1]{\bm{#1}}
\usepackage{csquotes}

\usepackage[authormarkup=none,deletedmarkup=xout,%final
]{changes}

\definechangesauthor[name={Marco O}, color={brown}]{Marco O}

\definechangesauthor[name={Michele}, color={Emerald}]{Michele}

\newcommand{\one}{\mathds{1}}

\newcommand{\xv}{\mathbf{x}}

\newcommand{\kv}{\mathbf{k}}

\newcommand{\kvp}{\mathbf{k}^{\prime}}

\newcommand{\mm}{\mathfrak{m}}
\newcommand{\nn}{\mathfrak{n}}
\newcommand{\dom}{\mathscr{D}}
\newcommand{\diff}{\mathrm{d}}
\newcommand{\eps}{\varepsilon}

\newcommand{\la}{\lambda}
\newcommand{\aav}{\mathbf{A}}

\newcommand{\beq}{\begin{equation}}
\newcommand{\eeq}{\end{equation}}

\newcommand{\OO}{\mathcal{O}}

\newcommand{\KK}{\mathcal{K}}
\newcommand{\GG}{\mathcal{G}}

\newcommand{\R}{\mathbb{R}}
\newcommand{\N}{\mathbb{N}}

\newcommand{\bdm}{\begin{displaymath}}
\newcommand{\edm}{\end{displaymath}}
\newcommand{\bdn}{\begin{eqnarray}}
\newcommand{\edn}{\end{eqnarray}}
\newcommand{\bay}{\begin{array}{c}}
\newcommand{\eay}{\end{array}}
\newcommand{\ben}{\begin{enumerate}}
\newcommand{\een}{\end{enumerate}}
\newcommand{\beqn}{\begin{eqnarray}}
\newcommand{\eeqn}{\end{eqnarray}}
\newcommand{\bml}[1]{\begin{multline} #1 \end{multline}}
\newcommand{\bmln}[1]{\begin{multline*} #1 \end{multline*}}

\newcommand{\lf}{\left}
\newcommand{\ri}{\right}
\newcommand{\disp}{\displaystyle}
\newcommand{\tx}{\textstyle}

\newcommand{\bra}[1]{\lf\langle #1\ri|}
\newcommand{\ket}[1]{\lf|#1 \ri\rangle}
\newcommand{\braket}[2]{\lf\langle #1|#2 \ri\rangle}
\newcommand{\braketr}[2]{\lf\langle #1\lf|#2\ri. \ri\rangle}
\newcommand{\braketl}[2]{\lf.\lf\langle #1\ri|#2 \ri\rangle}

\newcommand{\meanlrlr}[3]{\lf\langle #1\lf|#2\ri|#3\ri\rangle}

\renewcommand{\leq}{\leqslant}
\renewcommand{\geq}{\geqslant}

\newcommand{\gint}{\Upsilon_{\eps}}
\newcommand{\hgint}{\widehat{\Upsilon}_{\eps}}
\usepackage{extpfeil}

\numberwithin{equation}{section}

\begin{document}
\tolerance=2000
\setlength{\emergencystretch}{1em}

\newenvironment{sistema}%
{\left\lbrace\begin{array}{@{}l@{}}}%
    {\end{array}\right.}
\title{Quasi-Classical Dynamics}

\author[M.\ Correggi]{Michele Correggi}

\address{Dipartimento di Matematica, Politecnico di Milano, P.zza Leonardo da Vinci, 32, 20133, Milano, Italy}
\email{michele.correggi@gmail.com}

\urladdr{https://sites.google.com/view/michele-correggi/}

\author[M.\ Falconi]{Marco Falconi}

\address{Dipartimento di Matematica e Fisica, Universit\`{a} di Roma Tre, Largo S.\ Leonardo Murialdo 1/C, 00146, Roma}

\email{mfalconi@mat.uniroma3.it}

\urladdr{http://ricerca.mat.uniroma3.it/users/mfalconi/}

\author[M.\ Olivieri]{Marco Olivieri}

\address{Fakult\"{a}t f\"{u}r Mathematik, Karlsruher Institut f\"{u}r Technologie, D-76128, Karlsruhe, Germany.}

\email{marco.olivieri@kit.edu}

\urladdr{}

\keywords{Quasi-classical limit, Interaction of matter and light, Open quantum systems, Semiclassical analysis.}

\subjclass[2010]{Primary: 81Q20, 81T10. Secondary: 81V10, 81Q10.}

\date{\today}

\begin{abstract}
  We study quantum particles in interaction with a force-carrying field, in the quasi-classical
  limit. This limit is characterized by the field having a very large number of excitations (it is
  therefore macroscopic), while the particles retain their quantum nature. We prove that the interacting
  microscopic dynamics converges, in the quasi-classical limit, to an effective dynamics where the field
  acts as a classical environment that drives the quantum particles.
\end{abstract}

\maketitle

\onehalfspacing{}

%%%%%%%%%%%%%%%%%%%%%%%%%%%%%%%%%%%%%%%%%%%%%%%%%%%%%%%%%%%%%%%%%%%%%%%%%%%%%%%%%%%%%%%%%%%%%%%%%%%%%%%%%
%%%%%%%%%%%%%%%%%%%%%%%%%%%%%%%%%%%%%%%%%%% TODO List %%%%%%%%%%%%%%%%%%%%%%%%%%%%%%%%%%%%%%%%%%%%%%%%%%%
%%%%%%%%%%%%%%%%%%%%%%%%%%%%%%%%%%%%%%%%%%%%%%%%%%%%%%%%%%%%%%%%%%%%%%%%%%%%%%%%%%%%%%%%%%%%%%%%%%%%%%%%%
%%%
%%% Scrivere un outline della strategia di dimostrazione, coi vari step necessari
%%%
%%%%%%%%%%%%%%%%%%%%%%%%%%%%%%%%%%%%%%%%%%%%%%%%%%%%%%%%%%%%%%%%%%%%%%%%%%%%%%%%%%%%%%%%%%%%%%%%%%%%%%%%%
%%%%%%%%%%%%%%%%%%%%%%%%%%%%%%%%%%%%%%%%%%%%%%%%%%%%%%%%%%%%%%%%%%%%%%%%%%%%%%%%%%%%%%%%%%%%%%%%%%%%%%%%%
%%%%%%%%%%%%%%%%%%%%%%%%%%%%%%%%%%%%%%%%%%%%%%%%%%%%%%%%%%%%%%%%%%%%%%%%%%%%%%%%%%%%%%%%%%%%%%%%%%%%%%%%%

\setcounter{tocdepth}{1}
\tableofcontents
 
\section{Introduction and Main Results}
\label{sec:introduction}

This paper is devoted to the study of the quasi-classical dynamics of a coupled quantum system composed of finitely many non-relativistic particles interacting with a bosonic field. The quasi-classical regime is concretely realized by taking a suitable partial
semiclassical limit, introduced by the authors in \cite{correggi2017ahp,correggi2017arxiv} to derive
external potentials as effective
interactions emerging from the particle-field coupling. The physical meaning of such limit is discussed in
\cref{sec:physical-motivation}. 

Our analysis clarifies, both mathematically and physically, the role played by external macroscopic
classical force fields on quantum systems, and in which regime such macroscopic fields provide an accurate
description of the interaction between an open quantum system and its environment (bosonic field).

In order to study the dynamical quasi-classical limit, we develop a mathematical framework of infinite
dimensional quasi-classical analysis, in analogy with the semiclassical scheme initially introduced and in
\cite{ammari2008ahp,ammari2009jmp,ammari2011jmpa,ammari2015asns}, and further discussed in
\cite{falconi2017ccm,falconi2017arxiv}. Such a framework allows to characterize the quasi-classical
behavior of quantum states which are not factorized, {\it i.e.}, in which the degrees of freedom of the quantum particles and the bosonic field are entangled.  Although our mathematical scheme is more general, we are going to focus our attention on three concrete models of interaction between particles and
force-carrying fields: the Nelson, Pauli-Fierz, and Fr\"{o}hlich polaron models (see \cref{sec:concr-models:-nels}). 
Note that partial semiclassical limits have already been studied, with somewhat different purposes, in
\cite{ginibre2006ahp,amour2015arxiv,amour2017arxiv,amour2017jmp}, as well as in the context of adiabatic theories (see, \emph{e.g.}, \cite{teufel2002ahp,panati2003atmp,tenuta2008cmp,stiepan2013cmp}).

The paper is organized as follows. In the rest of
  \cref{sec:introduction} we introduce the paper's mathematical framework,
  and we formulate and motivate our results. In
\cref{sec:quasi-class-analys} we develop the main technical tools for the
subsequent analysis, that we call quasi-classical analysis, in analogy with
the more familiar semiclassical analysis. In fact, quasi-classical analysis
is semiclassical analysis on a bipartite system, where only one part is
semiclassical, and the other is quantum. In \cref{sec:microscopic-model} we
describe the relevant features of the microscopic Nelson model, that we use
as a reference to explain the strategy of the proof of \cref{thm:5}
below. We then take the limit $\varepsilon\to 0$ of the microscopic integral
equation of motion in \cref{sec:quasi-class-limit}, while we discuss in
\cref{sec:uniq-quasi-class} the uniqueness of solutions to the
quasi-classical equation obtained performing the aforementioned limit. In
\cref{sec:putting-it-all} we put together the results obtained in
\crefrange{sec:quasi-class-analys}{sec:uniq-quasi-class}, and prove
\cref{thm:5} for the Nelson model, and thus consequently also
\cref{cor:7,thm:8}. In \cref{sec:techn-modif-polar}, we provide the technical
modifications needed to prove the aforementioned theorems, for the
Pauli-Fierz and polaron models.

\bigskip

\begin{small}
  \noindent {\bf Acknowledgements:} The authors would like to thank Z.\
  Ammari, for many helpful discussions during the redaction of the paper, and
  the anonymous referees, for all the stimulating comments,
  particularly in relation to the use of the (PI) condition to obtain bounded
  quasi-classical convergence, and to the connection between the
    hypotheses \eqref{eq: ass1} and \eqref{eq: ass2}.  M.C.\ and M.O.\ are
  especially grateful to the Institut Mittag-Leffler, where part of this work
  was completed. M.F.\ has been supported by the
  Swiss National Science Foundation via the grant ``Mathematical Aspects of
  Many-Body Quantum Systems'', and by the European Research Council
    (ERC) under the European Union’s Horizon 2020 research and innovation
    programme (ERC CoG UniCoSM, grant agreement n.724939). M.O.\ has been
  partially supported by GNFM group of INdAM through the grant Progetto
  Giovani 2019 ``Derivation of effective theories for large quantum
  systems''.
\end{small}

\subsection{Physical Motivation}
\label{sec:physical-motivation}

The quasi-classical description, combining a quantum system with a classical force field, is often used in
physics to model external macroscopic forces acting on a quantum particle system. The best known examples
are atoms and electrons in a classical electromagnetic field (see, \emph{e.g.},
\cite{cohen-tannoudji1997wiley}), and particles subjected to external potentials, such as systems of
trapped atoms and of particles in optical lattices. Since these external force fields are macroscopic,
they are heuristically taken as classical, and inserted in the particles' Hamiltonian in the same way
their microscopic counterparts would appear. Note that in literature the terminology ``quasi-classical''
is often used as a synonymous of semiclassical, while here we use it to stress that the classical limit
we consider is not complete, but applies only to a part (radiation field or environment) of the
microscopic system.

In this paper we provide a detailed analysis of the quasi-classical dynamical scheme, and discuss its validity
as an approximation of a more fundamental microscopic model, thus justifying and completing the above
heuristic picture. The basic idea is the following: in experiments, the external force fields are
considered macroscopic because they live on an energy scale much larger than the ones of the quantum
particles under study: the number of field's excitations is much larger than the number of quantum
particles in the system. Let us denote by $ N $ the number of particles in the system. The force field is
itself a quantum object, and its excitations are created and annihilated by the interaction with the
particles. Let us denote the field's number operator by 
\beq
	\diff\mathcal{G}(1) =\int_{}^{}\mathrm{d}k \: a^{\dagger}(k)a(k),
\eeq
where $ \mathcal{G} $ stands for the second quantization functor. Therefore, the field is macroscopic if  the state $\Psi$ of the coupled system
particles+field is such that $
 	\meanlrlr{\Psi}{\diff\mathcal{G}(1) }{\Psi} \gg N $.
The number of particles $ N $ is fixed, and therefore of order 1. In other words, the quasi-classical
configurations are the ones for which
\begin{equation}
  \meanlrlr{\Psi}{\diff\mathcal{G}(1)}{\Psi} \gg 1\; .
\end{equation}
We thence introduce a quasi-classical parameter $\varepsilon$, playing the role of a semiclassical parameter but
only for the field's degrees of freedom: when $\varepsilon\to 0$, the system becomes quasi-classical. We quantify $\varepsilon$
as follows: a quasi-classical state $\Psi_{\varepsilon}$ is a state such that
\begin{equation*}
  \meanlrlr{\Psi_{\eps}}{\diff\mathcal{G}(1)}{\Psi_{\eps}}  \sim \frac{1}{\varepsilon}\; .
\end{equation*}
In other words $\varepsilon$ is proportional to the inverse of the average number of excitations of the force-carrying
field. It follows that on quasi-classical states,
\begin{equation}
	\label{eq: number of carriers}
 	\meanlrlr{\Psi_{\varepsilon}}{\varepsilon \diff\mathcal{G}(1) }{\Psi_\varepsilon} = \int_{}^{} \diff \kv \: \meanlrlr{\Psi_{\varepsilon}}{\varepsilon a^{\dagger}(\kv) a(\kv)}{\Psi_{\varepsilon}} = \int_{}^{} \diff \kv \: \meanlrlr{\Psi_{\varepsilon}}{a_{\eps}^{\dagger}(\kv) a_{\eps}(\kv)}{\Psi_{\varepsilon}} \leq C\;,
\end{equation}
where $a^{\#}_{\varepsilon}(\,\cdot \,) :=\sqrt{\varepsilon}a^{\#}(\,\cdot \,)$. The creation and annihilation operators $a^{\#}_{\varepsilon}$
satisfy $\varepsilon$-dependent semiclassical canonical commutation relations:
\begin{equation}
	\label{eq: ccr}
  	\lf[ a_{\varepsilon}(\kv),a^{\dagger}_{\varepsilon}(\kvp) \ri] = \varepsilon\delta(\kv- \kvp)\; .
\end{equation}
It is therefore clear that a quasi-classical state is a state that behaves semiclassically only with
respect to the field's degrees of freedom.

It remains to understand which microscopic dynamics would yield, in the quasi-classical limit, an external potential acting on the particles and generated by the macroscopic field. In concrete applications, the macroscopic field is not affected by the quantum system and acts as an environment. Therefore, the coupling should be such that the particles do not back-react on the environment, at least to leading order in $ \eps $ and for times of order $ 1 $. In addition, we may think that the environment itself either evolves freely, or it remains constant in time. The absence of back-reaction is determined by the $\varepsilon$-scaling of the microscopic interaction, while the dynamical behavior of the environment is determined by the $\varepsilon$-scaling of the field's free part. It turns out that it is indeed possible to tune the scaling of the interaction in such a way that, in the limit $\varepsilon\to 0$, the latter is precisely \emph{weak enough} to make the system decouple only partially: the classical field obeys a linear, unperturbed, evolution, while the quantum system's dynamics is driven by the classical field itself. The scaling yielding such a behavior is introduced in \cref{sec:concr-models:-nels}, and discussed in \cref{sec:main-results}. Let us stress the fact that, contrarily to a complete semiclassical limit, in the quasi-classical regime the aforementioned partial decoupling prevents any nonlinearity to appear in the effective dynamics of both the classical field and the quantum system.

In~\cref{sec:main-results} we prove that the quasi-classical description can be rigorously obtained from microscopic models of particle-field interaction in the limit $\varepsilon\to 0$ of a very large number of average field's excitations. Since such limit is a semiclassical limit on the field only, the resulting structure of quasi-classical systems is that of an hybrid quantum/classical probability theory. The quantum system is driven by the classical environment, whose configuration is a classical probability with values in the quantum states for the particles. This mathematical structure is described in detail in next \cref{sec:gener-struct-micr}.

\subsection{Notation}
\label{sec: notation}

Since we are going to consider a tensor product Hilbert space of the form $ \mathscr{H} \otimes \mathscr{K}_{\eps} $, we will distinguish between the full trace of operators $\Tr(\,\cdot \,)$ on $\mathscr{H}\otimes
\mathscr{K}_{\varepsilon}$, and the partial traces $\tr_{\mathscr{H}}(\,\cdot \,)$ and $\tr_{\mathscr{K}_{\varepsilon}}(\,\cdot \,)$ 
w.r.t.\ $\mathscr{H}$ and $\mathscr{K}_{\varepsilon}$, respectively.

We adopt the following convenient notation: an operator acting {\it only} on the particle space $\mathscr{H}$ is denoted
by a \emph{calligraphic capital letter} (\emph{e.g.}, $\mathcal{T}$ or $ \mathcal{T}_{\eps} $), whereas an operator on the full space
$\mathscr{H}\otimes \mathscr{K}_{\varepsilon}$ is identified by a \emph{roman capital letter} (\emph{e.g.}, $H_{\varepsilon}$). Given an operator
$\mathcal{T}$ on $\mathscr{H}$, we also conveniently denote its extension to
$\mathscr{H}\otimes \mathscr{K}_{\varepsilon}$, \emph{i.e.}\ $T=\mathcal{T}\otimes 1$, by the roman counterpart $T$. 

Given a Hilbert space $\mathscr{X}$, we denote by $ \mathscr{L}^p(\mathscr{X})$, $p\in [1,\infty]$, the
$p$-th Schatten ideal of $\mathscr{B}(\mathscr{X})$, the space of bounded operators on
$\mathscr{X}$. More in general the set $ \mathcal{L}(\mathscr{X}) $ identifies all linear operators on $ \mathscr{X} $. We also denote by $\mathscr{L}^p_+(\mathscr{X})$ and $ \mathscr{B}_+(\mathscr{X}) $ the cones of positive elements, and by
$\mathscr{L}^p_{+,1}(\mathscr{X})$ the set of positive elements of norm one. The corresponding norms are denoted by keeping track of the space, except for the case of the operator norm, for which we use the short notation $ \lf\| \: \cdot \: \ri\| : = \lf\| \: \cdot \: \ri\|_{\mathscr{B}(\mathscr{X})} $.

The space of finite measures on a measure space $ (X, \Sigma) $ is denoted
by $ \mathscr{M}(X,\Sigma) $, while the subset of probability measures is $
\mathscr{P}(X,\Sigma) $. If $X$ is an Hausdorff topological space and
  $\Sigma$ is the Borel $\sigma$-algebra, we denote by $\mathscr{M}(X) $ the finite
  Radon Borel measures on $X$, and by $\mathscr{P}(X)$ its subset of
  probability measures.

Throughout the paper, given a set $S$ we
denote by $ \one_S$ its indicator function. The symbol $ C $ also stands for a finite positive constant, whose value may vary from line to line.

\subsection{Quasi-Classical System}
\label{sec:gener-struct-micr}

We consider a microscopic system consisting of two parts in interaction. The first one contains the objects
whose microscopic nature remains relevant, while the second is a semiclassical environment. For the sake of
clarity, we focus on a specific class of systems: non-relativistic quantum particles in interaction with a
semiclassical bosonic force-carrying field (electromagnetic, vibrational, {\it etc}.). It is not difficult to
adapt the techniques to other coupled systems as well, consisting of a quantum and a semiclassical
part. We denote by $\mathscr{H}$ the Hilbert space of the quantum part, and by $\mathscr{K}_{\varepsilon}$ the
Hilbert space of the semiclassical part, that carries an $\varepsilon$-dependent, semiclassical, representation of
the canonical commutation relations as in \eqref{eq: ccr}. Therefore, the microscopic theory is set in the Hilbert space
$\mathscr{H}\otimes \mathscr{K}_{\varepsilon}$. 

We restrict our attention to Fock representations of the canonical commutation relations. Therefore, we assume that
\begin{equation*}
  \mathscr{K}_{\varepsilon}=  \mathcal{G}_{\eps}(\mathfrak{h}) = \bigoplus_{n = 0}^{\infty} \mathfrak{h}^{\otimes_{\mathrm{s}} n}\; ,
\end{equation*}
the symmetric Fock space constructed over a separable Hilbert space $\mathfrak{h}$. The space
$\mathfrak{h}$ is the space of classical fields\footnote{Strictly speaking, the space $ \mathfrak{h} $ should be the Hilbert completion of the set of test functions for the classical fields, but in the following we are going to restrict our attention to classical fields belonging to such space.}. The
canonical commutation relation \eqref{eq: ccr} in $ \mathscr{K}_{\eps} $ reads, for any $ z, w \in \mathfrak{h} $,
\begin{equation*}
  \lf[a_{\varepsilon}(z),a^{\dagger}_{\varepsilon}(w) \ri] = \varepsilon \braket{z}{w}_{\mathfrak{h}}\; ,
\end{equation*}
and the quasi-classical limit corresponds to the limit $\varepsilon\to 0$. 

According to the notation set above, a microscopic Fock-normal state
is thus described by a density matrix 
\beq
	\Gamma_{\varepsilon}\in \mathscr{L}^1_{+,1}(\mathscr{H}\otimes \mathscr{K}_{\varepsilon}).
\eeq
The dynamics is generated by a self-adjoint and bounded from below Hamiltonian on $\mathscr{H}\otimes
\mathscr{K}_{\varepsilon}$, that we denote by $H_{\varepsilon}$. Given the unitary dynamics $e^{-itH_{\varepsilon}}$, the evolved state is
\begin{equation}
  \Gamma_{\varepsilon}(t) := e^{-itH_{\varepsilon}} \Gamma_{\varepsilon} e^{itH_{\varepsilon}}\; .
\end{equation}

Let us now turn the attention to the effective quasi-classical system in the limit $\varepsilon\to 0$. This is an hybrid quantum-classical system, in which the classical part acts as an environment for the quantum part. In fact, as we will see, the classical field affects the quantum particles, but the converse is not true, the interaction is not strong enough to cause a back-reaction of the particles on the classical field.

The basic observables for the classical fields are the elements $z\in \mathfrak{h}$, or, more precisely, the real vectors of the form $z + \bar{z} $.  Scalar observables in a generalized sense are functions $z\mapsto f(z)\in \mathbb{C}$ semiclassically called {\it symbols}. In addition to scalar or field observables, there are more general observables involving both subsystems, which are thus represented by operator-valued functions $ z\mapsto \mathcal{F}(z) $, where $ \mathcal{F}(z) $ is a linear operator on the particle Hilbert space $ \mathscr{H}$. Note that one can easily associate an operator-valued function to a scalar symbol as well, by simply setting $ \mathcal{F}(z) = f(z) \cdot \mathds{1} $, where $ \mathds{1} \in \mathscr{B}(\mathscr{H}) $ stands here for the identity operator.

A state of the classical field (environment) is a Borel probability measure $\mu\in \mathscr{P}(\mathfrak{h}) $, while a state of the quantum
particles is a density matrix $ \gamma \in \mathscr{L}^1_{+,1}(\mathscr{H})$. Since in the quasi-classical regime
the environment affects the behavior of the quantum particle system, a quasi-classical state is a
\emph{state-valued probability measure} 
\beq 
	\mathfrak{m} \in	\mathscr{P}\lf(\mathfrak{h}; \mathscr{L}^1_{+}(\mathscr{H}) \ri).  
\eeq 
A state-valued measure thus takes values in $ \mathscr{L}^1_+(\mathscr{H})$, but it can also be conveniently described by its norm Radon-Nikod\'{y}m decomposition (see \cref{prop:7}): a pair $(\mu_{\mm},\gamma_{\mm}(z))$ consisting of a scalar Borel (probability) measure $\mu_{\mm} $, and a $\mu_{\mm}$-integrable, almost everywhere defined function $ \gamma_{\mm}(z) \in \mathscr{L}^1_{+,1}(\mathscr{H})$ taking values in normalized density matrices, {\it i.e.}, 
\beq 
	\diff \mm(z) = \gamma_{\mm}(z) \diff \mu_{\mm}(z).  
\eeq 
In other words, a generic normalized quasi-classical state consists of a measure $\mu_\mm $ describing the environment, and a function $ \gamma_\mm(z)$ describing how (almost) each configuration of the field affects the quantum particles' state. The quasi-classical equivalent of taking the partial trace w.r.t. the field's degrees of freedom is integrating w.r.t.\ the quasi-classical state-valued measure, {\it i.e.,} for any operator-valued function $ \mathcal{F}(z) \in \mathscr{B}(\mathscr{H}) $, \beq \int_{\mathfrak{h}}^{} \mathrm{d}\mm(z) \: \mathcal{F}(z)= \int_{\mathfrak{h}}^{} \mathrm{d}\mu_{\mm}(z) \:\gamma_{\mm}(z) \mathcal{F}(z) .  \eeq
Note that, when integrating against the state-valued measure, it is a priori relevant to keep the order as in the above expression, since $ \mathcal{F}(z) $ might not commute with $ \gamma_{\mm}(z) $.

The quasi-classical evolution also consists of two parts: an evolution of the environment's probability
measure $ \mu_{\mm} $, and one of the quantum system for (almost) every configuration of the classical
field. The evolution of the environment depends on the choice of a scaling parameter for the field's part
in $H_{\varepsilon}$, and we consider two cases: either the environment is stationary, {\it e.g.}, it is at equilibrium,
or it evolves freely. Concretely, the environment is evolved by a unitary, linear, flow
$e^{-it\nu\omega}:\mathfrak{h}\to \mathfrak{h}$, $t\in \mathbb{R}$, of classical fields, where $\omega$ is a positive self-adjoint
operator on $\mathfrak{h}$ (typically, a multiplication operator by the dispersion relation of the field), and $ \nu\in \{0,1\}$, depending on the chosen scaling. This flow pushes forward
the measure $\mu_{\mm} $, yielding
\begin{equation}
  \mu_{\mm, t} : = \lf( e^{-it \nu\omega} \ri) _{\star}  \mu_{\mm}\; .
\end{equation}
The explicit action of the pushforward, as it is well-known, is as follows: for all measurable Borel sets $B \subset
\mathfrak{h}$,
\begin{equation}
  	\lf[ \lf( e^{-it\nu\omega} \ri)_{\star} \mu \ri] (B) : = \mu\lf( e^{it\nu\omega} B \ri)\; ,
\end{equation}
where $e^{it\nu\omega}B$ stands for the preimage of $B$ w.r.t. the map $e^{-it\nu\omega}$. The quantum part of the
evolution is generated by a map from field configurations to two-parameter groups of unitary operators $z\mapsto
{\bigl(\mathcal{U}_{t,s}(z)\bigr)}_{t,s\in \mathbb{R}}$, and it acts as
\begin{equation}
  	\gamma_{\mm,t,s}(z) := \mathcal{U}_{t,s}(z) \gamma_{\mm}(z) \mathcal{U}^{\dagger}_{t,s}(z)\;.
\end{equation}
Let us remark that the pushforward of the measure \emph{does not affect} the Radon-Nikod\'{y}m derivative
$ \gamma_{\mm,t,s}(z)$, but only the integrated functions.

The quantum evolution is unitary for (almost) all configurations of the field. However, a measurement on
the classical system modifies the quantum state in a non-unitary, but explicit, way. Let $f(z)$ be a
scalar field's observable and suppose it is $\mu_{\mm}$-measurable. For $\lambda\in \mathbb{C}$, let us define the
level set of $f$ as
\begin{equation*}
  B_{\lambda}=\{z\in \mathfrak{h}, f(z)=\lambda\}\;.
\end{equation*}
Then the conditional quantum state $\gamma_{\mm,t,s}\bigr\rvert_{f=\lambda}\in
\mathscr{L}^1_{+,1}(\mathscr{H})$ at time $t\in \mathbb{R}$, describing the state of the quantum system conditioned to
an observed value $\lambda$ of the classical observable $f$, is given by
\begin{equation*}
  \gamma_{\mm,t,s}\bigr\rvert_{f=\lambda}=\int_{\mathfrak{h}}^{} \mathrm{d}\mu_{\mm,\lambda,t-s}(z) \: \gamma_{\mm,t}(z)\, \mathds{1}_{B_{\lambda}}(z)  = \int_{e^{i(t-s)\nu\omega}B_{\lambda}}^{} \mathrm{d}\mu_{\mm,\lambda}(z) \: \mathcal{U}_{t,s}(z) \gamma_{\mm}(z) \mathcal{U}^{\dagger}_{t,s}(z) \;,
\end{equation*}
where $(\mu_{\mm,\lambda})_{\lambda\in \mathbb{C}}$ is the disintegration of $\mu_{\mm}$ w.r.t. the function $f$. The
conditional evolution $(t,s) \mapsto  \gamma_{\mm,t,s}\bigr\rvert_{f=\lambda} $ is clearly non-unitary but it preserves positivity: the dynamics is in general non-Markovian, unless either $ B_{\lambda}=\{z_{\lambda}\}$ or $\mu_{\mm}=\delta_{z_0}$, {\it i.e.}, the group property might not be satisfied. One should indeed not expect that, for any $ t, s, \tau > 0 $, there exists some two-parameter unitary group $ \mathcal{W}_{t,s} \in \mathscr{B}(\mathscr{H}) $, such that
\bdm
	 \gamma_{\mm,t,s}\bigr\rvert_{f=\lambda} = \mathcal{W}_{t,\tau}  \gamma_{\mm,\tau,s}\bigr\rvert_{f=\lambda} \mathcal{W}_{t,\tau}^{\dagger}.
\edm

The quantum state at time $t\in \mathbb{R}$, conditioned to the fact that $f$ (or any other observable) is observed, irrespective of its value, is denoted by $\gamma_{\mm,t,s}$, it is independent of $f$, and it is given by
\begin{equation*}
  \gamma_{\mm,t,s}=\int_{\mathbb{C}}^{}\mathrm{d}\lambda  \int_{e^{i(t-s)\nu\omega}B_{\lambda}}^{}  \mathrm{d}\mu_{\mm,\lambda}(z) \: \mathcal{U}_{t,s}(z)\gamma_{\mm}(z)\mathcal{U}_{t,s}^{\dagger}(z) =\int_{\mathfrak{h}}^{}\mathrm{d}\mu_{\mm}(z) \: \mathcal{U}_{t,s}(z)\gamma_{\mm}(z)\mathcal{U}_{t,s}^{\dagger}(z)  \;.
\end{equation*}
Again the conditional evolution $(t,s) \mapsto\gamma_{\mm,t,s} $ preserves both positivity and the trace, but it  is still non-Markovian in general. It would be interesting to study the states of the environment, if any, not concentrated in a single field configuration, that make the conditional evolution Markovian, and possibly non-unitary. Such measures would yield a quasi-classical evolution on the open quantum system of Lindblad type (see, \emph{e.g.},\cite{kossakowski1972rmp,lindblad1976cmp}).

\subsection{The Concrete Models: Nelson, Pauli-Fierz, and Polaron}
\label{sec:concr-models:-nels}

Let us define more concretely the three models of interaction between non-relativistic particles and
bosonic force carrier fields that we consider throughout the paper: the Nelson, Pauli-Fierz, and polaron models.

\subsubsection{Nelson Model}
\label{sec:nelson-model-1}

The Nelson model describes quantum particles (\emph{e.g.}, nucleons), interacting with a force-carrying
scalar field (\emph{e.g.}, a meson field), and was firstly rigorously studied \cite{nelson1964jmp}. In this paper, we restrict our attention to the regularized Nelson model, where the
interaction is smeared by an ultraviolet cutoff. We consider $N$, $d$-dimensional, non-relativistic,
spinless particles, and therefore $\mathscr{H}=L^2 (\mathds{R}^{dN} )$. The classical fields are usually taken to be
in $\mathfrak{h}=L^2(\mathbb{R}^d)$, but other choices may be possible, {\it e.g.}, a cavity
field, whose classical space would then be $\ell^2(\mathbb{Z}^d)$. The Hamiltonian $H_{\varepsilon}$ has the form
\begin{equation*}
  H_{\varepsilon}= K_0+ \nu(\varepsilon)\mathrm{d}\GG_{\varepsilon}(\omega) + \sum_{j=1}^N \lf[ a^{\dagger}_{\varepsilon}\bigl(\lambda(\xv_j)\bigr)+a_{\varepsilon}\bigl(\lambda(\xv_j)\bigr) \ri]\; ,
\end{equation*}
where $K_0=\mathcal{K}_0\otimes 1$, with $\mathcal{K}_0$ self-adjoint and bounded from below on $\mathscr{H}$,
$\omega$ a positive operator on $\mathfrak{h}$ and $\mathrm{d}\GG_{\varepsilon}(\omega)$ its second
quantization, \emph{i.e.}, the Wick quantization of the symbol 
\beq
	\label{eq: kappa}
	\kappa(z) : =  \meanlrlr{z}{\omega}{z}_{\mathfrak{h}},
\eeq
and $\lambda\in L^{\infty}(\mathbb{R}^d; \mathfrak{h})$ is the coupling factor. 

If one naively replaces the quantum canonical variables $ a^{\#} $ with their classical counterparts, {\it i.e.}, $ z^{\#} $, one can easily deduce that the quasi-classical effective potential for the model above is given by the symbol $z\mapsto
\mathcal{V}(z)$, where (see also \cite[Sect. 2.2]{correggi2017ahp})
\begin{equation}
	\label{eq: nelson effective}
  \mathcal{V}(z)=\sum_{j=1}^N 2\Re \braket{z}{\lambda(\xv_j)}_{\mathfrak{h}}\in \mathscr{B}(\mathscr{H})\; .
\end{equation}
In most practical applications $ \lambda(\xv; \: \cdot \:) \in \mathfrak{h}=L^2 (\mathbb{R}^d)$ has the following explicit form:
\begin{equation*}
  \lambda(\xv; \kv) =\lambda_0(\kv) e^{-i \kv \cdot \xv}\;, \qquad		\lambda_0 \in L^2(\mathbb{R}^d)\;.
\end{equation*}
This leads to the effective potential $\mathcal{V}(z)$ being the Fourier transform of an integrable
  function, and thus continuous and vanishing at infinity. In order to obtain more singular potentials, it
  is necessary to consider microscopic states whose measures are not concentrated as Radon measures in
  $\mathfrak{h}$ \cite[Sect. 2.5]{correggi2017ahp}. This would, however, make the analysis more involved. We thus restrict our
  attention to states whose measures are indeed concentrated in $\mathfrak{h}$  (see \cref{rem:3} for
  additional details).

\subsubsection{Pauli-Fierz Model}
\label{sec:pauli-fierz-model-1}

We consider the class of Pauli-Fierz models describing $N$ non-relativistic, spinless,
extended $d$-dimensional charges moving in $\mathbb{R}^d$, $d\geq 2$, interacting with electromagnetic radiation in the Coulomb
gauge. Adding spin, adopting a different gauge, or constraining particles to an open subset
of $\mathbb{R}^d$ would not affect the results, but make the analysis more involved. The particles' Hilbert space
is thus $\mathscr{H}=L^2(\mathbb{R}^{dN})$, while the classical fields are in $\mathfrak{h}=L^2 (\mathbb{R}^d; \mathbb{C}^{d-1}
)$. The Hamiltonian $H_{\varepsilon}$ is customarily written as
\begin{equation*}
  H_{\varepsilon}=\sum_{j=1}^N{\bigl(-i\nabla_j + \aav_{\varepsilon}(\xv_j)\bigr)}^2 + W(\xv_1,\dotsc, \xv_N) + \nu(\varepsilon)\mathrm{d}\GG_{\varepsilon}(\omega)\; ,
\end{equation*}
with
\begin{equation*}
  \aav_{\varepsilon}(\xv)= a^{\dagger}_{\varepsilon}\bigl(\vec{\lambda}(\xv)\bigr) + a_{\varepsilon}\bigl(\vec{\lambda}(\xv)\bigl).
\end{equation*}
In the above: $W=W_1+W_2$ is a multiplicative potential describing the interaction among charges, with
$\mathcal{W}_1\in L^1_{\mathrm{loc}}(\mathbb{R}^{dN}; \mathbb{R}^+)$ and $\mathcal{W}_2{(-\Delta+1)}^{-\frac{1}{2}}\in
\mathscr{B}(\mathscr{H})$; $\omega$ is the field's dispersion relation, a positive multiplication operator on
$\mathfrak{h}$, such that $\omega^{-1}$ is also a positive self-adjoint operator on $\mathfrak{h}$, {\it e.g.}, $ \omega(\kv) = \lf| \kv \ri| $, and
$\vec{\lambda}=(\lambda_1,\dotsc,\lambda_d)$, with $\lambda_{\ell}\in L^{\infty}(\mathbb{R}^d; \mathscr{D}(\omega^{-1/2}+\omega^{1/2})) $ for all $ \ell \in
\{1,\dotsc,d\}$ and $\nabla\cdot \vec{\lambda}(x)=0$, is the particles' charge distribution. We denoted by $ \mathscr{D}(\omega^{-1/2}+\omega^{1/2}) \subset \mathfrak{h} $ the intersection of the self-adjointness domains of $ \omega^{-1/2} $ and $ \omega^{1/2} $. 

In this case, we have
$\mathcal{K}_0= -\Delta +\mathcal{W}$ and the effective potential can be easily seen to become \cite[Sect. 1.2]{correggi2017arxiv}
\begin{equation}
  \mathcal{V}(z)= 4 \sum_{j=1}^N \lf[ -  i \Re \braket{z}{\vec{\lambda}(\xv_j)}_{\mathfrak{h}}  \cdot \nabla_j + \lf( \Re \braket{z}{\vec{\lambda}(\xv_j)}_{\mathfrak{h}} \ri)^2 \ri] \;.
\end{equation}
Notice that the interaction term in $ H_{\eps} $ is not the Wick quantization of the above symbol $ \mathcal{V}(z) $, because $ H_{\eps} $ is not normal ordered and an additional term is missing, {\it i.e.},
\begin{equation*}
  \varepsilon\sum_{j=1}^N \sum_{\ell =1}^d \lVert \lambda_{\ell}(\xv_j)  \rVert_{\mathfrak{h}}^2\; = \OO(\eps),
\end{equation*}
but such a contribution vanishes in the limit $ \eps \to 0 $. Similarly to the Nelson model, the effective interaction $\mathcal{V}(z)$ describes the minimal
  coupling of the particles with a magnetic potential that is continuous and vanishing at infinity.

\subsubsection{Polaron}
\label{sec:polaron-model-1}

The Fr\"{o}hlich's polaron \cite{frohlich1937prslA} describes electrons moving in a quantum lattice crystal. The
$N$ $d$-dimensional electrons are modeled as non-relativistic spinless particles, and thus
again $\mathscr{H}=L^2 (\mathds{R}^{dN})$. For the phonon vibrational field, $\mathfrak{h}=L^2 (\mathds{R}^d )$. The
Hamiltonian $H_{\varepsilon}$ is formally written as
\begin{equation*}
  H_{\varepsilon}= -\Delta + a_{\varepsilon}^{\dagger}\bigl(\phi(\xv_j)\bigr)+a_{\varepsilon}\bigl(\phi(\xv_j)\bigr)+ W(\xv_1,\dotsc,\xv_N)+ \nu(\eps) \mathrm{d}\GG_{\varepsilon}(1) \; ,
\end{equation*}
with the particles' potential $W$ satisfying the same assumptions given in~\cref{sec:pauli-fierz-model-1}
for the Pauli-Fierz model. In addition, 
\begin{equation*}
  \phi(\xv; \kv) :=\alpha \frac{e^{-i \kv\cdot \xv}}{\lvert \kv  \rvert_{}^{\frac{d-1}{2}}}\; ,
\end{equation*}
$\alpha\in \mathbb{R}$, is the polaron's form factor and, for all $ \xv\in \mathbb{R}^d$, it does not belong to $\mathfrak{h}$. Hence, $H_{\varepsilon}$ as written above is only a formal expression. However, it makes sense as a closed
and bounded from below quadratic form: one can find a parameter $ r \in \R^+ $, a splitting $ \phi = \phi_{r} + \chi_{r} $, with
\begin{equation*}
  \phi_r(\xv; \kv) : = \one_{\lf\{ |\kv| \leq r \ri\}}(\kv) \phi(\xv; \kv),
\end{equation*}
and some $\vec{\lambda}_r \in
L^{\infty}(\mathbb{R}^d; \mathfrak{h}^d)$, such that, as a quadratic form,
\bmln{
      H_{\varepsilon}= -\Delta + a_{\varepsilon}^{\dagger}\bigl(\phi_r(\xv_j)\bigr) + a_{\varepsilon}\bigl(\phi_r(\xv_j)\bigr) + \lf[-i\nabla_j, a_{\varepsilon}\bigl(\vec{\lambda}_r(\xv_j)\bigr) - a^{\dagger}_{\varepsilon}\bigl(\vec{\lambda}_r(\xv_j)\bigr) \ri] \\+ W(\xv_1,\dotsc,\xv_N)+ \nu(\eps) \mathrm{d}\GG_{\varepsilon}(1)\;,
}
where the commutator between two vectors of operators involves a scalar product.

In the polaron model $\mathcal{K}_0= -\Delta +\mathcal{W}$, and the effective potential is given by \cite[Sect. 2.3]{correggi2017ahp}
\begin{equation*}
  \mathcal{V}(z)=2\sum_{j=1}^N\Re \braket{z}{\phi_r(\xv_j)}_{\mathfrak{h}} + \lf[-i\nabla_j, \Im \braket{\vec{\lambda}_r(\xv_j)}{z}_{\mathfrak{h}} \ri]\;.
\end{equation*}
Notice that one could formally resum the two terms above, obtaining the same expression \eqref{eq: nelson effective} as in the Nelson model. In the case of the polaron, the potential $\mathcal{V}(z)$ is not necessarily bounded,
  but still relatively form bounded w.r.t. $-\Delta$. In fact, $\mathcal{V}(z) $ can be any function in
  $\dot{H}^{\frac{d-1}{2}}(\mathbb{R}^d)\cap L^2_{\mathrm{loc}}(\mathbb{R}^d)$.

  Let us also remark that in the polaron case, the quasi-classical limit is mathematically analogous to
  the strong coupling limit. Strongly coupled polarons have been widely studied in the mathematical
  literature both from a dynamical and a variational point of view (see, \emph{e.g.}, \cite{gross1976ap,
    lieb1997cmp, griesemer2013jpa, frank2014lmp, frank2017apde, griesemer2016arxiv2, frank2019arxiv,
    lieb2019arxiv,leopold2019arxiv,frank2019arxiv2}). Compared to the available dynamical results
  \cite{frank2014lmp,frank2017apde,griesemer2016arxiv2, leopold2019arxiv}, our quasi-classical approach has the advantage
  of being applicable to a very general class of microscopic initial states. However, we have no control on the errors and we are not able to
  derive the higher order corrections to the effective dynamics, \emph{i.e.}, the ones given by the
  Landau-Pekar equations.

\subsection{Main Results}
\label{sec:main-results}

Before stating our main results, we provide more technical details about the general structure of the models we are considering in this
paper, by specifying some assumptions that are sufficient to prove our main results, and that are
satisfied in the above concrete models. We do not strive for the optimal assumptions nor for the most general setting.

First of all, we remark that all the Hamiltonians introduced in \cref{sec:concr-models:-nels} can be cast in the following form
\begin{equation}
  \label{eq:35}
  H_{\varepsilon}=\mathcal{K}_0\otimes 1 + \nu(\varepsilon)\, 1\otimes \mathrm{Op}_{\varepsilon}^{\mathrm{Wick}}(\kappa) + \mathrm{Op}_{\varepsilon}^{\mathrm{Wick}}(\mathcal{V}) + \OO(\varepsilon)\; ,
\end{equation}
where: $\mathcal{K}_0$ is self-adjoint and bounded from below on $\mathscr{H}$, and describes the
particle's system when it is isolated; $\nu(\varepsilon)$ is a quasi-classical scaling factor, such that 
\beq
	\nu=\lim_{\varepsilon\to
  0}\varepsilon\nu(\varepsilon) \in \{0, 1\},
\eeq
and the two relevant scalings are $\nu(\varepsilon)=1$, yielding an environment that remains constant in
time, and $\nu(\varepsilon)=\frac{1}{\varepsilon}$, yielding an environment that evolves freely; $\kappa$ is the symbol given by \eqref{eq: kappa} for a densely defined,
positive operator $ \omega  $ on $\mathfrak{h}$. Given a symbol $ z \mapsto \mathcal{F}(z) $, we denote by $ \mathrm{Op}_{\varepsilon}^{\mathrm{Wick}}(\mathcal{F})$ its Wick quantization, so that in particular $ \mathrm{Op}_{\varepsilon}^{\mathrm{Wick}}(\kappa) = \diff \GG_{\eps}(\omega) $. The symbol $z \mapsto \mathcal{V}(z) $ is operator-valued and polynomial, and it describes the
interaction between the particles and the environment. The possible concrete choices of $ \mathcal{V} $ have been presented in \cref{sec:concr-models:-nels}. Finally, $ \OO(\varepsilon) $ is a bounded particle
operator of order $\varepsilon$.

To study the limit $\varepsilon\to 0$ of evolved states $ \Gamma_{\varepsilon}(t)$, we make the following very
general assumption on\footnote{For simplicity, we set the initial time $ s$ equal to 0.} $\Gamma_{\varepsilon}(0)=\Gamma_{\varepsilon}$:
\begin{equation}
  \label{eq:20}\tag{$\mathrm{A1}$}
  \exists \delta>0 , \exists C_{\delta} < + \infty \quad \text{s.t.} \quad \Tr\bigl(\Gamma_{\varepsilon} {(\mathrm{d}\GG_{\varepsilon}(1)+1)}^{\delta}\bigr) \leq C_{\delta}\;,
\end{equation}
which is for instance satisfied if the state scales with $ \eps $ as in \eqref{eq: number of carriers},
  or if it is formed by a coherent superposition of vectors with a finite number of force carriers.  Such assumption is sufficient to prove the existence of a subsequence $ \lf\{ \varepsilon_n
\ri\}_{n \in \N} \to 0$ such that $\Gamma_{\varepsilon_n}$ converges to a quasi-classical state $\mm$ in the sense of the
\cref{def: qc convergence} below. For the polaron and Pauli-Fierz models, an additional assumption is
necessary to study the limit $\varepsilon\to 0$ of $\Gamma_{\varepsilon}(t)$, due to the fact that such models are ``more singular''
than the Nelson model:
\begin{equation}
  \label{eq:13}\tag{$\mathrm{A1'}$}
  \exists C < +\infty \quad \text{s.t.} \quad\Tr\bigl(\Gamma_{\varepsilon} \bigl(K_0+ \mathrm{d}\mathcal{G}_{\varepsilon}(\omega)\bigr)\bigr) \leq C\;.
\end{equation}
Finally, in order to ensure that no loss of mass occurs along the weak limit, or, equivalently, that the quasi-classical limit point $ \mm $ is still normalized and $ \lf\| \mm(\mathfrak{h}) \ri\|_{\mathscr{L}^1(\mathscr{H})} = 1 $, we also need a control of the particle component of the state $ \Gamma_{\eps} $. We thus define the {\it reduced density matrix} for the particles as
\beq
	\label{eq: reduced}
	\gamma_\eps : = \tr_{\mathscr{K}_{\eps}} \Gamma_{\eps} \in \mathscr{L}^1_{+,1}(\mathscr{H}),
\eeq
and impose the following alternative conditions on $ \gamma_{\eps} $: 
\beq
	\label{eq: ass1}
	\tag{$\mathrm{A2}$}
	\exists \mathcal{A} > 0, \mathcal{A}^{-1} \in \mathscr{L}^{\infty}(\mathscr{H}), \mbox{ s.t. } \tr_{\mathscr{H}}\lf(
          \mathcal{A} \gamma_{\eps} \ri) \leq C < + \infty\; ,
\eeq
or
\beq
	\label{eq: ass2}
	\tag{$\mathrm{A2'}$} 
	\exists \gamma\in \mathscr{L}^1_{+}(\mathscr{H}), \mbox{ s.t. } \gamma_{\eps} \leq \gamma \;.  
\eeq 
We are going to comment further about the above conditions in \cref{rem:
  role} and \cref{rem: assumptions}, but we point out here that the
  second is stronger than the first, in the sense that \eqref{eq: ass2}
  implies \eqref{eq: ass1}. A simple but relevant case in which \eqref{eq: ass2} is
          trivially satisfied is given by product states of the form $ \gamma \otimes \varsigma_{\eps} $ with $ \gamma \in
          \mathscr{L}^{1}_{+,1}(\mathscr{H}) $ independent of $ \eps
          $. Contrarily, the more general Assumption
          \eqref{eq: ass1} seems at a first glance also more arbitrary, but it could be put in relation with
          the physics of the model (see \cref{rem: assumptions}).

Let us define by $ \widehat{\Gamma}_{\varepsilon} $ the noncommutative Fourier transform or generating map of a state $ \Gamma_{\eps} \in \mathscr{L}^1_{+,1}(\mathscr{H} \otimes \mathscr{K}_{\eps}) $, {\it i.e.},
\begin{equation}
	\label{eq: nc fourier transform}
    \widehat{\Gamma}_{\varepsilon}(\eta):=\tr_{\mathscr{K}_{\varepsilon}}\bigl(\Gamma_{\varepsilon} W_{\varepsilon}(\eta) \bigr)\in \mathscr{L}^1(\mathscr{H})\; 
\end{equation}
for any $\eta\in
\mathfrak{h}$, where $W_{\varepsilon}(\eta)$ is the Weyl operator on $\mathscr{K}_{\varepsilon}$:
\begin{equation}
	\label{eq: weyl}
 	 W_{\varepsilon}(\eta) :=e^{i(a^{\dagger}_{\varepsilon}(\eta)+a_{\varepsilon}(\eta))}\; .
\end{equation}
Analogously, to any state-valued measure $\mm\in \mathscr{M}(\mathfrak{h};  \mathscr{L}^1_{+}(\mathscr{H}))$ there corresponds the Fourier transform
\begin{equation}
	\label{eq: fourier transform}
   \widehat{\mm}(\eta):=\int_{\mathfrak{h}}^{} \mathrm{d}\mu_{\mm}(z) \: \gamma_{\mm}(z) e^{2i\Re \braket{ \eta}{z}_{\mathfrak{h}}} \in \mathscr{L}^1(\mathscr{H})\;,
\end{equation}
for $ \eta \in \mathfrak{h} $.

 \begin{definition}[Quasi-classical convergence]
 	\label{def: qc convergence}
 	\mbox{}	\\
 	Let $ \Gamma_\eps \in \mathscr{L}^1_{+,1}(\mathscr{H} \otimes \mathscr{K}_{\eps}) $ and $ \mm \in \mathscr{M}\lf(\mathfrak{h};  \mathscr{L}^1_{+}(\mathscr{H})\ri) $. We say that
 	\beq
 		\Gamma_{\eps} \xrightarrow[\eps \to 0]{\mathrm{qc}} \mm,
	\eeq
	if and only if $ \widehat{\Gamma}_{\varepsilon}(\eta) \xrightarrow[\eps \to 0]{\mathrm{w}*} \widehat{\mm}(\eta)$ pointwise for
all $\eta\in \mathfrak{h}$ in weak-$*$  topology in $ \mathscr{L}^1(\mathscr{H}) $, {\it i.e.}, when testing against compact operators in $ \mathscr{L}^{\infty}(\mathscr{H}) $. 
 \end{definition}
  
 The above \cref{def: qc convergence} is given in terms of the Fourier transforms in order to completely
 characterize the limit quasi-classical measure $ \mm $. On the other hand, from the physical point of
 view, it is relevant to study the convergence of expectation values of quantum observables, which is
 discussed in \cref{sec:quasi-class-analys} and specifically in \cref{thm:8}. Note that in light of
 \cref{prop:8}, assumption \eqref{eq:20} guarantees that any such $ \Gamma_{\eps} $ admits at least one limit
 point in the sense of \cref{def: qc convergence}. 
 
In the following, we may omit the superscript $^{\mathrm{qc}}$ in $\Gamma_{\varepsilon} \xrightarrow[]{\mathrm{qc}} \mathfrak{m}$, if it is clear from the
context that we are considering the quasi-classical convergence of \cref{def: qc convergence} (not
its stronger counterpart of next \cref{def:11}).
        
	\begin{remark}[Reduced density matrix]
			\mbox{}	\\
			\label{rem: reduced}
			We point out that the reduced density matrix $ \gamma_{\eps} $ for the particle
                          system given in \eqref{eq: reduced} can be obtained by evaluating the noncommutative Fourier transform \eqref{eq: nc fourier transform} in $ \eta = 0 $, {\it i.e.},
			\bdm
				\gamma_{\eps} : = \tr_{\mathscr{K}_{\eps}} \Gamma_{\eps} = \widehat{\Gamma}_{\eps}(0).
			\edm
			Hence, the convergence $ \Gamma_{\eps_k} \xrightarrow[k \to + \infty]{} \mm $ can be easily seen to imply that
			\beq
				\gamma_{\eps_k} \xrightarrow[k \to + \infty]{\mathrm{w}*} \int_{\mathfrak{h}} \diff \mu_{\mm}(z) \: \gamma_{\mm}(z),
			\eeq
			where we have denoted by $ \mathrm{w}* $ the {\it weak-$*$ operator} topology.
		\end{remark}

  	\begin{remark}[Product states]
  		\label{rem: product states}
  		\mbox{}	\\
  		As a special case, we observe that, if $ \Gamma_{\eps} $ is a \emph{physical product
                    state}\footnote{Product states are the mathematical formulation of the fact that
                    the two parts of the system are independent. Since $\varepsilon$ characterizes only the behavior of
                    the field, it is not physically relevant to put an
                    $\varepsilon$-dependence on the particle part.}, {\it i.e.}, if there exist $ \gamma \in
                \mathscr{L}^1_{+,1}(\mathscr{H})$ and $\varsigma_{\varepsilon}\in \mathscr{L}^1_{+,1}(\mathscr{K}_{\varepsilon})$, such
                that $ \Gamma_{\eps} = \gamma \otimes \varsigma_{\varepsilon}$, then 
                \beq 
                	\Gamma_{\eps} \xrightarrow[\eps \to 0]{\mathrm{qc}} \mm \quad \Longleftrightarrow \quad \varsigma_{\varepsilon}
                \xrightarrow[\eps \to 0]{} \mu_{\mm} \mbox { and } \gamma_{\mm}(z) = \gamma. 
                \eeq 
                The proper definition
                of the convergence between scalar measures was given in \cite{correggi2017ahp}, but it
                coincides with \cref{def: qc convergence} when $ \mathscr{H} = \mathbb{C} $.
	\end{remark}

A stronger notion of quasi-classical convergence can be given lifting the weak-$*$ convergence of the Fourier transform in \cref{def: qc convergence} to weak convergence, \emph{i.e.}, by testing with bounded operators $\mathcal{B}\in \mathscr{B}(\mathscr{H})$. This leads to the following definition.

        \begin{definition}[Bounded quasi-classical convergence]
          	\label{def:11}
          	\mbox{}	\\
 		Let $ \Gamma_\eps \in \mathscr{L}^1_{+,1}(\mathscr{H} \otimes \mathscr{K}_{\eps}) $ and $ \mm \in \mathscr{M}\lf(\mathfrak{h};  \mathscr{L}^1_{+}(\mathscr{H})\ri) $. Then
 		\beq
 			\Gamma_{\eps} \xrightarrow[\eps \to 0]{\mathrm{bqc}} \mm,
		\eeq
		if and only if $ \widehat{\Gamma}_{\varepsilon}(\eta) \xrightarrow[\eps \to 0]{\mathrm{w}} \widehat{\mm}(\eta)$ pointwise for all $\eta\in \mathfrak{h}$ in weak  topology in $ \mathscr{L}^1(\mathscr{H}) $, {\it i.e.}, when testing against bounded operators in $ \mathscr{B}(\mathscr{H}) $.
	\end{definition}
	
	\begin{remark}[Mass conservation]
		\label{rem: mass conservation}
		\mbox{}	\\
		Bounded convergence ensures that no mass is lost in the quasi classical limit: in fact, if
  $\mathscr{L}^1_{+,1}(\mathscr{H} \otimes \mathscr{K}_{\eps}) \ni \Gamma_{\varepsilon} \xrightarrow[]{\mathrm{bqc}} \mathfrak{m}$, then
		\begin{equation}
  			\tr_{\mathscr{H}} \int_{\mathfrak{h}}^{}  \mathrm{d}\mathfrak{m}(z)=\tr_{\mathscr{H}}\Bigl(\hat{\mathfrak{m}}(0)\Bigr)=\lim_{\varepsilon\to 0}\tr_{\mathscr{H}}\Bigl(\hat{\Gamma}_{\varepsilon}(0)\Bigr)=1\; ,
		\end{equation}
		and therefore $\mathfrak{m}\in \mathscr{P}\lf(\mathfrak{h}; \mathscr{L}^1_{+}(\mathscr{H})\ri) $. On the contrary, there are normalized states $\Gamma_{\varepsilon} \in \mathscr{L}^1_{+,1}(\mathscr{H}\otimes \mathscr{K}_{\varepsilon})$ that converge to a quasi-classical measure with mass less than one, and possibly zero, thus occurring in a loss of mass phenomenon. For example, let us consider a state similar to the one defined in \cref{rem: product states}, where however $\gamma=\gamma_{\varepsilon}$ also depends on $\varepsilon$, with $\gamma_{\varepsilon}=  \ket{e_{n(\varepsilon)}} \bra{e_{n(\varepsilon)}} $, $ \lf\{ e_n \ri\}_{n\in \mathbb{N}}$ being an orthonormal basis in $\mathscr{H}$, and $n(\varepsilon)\underset{\varepsilon\to 0}{\longrightarrow} + \infty$. Then,$\gamma_{\varepsilon}\otimes \varsigma_{\varepsilon}\to 0 $, the zero state-valued measure, and thus all the mass is lost in the limit.
	\end{remark}
     
Our main result (see \cref{thm:5} and \cref{cor:7} below) is that initial convergence is propagated
  in time: for all $t\in \mathbb{R}$, $\Gamma_{\varepsilon_n}(t)$ converges to the quasi-classical state $\mm_t$ defined by the norm Radon-Nikod\'{y}m decomposition
\begin{equation}
	\label{eq: measure mm}
  	\boxed{\diff \mm_t = \mathcal{U}_{t,0}(z) \gamma_{\mm}(z) \mathcal{U}^{\dagger}_{t,0}(z) \: \diff \lf( e^{-it\nu\omega}\,_{\star} \, \mu_{\mm} \ri) \;,}
\end{equation}
where $\mathcal{U}_{t,s}(z)$ is the above mentioned quasi-classical two-parameter unitary group of
evolution, that turns out to be weakly generated by the time-dependent Schr\"{o}dinger operator
\begin{equation}
	\label{eq: keff}
  	\KK_t : =  \mathcal{K}_0+\mathcal{V}_{t}(z)\; ,
\end{equation}
with
\begin{equation}
  \mathcal{V}_{t}(z) := \mathcal{V}(e^{-i t \nu\omega}z)\; .
\end{equation}
Notice again that the pushforward in \eqref{eq: measure mm} does not affect the Radon-Nikod\'{y}m derivative $ \mathcal{U}_{t,0}(z) \gamma_{\mm}(z) \mathcal{U}^{\dagger}_{t,0}(z) $. The interplay between the quasi-classical limit and the time-evolution can be summed up in the following commutative diagram involving the Radon-Nikod\'{y}m derivatives 
\begin{equation}
	\label{eq: diagram}
  \begin{tikzcd}
	 \Gamma_{\eps} \arrow[r, "e^{-i H_{\eps} t}",mapsto] \arrow[d, "\eps \to 0",mapsto] &[2.5cm] \Gamma_{\eps}(t) \arrow[d, "\eps \to 0",mapsto]	\\[1cm]
	 \gamma_{\mm}(z) \: \diff\mu_{\mm}(z) 	\arrow[r, "\eqref{eq: measure mm}",mapsto] &  \mathcal{U}_{t,0}(z) \gamma_{\mm}(z) \mathcal{U}^{\dagger}_{t,0}(z) \: \diff \lf( e^{-it\nu\omega}\, _{\star} \,\mu_{\mm} \ri).
\end{tikzcd}
\end{equation}
where we have decomposed the initial state-valued measure as $ \diff \mm (z) = \gamma_{\mm}(z) \diff \mu_{\mm} (z) $, with $ \gamma_{\mm} \in \mathscr{L}^1_{+,1}(\mathscr{H})$ and $ \mu_{\mm} \in \mathscr{M}(\mathfrak{h}) $, and the convergence is always along a given subsequence $ \lf\{ \eps_n \ri\}_{n \in \N} $.

We state now the first result in detail. Recall that we say that $ \mm \in \mathscr{M} \lf(\mathfrak{h}; \mathscr{L}^1_{+,1}\ri) $ is a probability measure, and thus $\mm \in \mathscr{P}\lf(\mathfrak{h}; \mathscr{L}^1_{+,1}\ri)$, whenever $ \lf\| \mm(\mathfrak{h}) \ri\|_{\mathscr{L}^1(\mathscr{H})} = 1 $.

	\begin{thm}[Quasi-classical evolution in the Schr\"{o}dinger picture]
  		\label{thm:5}
  		\mbox{}	\\
  		Let $\nu(\varepsilon)$ be such that $\varepsilon\nu(\varepsilon)\to \nu\in\{0,1\}$, when $\varepsilon\to 0$ and let $\Gamma_{\varepsilon}\in \mathscr{L}^1_{+,1}(\mathscr{H} \otimes \mathscr{K}_{\eps}) $ be a state satisfying assumption \eqref{eq:20}. Let also \eqref{eq:13} be satisfied for the polaron and Pauli-Fierz models.  Then, there exist at least one subsequence $ \lf\{ \eps_n \ri\}_{n \in \N} $ and one measure $ \mm \in \mathscr{M} \lf( \mathfrak{h}; \mathscr{L}^1_{+,1}(\mathscr{H})\ri) $, such that 
  		\beq
			\label{eq: teo princ 1} 
			\Gamma_{\varepsilon_n} \xrightarrow[n \to + \infty]{\mathrm{qc}} \mm,  
		\eeq 
		and, if \eqref{eq: teo princ 1} holds, then for all $t\in \mathbb{R}$,
 		\begin{equation}
  			\label{eq: t convergence}
    			\Gamma_{\varepsilon_n}(t) \xrightarrow[n \to + \infty]{\mathrm{qc}} \mm_t \; ,
 		\end{equation}
  		where ${\mm}_t$ is given by \eqref{eq: measure mm}.
	\end{thm}

	\begin{corollary}[Mass conservation]
		\label{cor:2}
          	\mbox{}	\\
          	If in addition $\Gamma_{\varepsilon}$ satisfies assumption \eqref{eq: ass1} or \eqref{eq: ass2}, then $\mathfrak{m}\in \mathscr{P}\lf(\mathfrak{h}; \mathscr{L}^1_{+,1}(\mathscr{H})\ri)$ and thus $\mathfrak{m}_t\in \mathscr{P}\lf(\mathfrak{h}; \mathscr{L}^1_{+,1}(\mathscr{H})\ri)$ for all $t\in \mathbb{R}$. Furthermore, if condition \eqref{eq: ass1} holds also for $ \gamma_{\eps}(t) $ for any $ t \in [0,T)$, $T \in \R^+ $, then for all $t\in [0,T)$,
        	\begin{equation*}
          		\Gamma_{\varepsilon_n}(t) \xrightarrow[n \to + \infty]{\mathrm{bqc}} \mathfrak{m}_t\; .
        	\end{equation*}
	\end{corollary}

  	\begin{remark}[Extraction of a subsequence]
  		\label{rem: extraction}
  		\mbox{}	\\
  		Let us point out that, as anticipated above, the limit measure $ \mm $ at initial time, according to \cref{def: qc convergence}, might depend on the choice of the subsequence $ \lf\{ \eps_n \ri\}_{n \in \N} \to 0 $. However, we stress that the convergence at time $ t $ stated in \eqref{eq: t convergence} occurs along the {\it same} subsequence. 
  	\end{remark}
  
  	\begin{remark}[Loss of mass]
  		\label{rem: role}
  		\mbox{}	\\
           \cref{thm:5} holds irrespective of any possible loss of mass for the initial-time convergence. The quasi-classical evolution preserves the mass, thus proving that \emph{the same amount of mass is lost at any time}. Conditions \eqref{eq: ass1} and \eqref{eq: ass2} ensure that no mass is lost at initial time, and thus at any further time. Another sufficient condition to ensure no-loss of mass is the so-called (PI) condition, that will be discussed in detail in \cref{sec:bound-conv-pi}. However, as suggested by the fact that physical factorized states $\gamma\otimes \varsigma_{\varepsilon}$ do not lose mass, this peculiar loss of mass phenomenon is due either to a ``bad''  correlation between the field and particle subsystems, or to a somewhat artificial dependence of the particle subsystem on the quasi-classical parameter in an uncorrelated state (see \cref{sec:bound-conv-pi} for a more detailed discussion).
	\end{remark}
  
  	\begin{remark}[Assumptions \eqref{eq: ass1} and \eqref{eq: ass2}]
  		\label{rem: assumptions}
  		\mbox{}	\\
  		The implications and the meaning of assumptions \eqref{eq:
                  ass1} and \eqref{eq: ass2} is \emph{a priori} quite
                different. For instance, \eqref{eq: ass2} provides a uniform
                control on the reduced density matrix $ \gamma_{\eps} $ but has
                little physical motivation, unless the two subsystems
                  in the state are uncorrelated (\emph{i.e.}, the state has a
                  tensor product structure). Assumption
                \eqref{eq: ass1} on the other hand implies the stronger convergence of
                $\Gamma_{\eps} $ to $ \mm $ in the bounded quasi-classical sense
                of \cref{def:11}. Such a stronger convergence holds true
                however only at initial time, and its propagation along the
                time-evolution is typically very difficult to prove. A
                notable exception is given by trapped particle systems, {\it
                  i.e.,} when $ \mathcal{K}_0 $ has compact resolvent and
                thus one can take $ \mathcal{A} = (\mathcal{K}_0+1)^{\delta} $,
                for some $ \delta > 0 $, in \eqref{eq: ass1}. Thus, in this case
                the assumption $ \Tr(\Gamma_{\eps}(K_0+\mathrm{d}\GG_{\varepsilon}(\omega))^\delta ) \leq
                C $ on the initial state is sufficient to strengthen the
                convergence at any time (see, {\it e.g.}, \cite[Lemma
                3.4]{ammari2014jsp} for the Nelson model and
                \cref{sec:pull-thro-form} for the polaron and Pauli-Fierz
                models). As we already remarked, the two assumptions are
                  in fact related, because \emph{\eqref{eq: ass2} implies
                    \eqref{eq: ass1}}\footnote{The inference $\text{\eqref{eq: ass2} } \Rightarrow \text{ \eqref{eq: ass1}}$ can be proved as follows. Since $\gamma$ is a positive trace class  operator, it can be decomposed as $\gamma=\sum_{j\in \mathbb{N}}^{} \ket{  \varphi_j} \bra{\varphi_j}$, where $(\varphi_j)_{j\in \mathbb{N}}$ is an orthonormal basis   of $\mathscr{H}$ and $\lambda_j\geq 0$ are the singular values satisfying $\sum_{j\in \mathbb{N}}^{}\lambda_j <+\infty$. Therefore, there exists a non-negative sequence $(\mu_j)_{j\in \mathbb{N}}$ such that: $\mu_j\to +\infty$ and $\sum_{j\in\mathbb{N}}^{}\mu_j\lambda_j<+\infty$. In fact, if there is only a finite number of nonzero $\lambda_j$s, then the existence is trivial, while, if the number of nonzero $\lambda_j$s is infinite, one can set, for all $k\in \mathbb{N}$, $J_k= \min \{ J \in \N \: | \: \sum_{J\leq j}^{} \lambda_j< 2^{-k} \}$, and $\mu_j=1$, for $j<  J_0$, and $\mu_j= 2^{k/2}$, for $J_k \leq j < J_{k+1}$. Then, by construction, the inverse of the operator $\mathcal{A}_{\gamma}:= \sum_{j\in \mathbb{N}}^{}\mu_j \ket{ \varphi_j} \bra{\varphi_j}$ is compact and $\tr_{\mathscr{H}}(\mathcal{A}_{\gamma}\gamma_{\varepsilon}) =  \tr_{\mathscr{H}}(\mathcal{A}_{\gamma}^{1/2}\gamma_{\varepsilon}\mathcal{A}_{\gamma}^{1/2}) \leq \tr_{\mathscr{H}}(\mathcal{A}_{\gamma}^{1/2}\gamma\mathcal{A}_{\gamma}^{1/2}) = \sum_{j\in \mathbb{N}}^{}\mu_j\lambda_j<+\infty$, so that $\gamma_{\varepsilon}$ satisfies assumption \eqref{eq: ass1}.}. However, due to the different
              physical implications (particles' trapping on one
              hand, isolation of the subsystems on the other), we preferred
              to keep the two separated.
	\end{remark}
  
	\begin{remark}[Rougher potentials]
  		\label{rem:3}
  		\mbox{}	\\
  		As already remarked in \cref{sec:concr-models:-nels}, states satisfying~\eqref{eq:20} yield effective potentials $\mathcal{V}_{t}$ that are ``regular''. For example, no confining potential can be obtained with such quasi-classical states. It is possible to obtain more general effective potentials relaxing assumption~\eqref{eq:20} to accommodate states whose limit are \emph{cylindrical measures} \cite{falconi2017arxiv}, however the analysis becomes more complicated. In the polaron model, for coherent states, whose cylindrical measure is concentrated in a single ``singular'' point (a suitable tempered distribution), the analysis has been carried out in \cite{carlone2019arxiv} to obtain an effective (time-dependent) point interaction.
	\end{remark}

	Before proceeding further we discuss in some detail the scaling factor $\nu(\varepsilon)$ that appears in front of the free energy of the field in the Hamiltonian $H_{\varepsilon}$. Physically, one should distinguish between two relevant situations: $\nu(\varepsilon)=1$, and $\nu(\varepsilon)=\frac{1}{\varepsilon}$; all the other possibilities are physically less relevant,
and yield the same qualitative results, up to rescaling of the parameters. Let us remark first however that, despite the fact that the two cases yield different evolutions for the classical field, \emph{the interaction is always too weak to cause a back-reaction of the particles on the   field}, when $\varepsilon\to 0$.  Thus, the quasi-classical field can indeed be seen as an environment.
  
	\begin{remark}[$\nu = 0$]
		\label{rem: nu 0}
		\mbox{}	\\
		When $\nu(\varepsilon)=O(\varepsilon^{-\delta})$, $\delta<1$, the quasi-classical field remains constant in time. In fact,
                in such a case $\nu=\lim_{\varepsilon\to 0}\varepsilon\nu(\varepsilon)=0$, and therefore $\mathcal{U}_{t,s}(z) =
                \mathcal{U}_{t-s}(z) $ is the strongly continuous group generated by the self-adjoint
                operator $\mathcal{K}_0+\mathcal{V}(z)$, with, \emph{e.g.},
                $\mathcal{V}(z)=\sum_{j=1}^N2\Re \braketl{\lambda(\xv_j)}{z}_{\mathfrak{h}}$ for the Nelson
                  model. Also, the measure $\mu_t$ is constant: $\mu_t=\mu$ for all $t\in \mathbb{R}$. Therefore, in the
                scaling yielded by $\nu(\varepsilon)= \OO(\varepsilon^{-\delta})$ the radiation field does not evolve. Let us remark
                that, in the case of the polaron, this is the scaling equivalent, up to suitable
                rescalings, to the well-known strong coupling regime.
	\end{remark}
	
	\begin{remark}[$\nu = 1$]
		\label{rem: nu 1}
		\mbox{}	\\
		When $ \nu = 1 $, {\it e.g.}, if $ \nu(\eps) = \frac{1}{\eps} $, the quasi-classical radiation field evolves in time in a non-trivial way, obeying a free
field equation, and therefore the effective evolution operator for the
particles $\mathcal{U}_{t,s}(z)$ has a time-dependent
generator. For the regularized Nelson model, such a free evolution is given by the Klein-Gordon-like
equation 
	\beq
		(\partial_t^2 + \omega^2(D)) A=0,
	\eeq
	where $\omega(D)$ is the pseudodifferential operator defined by the Fourier
transform of the function $\omega$; for the Pauli-Fierz model it is given by the free Maxwell equations in
the Coulomb gauge, and for the polaron  by the equation $(\partial^2_t+1)A=0$. Here, for clarity, we have written such equations in the
usual form, that involves the real field $A$ and its time derivatives. Throughout the paper
however, we use the complex counterpart of such real field, that we denoted by $z$, and which is given in terms of $ A $, {\it e.g.}, in the regularized Nelson model, by
	\bdm
			z = \tx\frac{1}{2} \lf(\omega^{1/2}(D) A + i \omega^{-1/2}(D)  \partial_t A \ri),
	\edm
	Hence, the evolution equation for $z$ becomes $i\partial_t
z=\omega z$. 
	\end{remark}
	
A consequence of \cref{thm:5} is that, for any compact operator $ \mathcal{B} \in \mathscr{L}^{\infty}(\mathscr{H}) $, its Heisenberg evolution satisfies
	\beq
		\Tr \bigl( \Gamma_{\eps_n} e^{it H_{\varepsilon_n}} B e^{-it H_{\varepsilon_n}} \bigr) \xrightarrow[n \to + \infty]{} \int_{\mathfrak{h}}^{} \mathrm{d}\mu_{\mm}(z) \: \tr_{\mathscr{H}} \lf[ \gamma_{\mm}(z) \: \mathcal{U}^{\dagger}_{t,0}(z)\,\mathcal{B} \,\mathcal{U}_{t,0}(z) \ri] \; .
	\eeq
There is also a counterpart of the above statement for the particle degrees of freedom alone: for any $ \Gamma_{\eps} $ as in Theorem \ref{thm:5}, the following weak-$*$ convergence holds in $\mathscr{L}^1(\mathscr{H})$,
	\beq
		\tr_{\mathscr{K}_{\eps}} \bigl(e^{-it H_{\varepsilon_n}} \Gamma_{\eps_n} e^{it H_{\varepsilon_n}} \bigr) \xrightarrow[n \to + \infty]{\mathrm{w}*} \int_{\mathfrak{h}}^{} \mathrm{d}\mu_{\mm}(z) \: \mathcal{U}_{t,0}(z)\,\gamma_{\mm}(z) \,\mathcal{U}^{\dagger}_{t,0}(z)  \; ,
	\eeq
{\it i.e.}, the particle state obtained by tracing out the field degrees of freedom evolves as $ \eps \to 0 $ into the r.h.s.\ of the above expression. When the state is a product state, the above result can be made more explicit (see also \cref{rem: product states}):

	\begin{corollary}[Quasi-classical evolution of product states]
  		\label{cor:7}
  		\mbox{}	\\
		Let $\varsigma_{\varepsilon}\in \mathscr{L}^1_{+,1}(\mathscr{K}_{\varepsilon})$ be a field's state such that for all $ \gamma
                \in \mathscr{L}^1_{+,1}(\mathscr{H})$, $ \gamma \otimes \varsigma_{\varepsilon}$ satisfies assumption~\eqref{eq:20},
                and~\eqref{eq:13} for the polaron and Pauli-Fierz models, so that there exists $ \mu \in
                \mathscr{M}(\mathfrak{h}) $ such that \beq \varsigma_{\varepsilon_n} \xrightarrow[n \to + \infty]{} \mu.  \eeq Then,
                for all $\mathcal{B}\in \mathscr{L}^{\infty}(\mathscr{H}) $ and all $t\in \mathbb{R}$,
  		\begin{equation}
    			 \tr_{\mathscr{K}_{\varepsilon_n}}\bigl( \varsigma_{\varepsilon_n}e^{it H_{\varepsilon_n}}B e^{-it H_{\varepsilon_n}} \bigr) \xrightarrow[n \to + \infty]{\mathrm{w}} \int_{\mathfrak{h}}^{}\mathcal{U}_{0,t}(z)\,\mathcal{B}\,\mathcal{U}_{t,0}(z)  \mathrm{d}\mu(z)\;.
  		\end{equation}
	\end{corollary}
	
	\begin{remark}[Bounded operators]
		\mbox{}	\\
		It would obviously be more satisfactory to extend the above result to bounded operators $\mathcal{B} \in \mathscr{B}(\mathscr{H}) $. However, this can not be done in full generality because the convergence in \cref{def: qc convergence} holds in weak-$*$ topology. As explained in \cref{rem: assumptions,sec:bound-conv-pi}, one can lift the convergence to weak topology, and thus extend the statement above to bounded observables, if an additional regularity on the initial state is assumed and such a regularity can be propagated by the dynamics, which can be done for example whenever the particle system is trapped.
	\end{remark}

The analogue of the above \cref{cor:7} for non-product states and more complicated observables, {\it i.e.}, self-adjoint operators acting on the full Hilbert space, is more involved to state and holds true only for a subclass of such operators. We indeed introduce a class of operators on $ \mathscr{H} \otimes \mathscr{K}_{\eps} $, consisting of
polynomials with $m$ creation and $n$ annihilation normal ordered operators, with arguments possibly
depending on the particle's positions: explicitly, we consider operators $ \mathrm{Op}_{\varepsilon_n}^{\mathrm{Wick}}(\mathcal{F}) $ obtained as the {\it Wick quantization} of symbols $ \mathcal{F} \in \mathscr{S}_{n,m} $, {\it i.e.}, of the form
\begin{equation}
  	\label{eq:5}\tag{$\mathscr{S}_{\ell,m}$}
  	\mathcal{F}(z) =\sum_{j=1}^{N} \braketr{z}{\lambda_1(\xv_j)}_{\mathfrak{h}} \dotsm \braketr{z}{\lambda_\ell(\xv_j)}_{\mathfrak{h}} \braketl{\lambda_{\ell+1}(\xv_j)}{z}_{\mathfrak{h}} \dotsm \braketl{\lambda_{\ell+m}(\xv_j)}{z}_{\mathfrak{h}}\;,
\end{equation}
where $ \la_j \in L^{\infty}(\R^{d}; \mathfrak{h}) $, $ j = 1, \ldots, m+\ell $.

To state the result, we also need to make more restrictive assumptions on the initial state $\Gamma_{\varepsilon}$:
\begin{equation}
  \label{eq:29}\tag{$\mathrm{A}_{\delta}$}
  \begin{cases}
    	\exists \delta>\frac{1}{4} , \exists C_{\delta} < +\infty \,:\;\Tr\bigl(\Gamma_{\varepsilon} {(\mathrm{d}\Gamma_{\varepsilon}(1)+1)}^{2\delta}\bigr)\leq C_{\delta}, & \text{Nelson model,}\\
    \delta=1, \exists C < +\infty \, :\; \Tr\bigl(\Gamma_{\varepsilon}H^2_{\varepsilon}\bigr)\leq C\nu(\varepsilon)^2, &\text{Pauli-Fierz model,}\\
    \exists \delta\in \mathbb{N}_{*} , \exists C_{\delta} < +\infty \, :\; \Tr\bigl(\Gamma_{\varepsilon}H_{\varepsilon}^{2\delta}\bigr)\leq C_{\delta}\nu(\varepsilon)^{2\delta}, &\text{polaron.}% \\
  \end{cases}
\end{equation}

	\begin{thm}[Quasi-classical evolution in the Heisenberg picture]
  		\label{thm:8}
 		\mbox{}	\\
  		Let $\Gamma_{\varepsilon}\in \mathscr{L}^1_{+,1}(\mathscr{H}\otimes \mathscr{K}_{\varepsilon})$ be a state satisfying
                assumption~\eqref{eq:29}, so that there exists $ \mm \in \mathscr{M}\lf(\mathfrak{h};
                \mathscr{L}^1_{+}(\mathscr{H})\ri) $ such that \beq \Gamma_{\varepsilon_n} \xrightarrow[n \to + \infty]{\mathrm{qc}} \mm.
                \eeq Then, for all $ \mathcal{F} \in \mathscr{S}_{\ell,m}$, with $\frac{\ell+m}{2}<2\delta $, for all
                $t\in \mathbb{R}$ and for all $ \mathcal{S}, \mathcal{T} \in \mathscr{B}(\mathscr{H}) $, such that
                  either $\mathcal{S}$ or $\mathcal{T}\in \mathscr{L}^{\infty}(\mathscr{H})$, 
                  \bml{
                  \Tr\bigl(\Gamma_{\varepsilon_n}(t)T \mathrm{Op}_{\varepsilon_n}^{\mathrm{Wick}}(\mathcal{F})S \bigr) \xrightarrow[n \to + \infty]{}  \tr_{\mathscr{H}} \lf( \int_{\mathfrak{h}} \diff \mm_t(z) \: \mathcal{T} \mathcal{F}(z) \mathcal{S} \ri)	\\
                  = \tr_{\mathscr{H}} \lf(\int_{\mathfrak{h}}^{}\mathrm{d}\mu_{\mm}(z) \: \mathcal{U}_{t,0}(z)
                  \gamma_{\mm}(z) \mathcal{U}_{0,t}(z)\,\mathcal{T} \mathcal{F}(e^{-it\nu \omega}z) \mathcal{S}\, \ri)\;.
                }
\end{thm}

\begin{remark}[Regularity assumptions for the Pauli-Fierz model]
  \label{rem:1}
  	\mbox{}	\\
   The constraint $\delta=1$ for the Pauli-Fierz model is due to some technical difficulties
    in propagating in time higher order regularity of the number operator, due to the fact that the
    number operator and the field's kinetic term are not comparable in such case, since the field carriers may be massless.
\end{remark}

\subsection{Semiclassical Analysis and Sketch of the Proof}
\label{sec:bound-conv-pi}

In this section we present a short sketch of the proof and discuss some of key features of semiclassical analysis for infinite dimensional systems, which is the core tool of our analysis. This discussion is meant to clarify the role of our assumptions and propose alternative approaches.

One of the main point in our investigation is the convergence of a family of quantum states as $ \eps \to 0  $ to a quasi-classical Wigner measure in the sense either of \cref{def: qc convergence} or \cref{def:11}. The latter is clearly preferable but there are known obstructions. Indeed, in infinite dimensional semiclassical analysis with no additional degrees of freedom (which we refer to as the {\it scalar case}), {\it i.e.}, when the limit Wigner measure is a conventional scalar measure, there can be two types of defects of convergence for a given family of normalized states $ \lf\{ \Gamma_{\eps} \ri\}_{\eps \in (0,1)} $: 
\begin{itemize}
	\item {\it a loss of mass}, as in the finite-dimensional case, {\it i.e.}, the limit measure may not be a probability measure and have a total mass strictly less than one;
	\item a dimensional {\it loss of compactness} that is characteristic of the infinite-dimensional setting (see \cite[\textsection 7.4]{ammari2008ahp}), where the mass is preserved but the expectation values of operators obtained as Wick quantization of non-compact symbols do not converge to their limit expressions.
\end{itemize} 

These defects are prevented formulating conditions that are both sufficient and reasonable to verify in a relevant class of concrete examples: the loss of mass is prevented by imposing a $\varepsilon$-uniformity condition on the expectation of some power of the number operator, analogous to assumption \eqref{eq:20} given above (see also \cite[\textsection
6.1]{ammari2008ahp}); loss of compactness is prevented by the so-called (PI) sufficient condition
\cite{ammari2011jmpa,ammari2019tjm}, that reads as follows: for all $k\in \mathbb{N}$,
\begin{equation*}
  \lim_{\varepsilon \to  0} \Tr\bigl(\mathrm{d}\mathcal{G}_{\varepsilon}(1)^k\Gamma_{\varepsilon}\bigr)=\int_{}^{}\mathrm{d}\mu(z)\,\lVert z  \rVert_{}^{2k}  \; .
\end{equation*}
If the above condition holds, then the expectations of all Wick polynomial bounded symbols converge.

A key difference of our quasi-classical setting compared with the scalar case is that assumption \eqref{eq:20} is \emph{not sufficient} to ensure that no mass is lost in the limit: the correlation with the $\mathscr{H}$-degrees of freedom may cause a new type of mass defect, as exemplified by the product state $\gamma_{\varepsilon}\otimes \varsigma_{\varepsilon}$ introduced above in \cref{rem: mass conservation}, that satisfies assumption \eqref{eq:20} and converges to the state-valued measure $ 0 $ with no mass. Assumptions \eqref{eq: ass1} and \eqref{eq: ass2} are both conditions on the $\mathscr{H}$-degrees of freedom, that are sufficient to prevent this quasi-classical defect and whose usefulness has been discussed in \cref{rem: assumptions}. Note however that the defect of compactness mentioned above may also occur: it is indeed not difficult to produce examples of states with no loss of mass but a defect of compactness, simply tensorizing any scalar example of such defect with a $\mathscr{H}$-state that is independent of the quasi-classical parameter $\varepsilon$. In order to
overcome the defect of compactness in the quasi-classical limit, a straightforward analogous of the scalar
(PI) condition can be formulated:
\begin{equation}
  \label{eq:PIK}
  \tag{PI$_K$}
  \lim_{\varepsilon\to 0} \Tr\Bigl(\mathrm{d}\mathcal{G}_{\varepsilon}(1)^k\, \Gamma_{\varepsilon}\Bigr)=\int_{\mathfrak{h}}^{}\mathrm{d}\mu_{\mathfrak{m}}(z)\,\lVert z  \rVert_{\mathfrak{h}}^{2k}\,,	\qquad		\forall 0\leq k\leq K\;, 
\end{equation}
so that the scalar condition (PI) of \cite{ammari2011jmpa,ammari2019tjm} corresponds here to the condition (PI$_{\infty}$). We allow for a possibly
finite index $K$ for a motivation to be explained in detail below, related to the propagation in time of the
condition.

Restricting to states that satisfy the condition \eqref{eq:PIK}, for suitable $K\in \mathbb{N}_{*}$, allows to use some powerful tools of infinite dimensional semiclassical analysis (see, \emph{e.g.}, \cite[Appendix B]{ammari2019tjm}), that make the study of quasi-classical dynamics less involved, and allow to obtain stronger results. In particular, the convergence in \cref{thm:5} can be lifted to bounded quasi-classical convergence of \cref{def:11} and \cref{thm:8} holds for a more general class of symbols. There are, however, also some drawbacks. The most relevant is that there are states of physical relevance that do not satisfy condition \eqref{eq:PIK}, or that are defined by abstract and \emph{a priori} considerations, in a way that does not provide enough information to test the validity of such condition. The foremost examples of this kind are those of states satisfying suitable variational problems (\emph{e.g.}, ground states of physical problems related to the ones under study, perhaps with an additional external potential that is removed at the initial time, or states belonging to some minimizing sequence of the model). In addition, there is a technical difficulty: the condition \eqref{eq:PIK} is in general difficult to propagate in time. As it will be explained later, to prove its propagation one has to rely on a propagation estimate for the number operator up to power $K$. This is possible for the Nelson model for all $K\in \mathbb{N}_{*}$ \cite{falconi2013jmp}, and for the Pauli-Fierz model, at least for $K\leq 2$ \cite{ammari2019prep}. However, it does not seem feasible for the polaron model.

In view of the above considerations, in this paper we mainly focus on the more general class of states
satisfying only assumption \eqref{eq:20} and that thus can be defective both on mass and compactness, as in the main results presented in \cref{sec:main-results}. Let us remark that, in order to consider more general states, a finer technical analysis on our part is required; this however makes the proofs also slightly more involved. We believe that it is interesting to present the results in such generality, both from a physical and a mathematical standpoint. Nonetheless, we also believe that it makes sense to informally present the proof of our results that can be obtained using the condition \eqref{eq:PIK}, for arbitrary $K\in \mathbb{N}_{*}$ in the Nelson model, and for $K\leq 2$ in the Pauli-Fierz model. The purpose of such outline is twofold: on one hand it serves as a summary of the
semiclassical strategies used throughout the paper, on the other hand it allows us to emphasize the
simplifications yielded by using the condition \eqref{eq:PIK}, and thus also the subtleties that we had to face otherwise.

Let us then assume, only in this section, that \eqref{eq:PIK} holds true, so that the statement of \cref{thm:5} takes the following stronger form: there exist at least one subsequence $ \lf\{ \eps_n \ri\}_{n \in \N} $ and one probability measure $ \mm \in \mathscr{P}\lf(\mathfrak{h};
  \mathscr{L}^1_{+,1}(\mathscr{H})\ri) $, such that 
\beq
  	\label{eq: bounded teo princ 1}
  	\Gamma_{\varepsilon_n} \xrightarrow[n \to + \infty]{\mathrm{bqc}} \mm, 
\eeq 
and, if \eqref{eq: bounded teo princ 1} holds, then for all $t\in  \mathbb{R}$,
\begin{equation}
	\label{eq: t convergence b}
	\Gamma_{\varepsilon_n}(t) \xrightarrow[n \to + \infty]{\mathrm{bqc}} \mm_t \; ,
\end{equation}
where ${\mm}_t$ is given by \eqref{eq: measure mm}. In addition,
  \begin{multline*}
    \Tr \lf[ \Gamma_{\varepsilon_n}(t)T \mathrm{Op}_{\varepsilon_n}^{\mathrm{Wick}} \lf( \big\langle z^{\otimes \ell}  \big\vert \tilde{b} z^{\otimes m}
    \big\rangle_{\mathfrak{h}^{\otimes_{\mathrm{s}}\ell}} \ri) S \ri] \xrightarrow[n \to + \infty]{}  \tr_{\mathscr{H}} \lf[ \int_{\mathfrak{h}} \diff \mm_t(z) \: \mathcal{T} \,\big\langle z^{\otimes \ell}  \big\vert \tilde{b} z^{\otimes m}
    \big\rangle_{\mathfrak{h}^{\otimes_{\mathrm{s}}\ell}}\, \mathcal{S} \ri] \\
    = \tr_{\mathscr{H}} \lf[ \int_{\mathfrak{h}}^{}\mathrm{d}\mu_{\mm}(z) \: \mathcal{U}_{t,0}(z)
    \gamma_{\mm}(z) \mathcal{U}_{0,t}(z)\,\mathcal{T} \,\lf\langle (e^{-it\nu \omega}z)^{\otimes \ell}  \lf\vert \tilde{b} (e^{-it\nu\omega}z)^{\otimes m}  \ri.\ri\rangle_{\mathfrak{h}^{\otimes_{\mathrm{s}}\ell}}\, \mathcal{S} \ri]\, \;,
  \end{multline*}
  for all $\mathcal{T},\mathcal{S}\in \mathscr{B}(\mathscr{H})$, and for all $\tilde{b}\in
  \mathscr{B}(\mathscr{H}\otimes \mathfrak{h}^{\otimes_{\mathrm{s}}m},\mathscr{H}\otimes
  \mathfrak{h}^{\otimes_{\mathrm{s}}\ell})$, with $\frac{m+\ell}{2}<K$.
  
The main steps of the proof of the above results are the following:

  \begin{enumerate}[i)]
  \item First of all, we pass to the interaction representation, setting
    \begin{equation*}
      \Upsilon_{\varepsilon}(t):= e^{it(K_0+\nu(\varepsilon)\mathrm{d}\mathcal{G}_{\varepsilon}(\omega))}\Gamma_{\varepsilon}(t)e^{-it(K_0+\nu(\varepsilon)\mathrm{d}\mathcal{G}_{\varepsilon}(\omega))}\; ,
    \end{equation*}
    and write Duhamel's formula for the time evolution of its Fourier transform
    \begin{multline}
      \label{eq:11}
      \lf[\hgint(t)\ri](\eta)=\lf[\hgint(s)\ri](\eta) \\-i\sum_{j=1}^N\int_s^t \mathrm{d}\tau \: e^{i\tau \mathcal{K}_0}\,\tr_{\mathscr{K}_{\varepsilon}}\lf( \lf[ \varphi_{\varepsilon}\lf(e^{-i\tau\varepsilon\nu(\varepsilon)\omega}\lambda(\xv_j) \ri),\gint(\tau) \ri]\,W_{\varepsilon}(\eta) \ri)\,e^{-i\tau \mathcal{K}_0} \;,
    \end{multline}
    where $ \varphi_{\eps}(z) : = a_{\eps}^{\dagger}(z) + a_{\eps}(z) $ is the field operator.
    
  \item Now, the goal is to extract a common subsequence of $\Upsilon_{\varepsilon_{n}}(t)$ that converges for  all times $t\in \mathbb{R}$. Hence, one preliminary needs to show that, at any time $ t \in \R $, $ \Upsilon_{\varepsilon_{n}}(t) $ converges along a suitable subsequence. In order to do that, we need to verify that condition \eqref{eq:PIK} (resp. \eqref{eq:20}, in the case of  \cref{thm:5}) is satisfied for any $ t \in \R $, which guarantees convergence in the sense of \cref{def:11} (resp. \cref{def: qc convergence}) at all times. Let us skecth how it is possible to propagate \eqref{eq:PIK} in the Nelson model: the Duhamel formula
    \begin{multline*}
      \Tr\Bigl(\mathrm{d}\mathcal{G}_{\varepsilon}(1)^k\, \Upsilon_{\varepsilon}(t)\Bigr)= \Tr\Bigl(\mathrm{d}\mathcal{G}_{\varepsilon}(1)^k\, \Upsilon_{\varepsilon}(0)\Bigr) \\
      -i \int_0^t  \Tr\Bigl(\bigl[\mathrm{d}\mathcal{G}_{\varepsilon}(1)^k,\varphi_{\varepsilon}\lf(e^{-i\tau\varepsilon\nu(\varepsilon)\omega}\lambda(\xv_j) \ri)\bigr]\Upsilon_{\varepsilon}(\tau)\Bigr) \mathrm{d}\tau\; ,
    \end{multline*}
    yields
    \begin{multline*}
      	\lf\| \lf[\mathrm{d}\mathcal{G}_{\varepsilon}(1)^k,\varphi_{\varepsilon}\lf(e^{-i\tau\varepsilon\nu(\varepsilon)\omega}\lambda(\xv_j) \ri) \ri]\Upsilon_{\varepsilon}(\tau)  \ri\|_{\mathscr{L}^1(\mathscr{H}\otimes \mathscr{K}_{\varepsilon})}^{}\leq \varepsilon\, C_k \lf\| (\mathrm{d}\mathcal{G}_{\varepsilon}(1)+1)^k\Upsilon_{\varepsilon}(\tau)  \ri\|_{\mathscr{L}^1(\mathscr{H}\otimes \mathscr{K}_{\varepsilon})}^{}\; .
    \end{multline*}
    The trace norm on the right hand side is uniformly bounded w.r.t $\varepsilon$ (see, \emph{e.g.} \cref{prop:12} below), yielding
    \begin{equation*}
      \Tr\Bigl(\mathrm{d}\mathcal{G}_{\varepsilon}(1)^k\, \Upsilon_{\varepsilon}(t)\Bigr)= \Tr\Bigl(\mathrm{d}\mathcal{G}_{\varepsilon}(1)^k\, \Upsilon_{\varepsilon}(0)\Bigr)+ \mathcal{O}_{k,t}(\varepsilon)\; .
    \end{equation*}
    This implies that condition \eqref{eq:PIK} holds for all times, provided it holds at the initial
    time. Let us remark again that such propagation estimate (more precisely, the $\varepsilon$-uniform number estimate) is not available for the polaron model, nor for the Pauli-Fierz model whenever $k\geq 3$.

    \item Once convergence of $ \Upsilon_{\varepsilon_{n}}(t) $ is obtained, one has to prove that $ \big\{\hat{\Upsilon}_{\varepsilon} \big\}_{\varepsilon\in (0,1)}$ is uniformly equicontinuous as a family of functions of time. This can be done exploiting once more \eqref{eq:PIK}: 
    \begin{multline*}
      \Tr \lf( \lf[\varphi_{\varepsilon}\lf(e^{-i\tau\varepsilon\nu(\varepsilon)\omega}\lambda(\xv_j)\ri),\Upsilon_{\varepsilon}(t)  \ri] W_{\varepsilon}(\eta) \ri)\\=\Tr \lf( \Upsilon_{\varepsilon}(t)(\mathrm{d}\mathcal{G}_{\varepsilon}(1))^{\delta}(\mathrm{d}\mathcal{G}_{\varepsilon}(1))^{-\delta} \lf[W_{\varepsilon}(\eta),\varphi_{\varepsilon}\lf(e^{-i\tau\varepsilon\nu(\varepsilon)\omega}\lambda(\xv_j)\ri) \ri] \ri)= \OO (\varepsilon^{\delta})\; .
    \end{multline*}
    
  \item\label{item:1} After the extraction of a subsequence $\Upsilon_{\varepsilon_{n_k}}(t)$ converging at all time, we take the limit $\varepsilon_{n_k} \to 0 $ of \eqref{eq:11}: by \eqref{eq:PIK}, the convergence follows by a direct generalization to operator-valued symbols of the analysis done in the scalar case, \emph{e.g.} in \cite[Appendix B]{ammari2019tjm}. On the other hand, to study the limit under assumption \eqref{eq:20} alone, we have to develop a specific quasi-classical calculus for the   symbols appearing in the energy functionals of the three models. We take advantage of an approximation in simple functions that allows to separate the two types of degrees of freedom, at the same time making the symbol compact and thus convergent without additional assumptions (see \cref{sec:quasi-class-analys}).

  \item\label{item:2} The equation obtained in the limit from the Duhamel equation is a transport equation
    for the Fourier transform of the quasi-classical measures in interaction picture $\mathfrak{n}_t$:
    \begin{equation*}
      \widehat{\nn}_t(\eta)= \widehat{\nn}_s(\eta)-i\int_s^t  \mathrm{d}\tau  \int_{\mathfrak{h}}^{} \mathrm{d}\mu_{\nn_\tau}(z) \: \lf[ \widetilde{\mathcal{V}}_{\tau}(e^{-i\tau\nu\omega}z) , \gamma_{\nn_\tau}(z) \ri] e^{2i\Re \braket{\eta}{z}_{\mathfrak{h}}}\; .
    \end{equation*}
    We prove that such equation has a unique solution, given by $ \nn_t = \lf( \mu_{\mm}, \:
    \widetilde{\mathcal{U}}_{t,0}(z) \gamma_{\mm}(z) \widetilde{\mathcal{U}}_{t,0}^{\dagger}(z) \ri) $, that yields
    once the interaction representation is removed,
    \begin{equation*}
      \diff \mm_t = \mathcal{U}_{t,0}(z) \gamma_{\mm}(z) \mathcal{U}^{\dagger}_{t,0}(z) \: \diff \lf(e^{-it\nu\omega}\, _{\star} \,\mu_{\mm} \ri)\; ,
    \end{equation*}
    as expected (see \cref{sec:uniq-quasi-class}).
    
  \item\label{item:3} The aforementioned uniqueness allows finally to extend the convergence to the
    original sequence $\Upsilon_{\varepsilon_n}(t)$.
  \end{enumerate}

\section{Quasi-Classical Analysis}
\label{sec:quasi-class-analys}

In this section we introduce the quasi-classical asymptotic analysis, needed to study the dynamical limit of
quasi-classical systems. In particular, we have to develop a semiclassical theory for operator-valued
symbols, since the latter are crucial to characterize the interaction part of the dynamics. The key tools presented here are
\begin{itemize}
	\item the convergence of regular states to state-valued measures in the quasi-classical limit (\cref{prop:8}) in the sense of \cref{def: qc convergence};
	\item the convergence of expectation values of suitable classes of operators to their classical counterparts (\cref{prop:10}).
\end{itemize}
Note that, in the context of finite dimensional semiclassical analysis, operator-valued symbols correspoding to additional degrees of freedom have already been studied \cite{ba,fg,pg,gms} (see also \cite[Appendix B]{teu} and references therein), although with different applications in mind.

We start by clarifying the notion of state-valued measure.

	\begin{definition}[State-valued measure]
  		\label{def:9}
  		\mbox{}	\\
  		An additive measure $\mm$ on a measurable space $(X,\Sigma)$ is $\mathscr{H}$-state-valued iff
  		\begin{itemize}
  			\item $\mm(S)\in \mathscr{L}^1_+(\mathscr{H})$ for any $S\in \Sigma$;
    
  			\item $\mm(\varnothing)=0$;
    
  			\item $\mm$ is unconditionally $\sigma$-additive in trace norm.
  		\end{itemize}
  		A $\mathscr{H}$-state-valued measure is a probability measure
                iff $\lVert \mm(X) \rVert_{\mathscr{L}^1}^{}=1$. We denote by $ \mathscr{M}(X, \Sigma; \mathscr{L}^1_+(\mathscr{H})) $ and $\mathscr{P}(X, \Sigma;\mathscr{L}^1_+(\mathscr{H}))$ and by $ \mathscr{M}(X;\mathscr{L}^1_+(\mathscr{H})) $ and $\mathscr{P}(X; \mathscr{L}^1_+(\mathscr{H}))$ the spaces of $\mathscr{H}$-state-valued measures and probability measures, respectively, w.r.t. either a generic
                  or the Borel $\sigma$-algebra, in case $X$ is a topological space.
	\end{definition}

Using the Radon-Nikod\'{y}m property and positivity, there is a simple characterization of state-valued measures:

	\begin{proposition}[Radon-Nikod\'{y}m decomposition]
  		\label{prop:7}
  		\mbox{}	\\
  		For any measure $\mm\in \mathscr{M}(X, \Sigma;\mathscr{L}^1_+(\mathscr{H}))$, there exists a  scalar measure $ \mu_{\mm} \in
  \mathscr{M}(X, \Sigma) $, with $\mu_{\mm}(X)=\lVert \mm(X) \rVert_{\mathscr{L}^1(\mathscr{H})}^{}$, and a  $\mu_{\mm}$-a.e. defined
  measurable function $ \gamma_{\mm}: X\to \mathscr{L}^1_{+,1}(\mathscr{H})$, such that for any $S\in \Sigma$,
  		\begin{equation}
    			\mm(S) = \int_S^{} \mathrm{d}\mu_{\mm}(z) \: \gamma_{\mm}(z), 
  		\end{equation}
  		with the r.h.s.\ meant as a Bochner integral. In addition, $\mm\in \mathscr{M}(X, \Sigma;\mathscr{L}^1_+(\mathscr{H}))$ is a $\mathscr{H}$-state-valued probability measure iff $\mu_{\mm}$ is a probability measure. We call $(\mu_{\mm}; \gamma_{\mm}(z))$ the Radon-Nikod\'{y}m decomposition of $\mm$.
	\end{proposition}

	\begin{proof}
		First of all we point out that the separable Schatten space of trace class operators $\mathscr{L}^1(\mathscr{H})$ has the Schatten space
of compact operators $\mathscr{L}^{\infty}(\mathscr{H})$ as predual, and therefore it has the Radon-Nikod\'{y}m
property (see, \emph{e.g.}, \cite{dunford1940tams,diestel1977ms}). In addition, since $\mm$ takes values in positive operators, we can define its ``norm'' measure as
		\beq
			\label{eq: norm measure}
			m( \: \cdot \: ) := \lVert \mm( \: \cdot \:) \rVert_{\mathscr{L}^1(\mathscr{H})}.  \eeq In fact, $ m $
                        is a scalar measure such that
                        $\mm\ll m\ll \mm$, {\it i.e.}, $ \mm $ and $m$ are absolutely continuous w.r.t. each other. The latter property can indeed be easily seen as follows:
                        $\mathfrak{m}(S)=0$, as an element of the vector space
                        $\mathscr{L}^1_+(\mathscr{H})$, if and only if $m(S)=\lVert \mathfrak{m}(S)
                        \rVert_{\mathscr{L}^1(\mathscr{H})}^{}=0$.
  	
                        Moreover, the Radon-Nikod\'{y}m property guarantees the existence of the
                        Radon-Nikod\'{y}m derivative $\frac{\mathrm{d}\mathfrak{m}}{\mathrm{d}\mu}\in
                        L^1\bigl(X,\mathrm{d}\mu;\mathscr{L}^1_+(\mathscr{H})\bigr)$, such that \beq
                        \mathfrak{m}(S)=\int_S^{}\mathrm{d}\mu(z) \frac{\mathrm{d}\mathfrak{m}}{\mathrm{d}\mu}(z)
                        \eeq for any measurable $S\in \Sigma$ and for any scalar measure $\mu$ such that
                        $\mathfrak{m}$ is absolutely continuous w.r.t. $ \mu $. 
                        
                        In our setting, compared to
                        the more general case of Banach-space-valued vector measures, there is an
                        additional notion of positivity, as discussed above. Such notion naturally singles
                        out a given scalar measure, w.r.t. which $\mathfrak{m}$ is absolutely
                        continuous. Such measure is the ``norm'' measure $m$ defined in \eqref{eq: norm
                          measure}. Indeed, combining the mutual absolute continuity of $ \mathfrak{m}$
                        and $m$ with the existence of the Radon-Nikod\'{y}m derivative, we deduce that,
                        for any measurable $ S \subset \mathfrak{h} $,
  		\begin{equation}
    			\mathfrak{m}(S)=\int_{S}^{}  \mathrm{d}m(z)\, \frac{\mathrm{d}\mathfrak{m}}{\mathrm{d}m}(z)\; ,
  		\end{equation}
  		and that, $m$-a.e.,
  		\begin{equation*}
   	 		\frac{\mathrm{d}\mathfrak{m}}{\mathrm{d}m}\neq 0\; .
  		\end{equation*}

  		Therefore, we can rewrite
  		\begin{equation}
    			\mathfrak{m}(S)= \int_{S}^{}\mathrm{d}m(z)  \:  \lf\| \frac{\mathrm{d}\mathfrak{m}}{\mathrm{d}m}  (z) \ri\|_{\mathscr{L}^1(\mathscr{H})}^{} \, \frac{\mathrm{d}\mathfrak{m}}{\mathrm{d}m}(z)\, \lf\| \frac{\mathrm{d}\mathfrak{m}}{\mathrm{d}m}  (z) \ri\|_{\mathscr{L}^1(\mathscr{H})}^{-1}\;,
 	 	\end{equation}
  		and setting
  		\beqn
    			\mathrm{d}\mu_{\mm}(z) &:= & \lf\| \frac{\mathrm{d}\mathfrak{m}}{\mathrm{d}m}  (z) \ri\|_{\mathscr{L}^1(\mathscr{H})}^{}  \mathrm{d}m(z)\;,\\
    \gamma_{\mm}(z)&:= &\frac{\mathrm{d}\mathfrak{m}}{\mathrm{d}m}(z)\,\lf\| \frac{\mathrm{d}\mathfrak{m}}{\mathrm{d}m}  (z) \ri\|_{\mathscr{L}^1(\mathscr{H})} ^{-1}\; ,
  		\eeqn
  		we obtain the sought Radon-Nikod\'{y}m decomposition.
	\end{proof}

Let now $ F:X\to \mathscr{B}(\mathscr{H})$ be a measurable function with respect to the weak-* topology on
$\mathscr{B}(\mathscr{H})$. It is then natural to define the ($\mathscr{L}^1$-Bochner) integrals of $f$
with respect to $\mm$ as follows: for any $S\in \Sigma$,
\beqn
    \int_{S}^{}\mathrm{d}\mm(z) F(z) &:= \disp\int_S^{}  \mathrm{d}\mu_{\mm}(z)\: \gamma_{\mm}(z) F(z) ,\\
    \int_{S}^{} F(z)\mathrm{d}\mm(z) &:= \disp\int_S^{}  \mathrm{d}\mu_{\mm}(z)\: F(z) \gamma_{\mm}(z).
\eeqn
Notice that one has to keep track of the order inside the integral, {\it i.e.}, putting the measure on the right or on the left of the integrand is not the same, because $ \gamma_{\mm} $ might not commute with $ F(z) $, since both are operators on $ \mathscr{H} $.

State-valued measures are important since they are the quasi-classical counterpart of quantum states
(see \cite{falconi2017arxiv} for a detailed discussion). Operator-valued symbols, such as the
aforementioned $\mathcal{F}$, are correspondingly the quasi-classical counterpart of quantum observables. From a
general point of view, we can summarize the main objective of quasi-classical analysis as follows: let
$\mathrm{Op}_{\varepsilon}(\mathcal{F})$ be a ``quantization'' of $\mathcal{F} $ acting on $\mathscr{H}\otimes \mathscr{K}_{\varepsilon}$, where the
space $\mathscr{K}_{\varepsilon}$ carries a semiclassical representation of the canonical commutation relations
corresponding to a symplectic space of test functions $(V, \sigma)$, and let $\Gamma_{\varepsilon}$ be a quantum state
converging to the Borel state-valued measure $\mm$ on the space $V'$ of suitably regular classical
fields, then we would like to prove that
\begin{equation}
  \label{eq:1}
  \begin{split}
    \lim_{\varepsilon\to 0}\tr_{\mathscr{K}_{\varepsilon}}\bigl(\Gamma_{\varepsilon}\mathrm{Op}_{\varepsilon}(\mathcal{F})\bigr) &= \int_{V'}^{}\mathrm{d}\mm(z) \: \mathcal{F}(z)\; ,\\
    \lim_{\varepsilon\to 0}\tr_{\mathscr{K}_{\varepsilon}}\bigl(\mathrm{Op}_{\varepsilon}(\mathcal{F})\Gamma_{\varepsilon}\bigr) &= \int_{V'}^{} \mathcal{F}(z) \: \mathrm{d}\mm(z)\; ;
  \end{split}
\end{equation}
where the convergence holds in a suitable topology of $\mathscr{L}^1(\mathscr{H})$. 

It is however difficult to obtain results such as the above for general symbols and quantum states. The most important obstruction is indeed the difficulty of define a proper quantization
procedure for symbols acting on infinite dimensional spaces. However, for the theories of particle-field
interaction under consideration (Nelson, polaron, Pauli-Fierz), the interaction terms in the
quasi-classical Hamiltonians contain only symbols of a specific form. We can therefore restrict our analysis
to such type of symbols.

Let us recall that we are considering the following concrete setting: $\mathscr{H}=L^2 (\mathbb{R}^{dN} )$, where
$d$ is the spatial dimension on which the particles move and $N$ is the number of quantum particles;
$\mathscr{K}_{\varepsilon}=\GG_{\eps}(\mathfrak{h})$, the symmetric Fock space over the complex separable
Hilbert space $\mathfrak{h}$, carrying the standard $\varepsilon$-dependent Fock representation of the canonical commutation relations
\begin{equation*}
  \lf[a_{\varepsilon}(z),a_{\varepsilon}^{\dagger}(\eta)\ri]=\varepsilon \braket{z}{\eta}_{\mathfrak{h}}\; .
\end{equation*}
Finally, we are interested in the case $V'=\mathfrak{h}$, \emph{i.e.}, the space of test functions
coincides with the space of classical fields. The type of symbols $ \mathcal{F} $ is given by the class defined in  \eqref{eq:5}, {\it i.e.}, 
\begin{equation*}
  \mathcal{F}(z) =\sum_{j=1}^{N} \braketr{z}{\lambda_1(\xv_j)}_{\mathfrak{h}} \dotsm \braketr{z}{\lambda_\ell(\xv_j)}_{\mathfrak{h}} \braketl{\lambda_{\ell+1}(\xv_j)}{z}_{\mathfrak{h}} \dotsm \braketl{\lambda_{\ell+m}(\xv_j)}{z}_{\mathfrak{h}}\;,
\end{equation*}
where the functions $\lambda_j\in L^{\infty}(\mathbb{R}^{d}; \mathfrak{h})$ for any $j\in \{1,\dotsc,\ell+m\}$ should be considered as fixed
``parameters'', and $ \mathcal{F}(z)$ acts as a multiplication operator on $L^2(\mathbb{R}^{dN})$.

Since $\mathcal{F}$ is a polynomial symbol with respect to $z$ and $\bar{z}$, it is natural to quantize it by the
Wick quantization rule. For such simple symbols the Wick rule has a very easy form: substitute each $z$
with $a_{\varepsilon}$ and each $\bar{z}$ with $a^{\dagger}_{\varepsilon}$, and then put the so obtained expression in normal order, by moving all the creation operators to the left of
the annihilation operators. Therefore, we obtain
\begin{equation}
  \mathrm{Op}_{\varepsilon}^{\mathrm{Wick}}(\mathcal{F})=\sum_{j=1}^N a^{\dagger}_{\varepsilon}\bigl(\lambda_1(\xv_j)\bigr)\dotsm a^{\dagger}_{\varepsilon}\bigl(\lambda_\ell(\xv_j)\bigr)a_{\varepsilon}\bigl(\lambda_{\ell+1}(\xv_j)\bigr)\dotsm a_{\varepsilon}\bigl(\lambda_{\ell+m}(\xv_j)\bigr)\; ,
\end{equation}
as a densely defined operator on $L^2(\mathbb{R}^{dN})\otimes \GG_{\eps}(\mathfrak{h})$.

In order to prove a weak convergence as in \cref{eq:1} for $T \mathrm{Op}_{\varepsilon}^{\mathrm{Wick}}(\mathcal{F}) S $, with $ \mathcal{S}, \mathcal{T} \in \mathscr{B}(\mathscr{H}) $, we need suitable hypotheses on the quantum state $\Gamma_{\varepsilon}$, and some preparatory results. The following condition ensures that all the quasi-classical Wigner measures corresponding to a state $\Gamma_{\varepsilon}\in \mathscr{L}^1_{+,1}\bigl(\mathscr{H} \otimes \mathscr{K}_{\eps} \bigr)$ are concentrated as Radon $L^2$-state-valued
probability measures on $\mathfrak{h}$. Recall the definition \eqref{eq: weyl} of the Weyl operator $ W_{\eps}(\eta) $, $ \eta \in \mathfrak{h} $, and the Fourier transform \eqref{eq: fourier transform} of a measure $ \mm \in \mathscr{M}(\mathfrak{h}; \mathscr{L}^1_{+}(\mathscr{H})) $. Recall that in this section and the rest of the paper we are going to consider only the convergence defined in \cref{def: qc convergence} and thus we simply write $ \Gamma_{\eps} \xrightarrow[\eps \to 0]{} \mm $ instead of $  \Gamma_{\eps} \xrightarrow[\eps \to 0]{\mathrm{qc}} \mm $, to simplify the notation.

	\begin{proposition}[Convergence of quantum to classical states]
  		\label{prop:8}
  		\mbox{}	\\
  		Let $ \Gamma_{\varepsilon}\in\mathscr{L}^1_{+,1}\bigl(\mathscr{H} \otimes \mathscr{K}_{\eps} \bigr) $ be such that there exist $\delta>0$, so that
  		\begin{equation}
    			\label{eq:15}
    			\Tr\bigl(\Gamma_{\varepsilon}{(\mathrm{d}\GG_{\varepsilon}(1)+1)}^{\delta} \bigr)\leq C\; .
 		 \end{equation}
                 Then, there exists at least one subsequence $ \lf\{\varepsilon_n\ri\}_{n \in \N} $ and a
                 $\mathscr{H}$-state-valued cylindrical measure $\mm$ (that may depend on the sequence)
                 such that \beq \Gamma_{\eps_n} \xrightarrow[n \to +\infty]{} \mm, \eeq in the sense of Definition
                 \ref{def: qc convergence}. Furthermore, all cluster points $ \mm $ of $\Gamma_{\varepsilon}$ are
                 state-valued Radon measures on $\mathfrak{h}$, and, for any $0\leq \delta'\leq \delta$, there exists $
                 C_{\delta'}\leq C_{\delta}$, with $C_0=1$, such that
                 \begin{equation}
                   \label{eq:19}
    			\lf\| \int_{\mathfrak{h}}^{}  \mathrm{d}\mu_{\mm}(z) \: \gamma_{\mm}(z) \lf( \lVert z  \rVert_{\mathfrak{h}}^2 +1 \ri)^{\delta'} \ri\|_{\mathscr{L}^1(\mathscr{H})}^{}= \int_{\mathfrak{h}}^{} \mathrm{d}\mu_{\mm}(z) \lf(\lVert z  \rVert_{\mathfrak{h}}^2+1 \ri)^{\delta'} \leq C_{\delta'} \; .
                      \end{equation}
                      If in addition $\Gamma_{\varepsilon}$ satisfies either Assumption \eqref{eq: ass1} or
                        \eqref{eq: ass2}, then $\mathfrak{m}$ is a probability measure, \emph{i.e.},
                      \begin{equation}
                        \label{eq:12}
                        \lf\| \int_{\mathfrak{h}}^{}  \mathrm{d}\mu_{\mm}(z) \: \gamma_{\mm}(z)  \ri\|_{\mathscr{L}^1(\mathscr{H})}^{}= \int_{\mathfrak{h}}^{} \mathrm{d}\mu_{\mm}(z) =1\; .
                      \end{equation}
	\end{proposition}
	
      In order to prove the last part of the above proposition, we need a couple of preparatory
          results, that will be useful as well in \cref{sec:quasi-class-analys-1}.
        	
	\begin{lemma}
  		\label{lemma:7}
  		\mbox{}	\\
  		Let $\mathcal{T}$ be a densely defined self-adjoint operator on $\mathscr{H}$, and $\mathds{1}_m(\mathcal{T})$	
  its spectral projection on the interval $[-m,m]$, $m\in \mathbb{N}$. Then, the set of operators
  		\beq
  			\label{eq: set K}
      		\mathfrak{K} := \{\mathcal{B}_m:=\mathds{1}_m(\mathcal{T}) \mathcal{B}\,\mathds{1}_m(\mathcal{T})\,,\, \mathcal{B}\in \mathscr{L}^{\infty}_+(\mathscr{H}),m\in \mathbb{N}\}\;,
   		\eeq
  		separates points in $\mathscr{L}^1_+(\mathscr{H})$ w.r.t. the weak-$*$ topology.
	\end{lemma}
	
	\begin{proof}
  		Let $\gamma \in
  \mathscr{L}^1_+(\mathscr{H})$ be such that, for all $\mathcal{B}_m\in \mathfrak{K}$,
  		\begin{equation*}
    			\tr_{\mathscr{H}}(\gamma \mathcal{B}_m)=0\; .
  		\end{equation*}
  		Let $\sum_j^{}\lambda_j \lvert \psi_j\rangle\langle \psi_j\rvert$ be the decomposition of $ \gamma $. Then, it follows that for all $j\in \mathbb{N}$,
  		\begin{equation*}
    			\sum_{j \in \N} \lambda_j \meanlrlr{\psi_j}{\mathcal{B}_m}{\psi_j}_{\mathscr{H}} = 0 \Longrightarrow  \meanlrlr{\psi_j}{\mathcal{B}_m}{\psi_j}_{\mathscr{H}} = 0,	\qquad		\forall j \in \N,
  		\end{equation*}
  		by positivity of $ \mathcal{B} $.
  		Taking the limit $m\to + \infty$ of the last equation, one obtains that for any $\mathcal{B}\in
  \mathscr{L}_+^{\infty}(\mathscr{H})$ and $j\in \mathbb{N}$,
  		\begin{equation}
    			\meanlrlr{\psi_j}{\mathcal{B}}{\psi_j}_{\mathscr{H}} = 0\;,
  		\end{equation}
  		but, taking in particular $ \mathcal{B} = \ket{\psi_j} \bra{\psi_j} $, we get $\psi_j=0$ for any $j\in \mathbb{N}$, and therefore $ \gamma =0$.
	\end{proof}
	
        \begin{proposition}[Convergence of general state sequences]
          \label{prop:6}
          \mbox{}	\\
        	Let $\Gamma_{\varepsilon}\in \mathscr{L}_{+,1}^1(\mathscr{H}\otimes \mathscr{K}_{\varepsilon})$ be such that:
          \begin{itemize}
          \item assumption \eqref{eq:20} is satisfied,
            
          \item $\Gamma_{\varepsilon_n}\xrightarrow[n\to + \infty]{} \mathfrak{m}$,
            
          \item $\bigl\lvert\tr_{\mathscr{H}} (\mathcal{T} \gamma_{\varepsilon}\mathcal{T})\bigr\rvert\leq C $ for some self-adjoint
            $\mathcal{T} \in \mathcal{L}(\mathscr{H}) $, where $\gamma_{\varepsilon}$ is  given by \eqref{eq: reduced}.
          \end{itemize}
          Then,
          \begin{equation}
            \label{eq:24}
            T\,\Gamma_{\varepsilon_n}\,T\xrightarrow[n\to +\infty]{} \mathcal{T}\,\mathfrak{m}\, \mathcal{T}\; ,
          \end{equation}
          where the latter is defined by the Radon-Nikod\'{y}m decomposition $
            \lf(\mu_{\mm}, \mathcal{T}\, \gamma_{\mm}(z)\, \mathcal{T} \ri)\; $.
        \end{proposition}
        
        \begin{proof}
          Since $\bigl\lvert\tr_{\mathscr{H}} (\mathcal{T} \gamma_{\varepsilon}\mathcal{T})\bigr\rvert\leq C$, $T\, \Gamma_{\varepsilon}\, T$ is a
            quasi-classical family of states and thus it follows that there exists a generalized subsequence
            $(\Gamma_{\varepsilon_{n_\alpha}})_{\alpha\in A}$ of $\Gamma_{\varepsilon_n}$ and a cylindrical state-valued measure $\mathfrak{n}$ such that
            (see \cite{falconi2017arxiv} for additional details):
          \begin{itemize}
          \item $T\,\Gamma_{\varepsilon_{n_\alpha}}\, T$ converges to $\mathfrak{n}$ when tested on the Weyl quantization
              of smooth cylindrical symbols,
            
          \item $\tr_{\mathscr{H}}\bigl(\mathcal{T}\,\widehat{\Gamma}_{\varepsilon_{n_\alpha}}(\eta)\, \mathcal{T}
              \mathcal{B}\bigr)$ converges to $\tr_{\mathscr{H}}\bigl(\gamma(\eta)\mathcal{B}\bigr)$ for all $\eta\in
              \mathfrak{h}$ and $\mathcal{B}\in \mathscr{L}^{\infty}(\mathscr{H})$, where $\gamma(\eta)\in
              \mathscr{L}^1(\mathscr{H})$ has yet to be determined.
          \end{itemize}
          Now, let $\mathcal{K}\in \mathfrak{K}$. Then, $\mathcal{T}\mathcal{K}\mathcal{T}\in
            \mathscr{L}^{\infty}(\mathscr{H})$ and therefore
          \begin{equation*}
            \lim_{\alpha\in A}\tr_{\mathscr{H}}\bigl(\mathcal{T}\,\widehat{\Gamma}_{\varepsilon_{n_\alpha}}(\eta)\, \mathcal{T}\mathcal{K}\bigr)=\lim_{\alpha\in A}\tr_{\mathscr{H}}\bigl(\widehat{\Gamma}_{\varepsilon_{n_\alpha}}(\eta)\, \mathcal{T}\mathcal{K}\mathcal{T}\bigr)= \tr_{\mathscr{H}}\bigl(\widehat{\mathfrak{m}}(\eta)\,\mathcal{T}\mathcal{K}\mathcal{T}\bigr)=\tr_{\mathscr{H}}\bigl(\mathcal{T}\,\widehat{\mathfrak{m}}(\eta)\,\mathcal{T}\mathcal{K}\bigr)\; .
          \end{equation*}
          However, the set $\mathfrak{K}$ separates points by \cref{lemma:7}, and therefore we can
            conclude that
          \begin{equation*}
            \gamma(\eta)= \mathcal{T}\,\widehat{\mathfrak{m}}(\eta)\,\mathcal{T}\; .
          \end{equation*}
          
          On the other hand, an analogous reasoning when testing with the Weyl quantization of smooth
            cylindrical symbols yields that
          \begin{equation*}
            \mathfrak{n}=\mathcal{T}\,\mathfrak{m}\,\mathcal{T}\; .
          \end{equation*}
          Therefore, we conclude that $T\,\Gamma_{\varepsilon_{n_{\alpha}}}\, T\,\xrightarrow[\alpha\in A]{}
            \mathcal{T}\,\mathfrak{m}\,\mathcal{T}$. Finally, let $\Gamma_{\varepsilon_{n_{\alpha'}}}$ be any generalized
            subsequence such that for any $\eta\in \mathfrak{h}$ and $\mathcal{B}\in
            \mathscr{L}^{\infty}(\mathscr{H})$
          \begin{equation*}
            \lim_{\alpha'\in A'}\tr_{\mathscr{H}}\bigl(\mathcal{T}\,\widehat{\Gamma}_{\varepsilon_{n_{\alpha'}}}(\eta)\,\mathcal{T}\,\mathcal{B}\bigr)=\tr_{\mathscr{H}}\bigl(\gamma'(\eta) \mathcal{B}\bigr)\; .
          \end{equation*}
          Then, repeating the above reasoning it follows that
            $\gamma'(\eta)=\mathcal{T}\,\widehat{\mathfrak{m}}(\eta)\,\mathcal{T}$. In other words, the cluster point is
            unique, and therefore $T\,\Gamma_{\varepsilon_n}\,T\xrightarrow[n\to +\infty]{} \mathcal{T}\,\mathfrak{m}\,
            \mathcal{T}$.
        \end{proof}
	
	\begin{proof}[Proof of \cref{prop:8}]
          The key result about the weak-$*$ convergence in the semiclassical case is proven in
          \cite[Theorem 6.2]{ammari2008ahp}. The generalization to the quasi-classical setting is trivial:
          for all compact operators $ \mathcal{B} \in \mathscr{L}^{\infty}(\mathscr{H})$ and all $\eta\in
          \mathfrak{h}$, one immediately gets
  		\begin{equation}
    			\label{eq:2}
    			\lim_{n\to \infty}\Tr\Bigl(\Gamma_{\varepsilon_n}\, W_{\eps_n}(\eta) \, B \Bigr)=\tr_{\mathscr{H}} \lf( \int_{\mathfrak{h}}^{} \mathrm{d}\mu_{\mm}(z) \: \gamma_{\mm}(z)\, e^{2i\Re \braket{\eta}{z}_{\mathfrak{h}}} \, \mathcal{B} \ri) = \tr_{\mathscr{H}} \lf( \widehat{\mm}(\eta)\, \mathcal{B} \ri) \; ,
  		\end{equation}
 		 where $ B:= \mathcal{B}\otimes 1$. Moreover, the Fourier transform $\widehat{\mm}:\mathfrak{h}\to
                 \mathscr{L}^1(L^2)$ identifies uniquely the measure $\mm$ by Bochner's theorem
                 \cite{falconi2017arxiv}. The bound \eqref{eq:19} is also an immediate extension of
                   \cite[Theorem 6.2]{ammari2008ahp} to the quasi-classical case.

                 It remains to prove that under either Assumption \eqref{eq: ass1} or \eqref{eq: ass2} $\mathfrak{m} \in \mathscr{P}(\mathfrak{h};\mathscr{L}_+^1(\mathscr{H}))$. Let us start assuming \eqref{eq:
                     ass1}. Then, by \cref{prop:6}, for any \emph{bounded} $\mathcal{B}\in
                   \mathscr{B}(\mathscr{H})$, and $\eta\in \mathfrak{h}$:
                 \begin{multline*}
                   	\lim_{n\to +\infty}\tr_{\mathscr{H}} \lf( \widehat{\Gamma}_{\varepsilon_n}(\eta) B \ri)=\lim_{n\to
                      +  \infty}\tr_{\mathscr{H}} \lf( A^{\frac{1}{2}}\,\widehat{\Gamma}_{\varepsilon_n}(\eta)\,A^{\frac{1}{2}} \:
                     A^{-\frac{1}{2}}B A^{-\frac{1}{2}}
                     \ri)\\
                  	=\tr_{\mathscr{H}} \lf(\mathcal{A}^{\frac{1}{2}}\, \widehat{\mathfrak{m}}(\eta)\,
                     \mathcal{A}^{\frac{1}{2}}\: \mathcal{A}^{-\frac{1}{2}}\mathcal{B}
                     \mathcal{A}^{-\frac{1}{2}} \ri)=\tr_{\mathscr{H}}\lf(\widehat{\mathfrak{m}}(\eta)\,
                     \mathcal{B} \ri)\; .
                 \end{multline*}
                 In particular, for $\eta=0$ and $\mathcal{B}=\mathds{1}$, 
                 \begin{equation*}
                  	1=\lim_{n\to +\infty}\tr_{\mathscr{H}} \lf( \gamma_{\varepsilon_n}\ri)=\lim_{n\to + \infty}\tr_{\mathscr{H}} \lf( \widehat{\Gamma}_{\varepsilon_n}(0)\ri)=\tr_{\mathscr{H}} \lf(\widehat{\mathfrak{m}}(0) \ri)=\tr_{\mathscr{H}} \lf(\mathfrak{m}(\mathfrak{h}) \ri)\; .
                 \end{equation*}
                 
                 If we instead assume \eqref{eq: ass2}, the proof goes as follows. Since
                   $\Gamma_{\varepsilon_n}\xrightarrow[n\to+ \infty]{}\mathfrak{m}$, it follows that $\gamma_{\varepsilon_n}$ converges in weak
                   operator topology to $\hat{\mathfrak{m}}(0)$, by compactness of rank-one operators. Let
                   $ \lf\{ m_j \ri\}_{j\in \mathbb{N}}$ be the eigenvalues of $\widehat{\mathfrak{m}}(0)$, and $ \lf\{ \psi_j \ri\}_{j\in \mathbb{N}}$ the
                   corresponding eigenvectors. By the aforementioned weak operator convergence, it follows
                   that for any $j\in \mathbb{N}$
                 \begin{equation}
                   	\lim_{n\to +\infty}  \meanlrlr{\psi_j }{\gamma_{\varepsilon_n}}{\psi_j} _{\mathscr{H}}=  \meanlrlr{\psi_j}{\widehat{\mathfrak{m}}(0)}{\psi_j}_{\mathscr{H}}=m_j\; .
                 \end{equation}
                 On the other hand, by \eqref{eq: ass2}, we can apply Lebesgue's dominated convergence
                   theorem to the series
                 \begin{equation*}
                   \sum_{j=0}^{+\infty} \meanlrlr{\psi_j}{\gamma_{\varepsilon_n}}{\psi_j}_{\mathscr{H}}\; ,
                 \end{equation*}
                 since $ \meanlrlr{\psi_j}{\gamma_{\varepsilon_n}}{\psi_j}_{\mathscr{H}}\leq \meanlrlr{\psi_j}{\gamma}{\psi_j}_{\mathscr{H}}$, and
                 \begin{equation*}
                  \tr_{\mathscr{H}} \lf(\gamma \ri)=\sum_{j=0}^{+\infty} \meanlrlr{\psi_j}{\gamma}{\psi_j}_{\mathscr{H}}<+\infty\; .
                 \end{equation*}
                Therefore,
                 \begin{equation*}
                  	1=\lim_{n\to +\infty}\sum_{j=0}^{+\infty} \meanlrlr{\psi_j}{ \gamma_{\varepsilon_n}}{\psi_j}_{\mathscr{H}}=\sum_{j=0}^{+\infty}\lim_{n\to +\infty}\meanlrlr{\psi_j}{ \gamma_{\varepsilon_n}}{\psi_j}_{\mathscr{H}}=\sum_{j=0}^{+\infty}m_j=\tr_{\mathscr{H}} \lf( \widehat{\mathfrak{m}}(0) \ri)\; .
                 \end{equation*}
	\end{proof}

It is clear that together with \cref{prop:8}, all the other results that hold in semiclassical analysis for
infinite dimensions can be adapted to quasi-classical analysis, considering the semiclassical symbols and
corresponding quantizations in tensor product with the identity acting on $ \mathscr{H} $, replacing Wigner scalar
measures with state-valued Wigner measures, and replacing convergence of the trace with
$\mathscr{L}^1(\mathscr{H})$-weak-$*$ convergence of the partial trace, \emph{i.e.}, one should test the partial
traces and integrals with compact operators. 

	\begin{proposition}[Convergence of expectation values]
  		\label{prop:10}
  		\mbox{}	\\
  		Let $ \mathcal{F} \in \mathscr{S}_{\ell,m}$, and let $\Gamma_{\varepsilon}\in \mathscr{L}^1_{+,1}(\mathscr{H} \otimes \mathscr{K}_{\eps})$. Assume that there exist $\delta> \frac{n+m}{2}$, such that
  		\begin{equation}
   			\label{eq:6}
    			\Tr \lf( \Gamma_{\varepsilon}{(\mathrm{d}\GG_{\varepsilon}(1)+1)}^{\delta} \ri) \leq C\; .
  		\end{equation}
  		Then, if $\Gamma_{\varepsilon_n} \xrightarrow[n \to + \infty]{} \mm$, for any $ \mathcal{S}, \mathcal{T} \in \mathscr{B}(\mathscr{H})$, $ \mathcal{B} \in \mathscr{L}^{\infty}(\mathscr{H}) $ and $\eta\in \mathfrak{h}$,
  		\bml{
    			\label{eq:7}
   			 \disp\lim_{n\to +\infty} \Tr \lf(\Gamma_{\varepsilon_n}\, T\mathrm{Op}_{\varepsilon_n}^{\mathrm{Wick}}(\mathcal{F})S\, \bigl(\mathcal{B}\otimes W_{\varepsilon_n}(\eta)\bigr) \ri) \\
   			 = \tr_{\mathscr{H}} \lf[ \disp\int_{\mathfrak{h}}^{}\mathrm{d}\mu_{\mm}(z) \: \gamma_{\mm}(z)\, \mathcal{T} \mathcal{F}(z) \mathcal{S}\, e^{2i\Re \braket{\eta}{z}_{\mathfrak{h}}}\,   \mathcal{B} \ri];
		}
		with analogous statement when the positions of $ \Gamma_{\varepsilon_n} $ and $ T\mathrm{Op}_{\varepsilon_n}^{\mathrm{Wick}}(\mathcal{F})S $ are exchanged.
	\end{proposition}
	
To prove \cref{prop:10}, we need the following preparatory lemma, which introduces the approximation of $ \mathcal{F} $ by simple functions.

	\begin{lemma}
  		\label{lemma:4}
  		\mbox{}	\\
  		Let $ \mathcal{F} \in \mathscr{S}_{\ell,m}$. Then, there exists a sequence of operator-valued functions $\lf\{\mathcal{F}_M \ri\}_{M\in \mathbb{N}}$,  $\mathcal{F}_M:\mathfrak{h}\to \mathscr{B}(\mathscr{H})$, such that
  		\begin{itemize}
  			\item for all $z\in \mathfrak{h}$,
    				\begin{equation}
     					 \lim_{M\to \infty} \lf\|  \mathcal{F}(z)-  \mathcal{F}_M(z) \ri\|_{\mathscr{H}}^{}=0\; ;
    				\end{equation}
    
  			\item $\mathcal{F}_M(z)$ acts as the multiplication operator by
    				\begin{equation}
  				    \mathcal{F}_M(z)= \sum_{j=1}^{N}\sum_{k=1}^{J(M)} \braketr{z}{\varphi_{k,1}}_{\mathfrak{h}} \dotsm \braketr{z}{\varphi_{k,\ell}}_{\mathfrak{h}} \braketl{\varphi_{k,\ell+1}}{z}_{\mathfrak{h}} \dotsm \braketl{ \varphi_{k,\ell+m}}{z}_{\mathfrak{h}} \one_{B_k}(\xv_j)\; ,
    				\end{equation}
    				where $J:\mathbb{N}\to \mathbb{N}$, $\varphi_{j,l} \in \mathfrak{h} $, $l\in \{1,\dotsc,\ell+m\}$, and $ \one_{B_j}$ is the characteristic function of the Borel set $B_j \subseteq \mathbb{R}^{d} $ and the $B_j$ are pairwise disjoint.
  		\end{itemize}
	\end{lemma}

	\begin{proof}
  It is sufficient to prove the convergence in the case $N=1,n=1, m=0$, since the case $N=1,n=0,m=1$ is
  perfectly analogous, and the general one $N\in \mathbb{N}$, $n\in \mathbb{N}$, $m\in \mathbb{N}$, can be obtained combining the
  approximation for each term of the product within each term of the sum and possibly reorder the sum.

  So let us restrict to the case $ \mathcal{F}(z)= \braket{z}{\lambda(\xv)}_{\mathfrak{h}}$, $\xv\in \mathbb{R}^d$, acting as a multiplication
  operator on $ \mathscr{H} = L^2 (\mathbb{R}^d )$. Since both $\mathcal{F}(z)$ and 
  \beq
  	\label{eq: lambda M}
  	\mathcal{F}_M(z)=\sum_{k=1}^{J(M)}  \braket{z}{\varphi_k}_{\mathfrak{h}} \one_{B_k}(\xv) = : \braket{z}{\lambda_M(\xv)}_{\mathfrak{h}}
	\eeq
	are multiplication operators, we have
  \begin{equation*}
    \lf\| \mathcal{F}(z)- \mathcal{F}_M(z) \ri\|_{\mathscr{L}(L^2(\R^d))}^{} = \underset{\xv \in \mathbb{R}^d}{\mathrm{ess}\, \sup} \lf| \braket{z}{\lambda(\xv) - \lambda_M(\xv)}_{\mathfrak{h}}  \ri|_{}^{}\;.
  \end{equation*}
  
  Now, let us fix $z\in \mathfrak{h}$ and consider $ \mathcal{F}(z) = \mathcal{F}_z(\xv) $ only as a function of  $ \xv \in \R^d $. We can decompose
  $ \mathcal{F}_z(\xv) = \mathcal{F}_{\mathbb{R}}(\xv)+i F_{\mathbb{C}}(\xv)$, and split both the real and imaginary part as
  $ \mathcal{F}_{\mathbb{R}/\mathbb{C}}(\xv)= \mathcal{F}_{\mathbb{R}/\mathbb{C},+}(\xv)-\mathcal{F}_{\mathbb{R}/\mathbb{C},-}(\xv)$. Setting $K:=\lVert \lambda \rVert_{L^{\infty}(\mathbb{R}^{d},\mathfrak{h})}^{}$, we can partition the real positive half-line as
  \begin{equation*}
    \mathbb{R}_+=A\cup \bigcup_{m=1}^{M }A_{m}\; ,
  \end{equation*}
  where 
  \beq
  	A : = \lf[K\lVert z \rVert_{\mathfrak{h}}^{},\infty \ri), \qquad	A_{m} :=K \lf[\tx\frac{m-1}{M}\lVert z
  \rVert_{\mathfrak{h}}^{}, \tx\frac{m}{M}\lVert z \rVert_{\mathfrak{h}}^{} \ri).
	\eeq 
	Let us now focus on the real positive part
  $\mathcal{F}_{\mathbb{R},+}(\xv) $: we can introduce the measurable sets
  \begin{equation*}
    D^+:=\mathcal{F}^{-1}_{\mathbb{R},+}(A), 	\qquad D^+_{m}:=\mathcal{F}^{-1}_{\mathbb{R},+}(A_{m})\; .
  \end{equation*}
  By construction, $D^+=\varnothing$, while , for all $m\in \{1,\dotsc,M\}$, there exists $\eta^+_{m}\in
  \mathfrak{h}$ such that 
  \bdm
  	\braketl{\eta^+_{m}}{z}_{\mathfrak{h}} \in A_{m}. 
	\edm
	For any given $\xv\in \mathbb{R}^{d}$, there is a
  single $\tilde{m}\in \{1,\dotsc,2^M\}$ such that $ \mathcal{F}_{\mathbb{R},+}(\xv)\in A^+_{\tilde{m}}$. Therefore, uniformly with
  respect to $ \xv \in D^+_{\tilde{m}}$,
  \begin{equation}
    \lf| \mathcal{F}_{\mathbb{R},+}(\xv)- \braketl{\eta^+_{m}}{z}_{\mathfrak{h}} \ri|_{}^{}<\frac{K\lVert z \rVert_{\mathfrak{h}}^{}}{M}\; .
  \end{equation}
  Repeating the same procedure for the real negative and complex positive and negative parts, we obtain
  collections of sets and elements, respectively
  \begin{equation*}
    {\{D^-_{m} \}}_{m=1}^{M}\;,\; {\{\eta^-_{m} \}}_{m=1}^{M}\; ;	\qquad  {\{E^{\pm}_{m} \}}_{m=1}^{M}\;,\; {\{\xi^{\pm}_{m} \}}_{m=1}^{M}\; ,
  \end{equation*}
  approximating $\mathcal{F}_{\mathbb{R},-}$ and $\mathcal{F}_{\mathbb{C},\pm}$.

  Let us now define the collection ${\{B_k\}}_{k=1}^{M^4}$ of
  disjoint Borel sets of $\mathbb{R}^d$ for the simple approximation of $\mathcal{F}(z)$. We first identify $k\in \{1,\dotsc,M^4\}$
  with the image $ \jmath (m_1,m_2,m_3,m_4)$ with respect to some fixed set bijection $\jmath: {\{1,\dotsc,M\}}^4\to
  \{1,\dotsc,M^{4}\}$, and then set
  \begin{equation}
    B_k:=D^+_{m_1}\cap D^-_{m_2}\cap E^+_{m_3}\cap E^-_{m_4}\; .
  \end{equation}
  Therefore, we define $ \varphi_k : = \eta^+_{m_1} - \eta^-_{m_2} + i(\xi^+_{m_3} - \xi^-_{m_4}) $ and
  \bml{
      \mathcal{F}_M(z)=\sum_{k=1}^{M^{4}} \braket{z}{\varphi_k}_{\mathfrak{h}}\, \one_{B_k}(\xv) \\
      =\sum_{m_1,m_2,m_3,m_4=1}^{M} \braketr{z}{\eta^+_{m_1} - \eta^-_{m_2} + i(\xi^+_{m_3} - \xi^-_{m_4})}_{\mathfrak{h}} \one_{B_{\jmath(m_1,m_2,m_3,m_4)}}(\xv)\;.
   }
  By construction,
  \begin{equation}
    \lVert \mathcal{F}(z)-  \mathcal{F}_M(z)\rVert_{\mathscr{L}(L^2(\R^d))}^{}\leq \frac{4K\lVert z  \rVert_{\mathfrak{h}}^{}}{M}\; ,
  \end{equation}
  and therefore the convergence is proved.
\end{proof}

	\begin{corollary}
  		\label{cor:3}
  		\mbox{}	\\
  		The approximating function $\mathcal{F}_M(z)$ can be rewritten as
  		\begin{equation}
    			\mathcal{F}_M(z)= \sum_{j=1}^{N} \braketr{z}{\lambda_{M,1}(\xv_j)}_{\mathfrak{h}} \dotsm \braketr{z}{\lambda_{M,\ell}(\xv_j)}_{\mathfrak{h}} \braketl{\lambda_{M,\ell+1}(\xv_j)}{z}_{\mathfrak{h}} \dotsm \braketl{\lambda_{M,\ell+m}(\xv_j)}{z}_{\mathfrak{h}}\;,
  		\end{equation}
  		where $\lambda_{M,j}\in L^{\infty}(\mathbb{R}^d;\mathfrak{h})$, $j\in \{1,\dotsc,\ell+m\}$, and
  		\begin{equation}
  			\label{eq: lambdaM converges}
   			\lim_{M\to + \infty} \lf\| \lambda_j - \lambda_{M,j}  \ri\|_{L^{\infty}(\mathbb{R}^d;\mathfrak{h})}^{}=0\; .
  		\end{equation}
	\end{corollary}

	\begin{proof}
  		Again, it is sufficient to prove the corollary for $N=1$ and $n=1,m=0$, the other cases being direct
  consequences. The function $\lambda_M$ approximating $\lambda$ is defined in \eqref{eq: lambda M} in the proof of \cref{lemma:4}, {\it i.e.},
  		\begin{equation*}
    			\lambda_M(\xv) := \sum_{k=1}^{J(M)}\varphi_k \one_{B_k}(\xv)\; .
  		\end{equation*}
  		From the same proof it also follows that, for all $z\in \mathfrak{h}$ and for all $ \xv\in \mathbb{R}^d$,
  		\begin{equation*}
    			\lf| \braketr{z}{\lambda(\xv)- \lambda_M(\xv)}_{\mathfrak{h}}  \ri|_{}^{}\leq \frac{4K\lVert z  \rVert_{\mathfrak{h}}^{}}{M}\; .
  		\end{equation*}
  		Therefore, it follows that
  		\begin{equation*}
    			\underset{\xv \in \mathbb{R}^d}{\mathrm{ess}\, \sup} \lf\| \lambda(\xv)- \lambda_M(\xv)   \ri\|_{\mathfrak{h}}^{}= \underset{\xv \in \mathbb{R}^d}{\mathrm{ess}\, \sup} \sup_{\lVert z  \rVert_{\mathfrak{h}}^{}=1} \lf| \braket{z}{\lambda(\xv)-\lambda_M(\xv)}_{\mathfrak{h}}  \ri|  \leq \frac{4K}{M}\;,
  		\end{equation*}
  		and thence the convergence is proved.
	\end{proof}

	\begin{proof}[Proof of \cref{prop:10}]
  		Let us prove~\eqref{eq:7}. Let us approximate $\mathcal{F}(z)$ with
  $\mathcal{F}_M(z)$, as dictated by \cref{lemma:4}. The advantage of $\mathcal{F}_M(z)$ is that its dependence on the $z$ and $\xv$
  variables is separated, and thus its Wick quantization is the finite sum of tensor products of operators:
  		\bml{
    			\label{eq:9}
      		\mathrm{Op}_{\varepsilon}^{\mathrm{Wick}}(\mathcal{F}_M)=\sum_{j=1}^N \sum_{k=1}^{J(M)} \one_{B_k}(\xv_j)\otimes a^{\dagger}_{\varepsilon}(\varphi_{k,1})\dotsm a^{\dagger}_{\varepsilon}(\varphi_{k,\ell})a_{\varepsilon}(\varphi_{k,\ell+1})\dotsm a_{\varepsilon}(\varphi_{k,\ell+m})\\
      		=\sum_{j=1}^N a^{\dagger}_{\varepsilon}(\lambda_{M,1}(\xv_j))\dotsm a^{\dagger}_{\varepsilon}(\lambda_{M,\ell}(\xv_j))a_{\varepsilon}(\lambda_{M,\ell+1}(\xv_j))\dotsm a_{\varepsilon}(\lambda_{M,\ell+m}(\xv_j)).
  		}
  		
  		Next we exploit the linearity of Wick quantization to split 
  		\bml{
    			\label{eq:10}
      		\Tr\Bigl(\Gamma_{\varepsilon_n}\, T\mathrm{Op}_{\varepsilon_n}^{\mathrm{Wick}}(\mathcal{F})S\bigl(\mathcal{B}\otimes W_{\varepsilon_n}(\eta)\bigr)\Bigr)=\Tr\Bigl(\Gamma_{\varepsilon_n}\, T\mathrm{Op}_{\varepsilon_n}^{\mathrm{Wick}}(\mathcal{F}-\mathcal{F}_M)S\bigl(\mathcal{B}\otimes W_{\varepsilon_n}(\eta)\bigr)\Bigr)\\+\Tr\Bigl(\Gamma_{\varepsilon_n}\, T\mathrm{Op}_{\varepsilon_n}^{\mathrm{Wick}}(\mathcal{F}_M)S\bigl(\mathcal{B}\otimes W_{\varepsilon_n}(\eta)\bigr)\Bigr).
  		}
  		The first term on the r.h.s.\ can be estimated using well-known estimates for creation
  and annihilation operators, the hypothesis on the expectation of the number operator, and \cref{cor:3}:
  		\bmln{
      		\Bigl\lvert \Tr\Bigl(\Gamma_{\varepsilon_n}\, T\mathrm{Op}_{\varepsilon_n}^{\mathrm{Wick}}(\mathcal{F}-\mathcal{F}_M)S\bigl(\mathcal{B}\otimes W_{\varepsilon_n}(\eta)\bigr)\Bigr)  \Bigr\rvert_{}^{}\leq N \lf\| \mathcal{B}  \ri\| \lf\| \mathcal{T} \ri\| \lf\| \mathcal{S}  \ri\| \Tr\Bigl(\Gamma_{\varepsilon_n}{\mathrm{d}\GG_{\varepsilon}(1)^{\frac{\ell+m}{2}}\Bigr)} \times \\ \times \sum_{p=1}^{\ell+m}\lVert \lambda_{1}  \rVert_{L^{\infty}(\mathbb{R}^d,\mathfrak{h})}^{}\dotsm\lVert \lambda_p - \lambda_{M,p}  \rVert_{L^{\infty}(\mathbb{R}^d,\mathfrak{h})}^{}\dotsm \lVert \lambda_{M,\ell+m}  \rVert_{L^{\infty}(\mathbb{R}^d,\mathfrak{h})}^{}\\
      \leq C N (\ell+m) \max_{p \in \{1, \ldots, \ell+m \}} \lf\| \lambda_p-\lambda_{M,p}  \ri\|_{L^{\infty}(\mathbb{R}^d; \mathfrak{h})},
    	}
    	where we have used that the $\lVert \lambda_{M,p}
  \rVert_{L^{\infty}(\mathbb{R}^d,\mathfrak{h})}^{}$ are all uniformly bounded with respect to $M$ by \eqref{eq: lambdaM converges}. 
  The r.h.s.\ of the above expression then converges to zero as $M\to +\infty$ by \cref{cor:3},
  uniformly w.r.t.\ $\varepsilon_n$.

  Let us now discuss the limit $n\to +\infty$ of the second term on the r.h.s.\ of \eqref{eq:10}: for any $ \mathcal{B} \in \mathscr{L}^{\infty}(L^2(\R^{Nd}))$, using the first identity of \eqref{eq:9}, we obtain
  \bmln{
      \Tr\Bigl(\Gamma_{\varepsilon_n}\, T\mathrm{Op}_{\varepsilon_n}^{\mathrm{Wick}}(\mathcal{F}_M)S\bigl(\mathcal{B}\otimes W_{\varepsilon_n}(\eta)\bigr)\Bigr) \\
      =\sum_{j=1}^N\sum_{k=1}^{J(M)}\tr_{\mathscr{H}}\Bigl( \tr_{\mathscr{K}_{\eps}}\Bigl(\Gamma_{\varepsilon_n}a^{\dagger}_{\varepsilon_n}(\varphi_{k,1})\dotsm a_{\varepsilon}(\varphi_{k,\ell+m})W_{\varepsilon_n}(\eta)  \Bigr) \one_{B_k}(\xv_j)\mathcal{S}\mathcal{B}\mathcal{T} \Bigr)\; .
    }
  Now, on one hand we know that $\Gamma_{\varepsilon_n}\to \mm$ by \cref{prop:8}, and on the other hand
  \begin{equation*}
    a^{\dagger}_{\varepsilon_n}(\varphi_{k,1})\dotsm a_{\varepsilon}(\varphi_{k,\ell+m})= \mathrm{Op}_{\varepsilon_n}^{\mathrm{Wick}} \lf( \braketr{z}{\varphi_{k,1}}_{\mathfrak{h}}\dotsm \braketl{\varphi_{k,\ell+m}}{z}_{\mathfrak{h}} \ri)\; ,
  \end{equation*}
  where the scalar symbol on the r.h.s.\ is polynomial and cylindrical. Therefore, since $
  \one_{B_k}(\xv_j) \mathcal{S} \mathcal{B} \mathcal{T} \in \mathscr{L}^{\infty}(L^2(\R^{Nd})) $, by the
  quasi-classical analogue of \cite[Theorem 6.13]{ammari2008ahp}, 
  \bmln{
    \lim_{n\to+ \infty}\Tr\Bigl(\Gamma_{\varepsilon_n}\, T\mathrm{Op}_{\varepsilon_n}^{\mathrm{Wick}}(\mathcal{F}_M)S\bigl(\mathcal{B}\otimes W_{\varepsilon_n}(\eta)\bigr)\Bigr)	\\
    =\sum_{j=1}^N\sum_{k=1}^{J(M)}\tr_{\mathscr{H}} \lf( \int_{\mathfrak{h}}^{}\mathrm{d}\mu_{\mm}(z) \: e^{2i\Re
      \braket{\eta}{z}_{\mathfrak{h}}} \braketr{z}{\varphi_{k,1}}_{\mathfrak{h}} \dotsm
    \braketl{\varphi_{k,\ell+m}}{z}_{\mathfrak{h}} \gamma_{\mm}(z) \one_{B_k}(\xv_j) \mathcal{S}\mathcal{B}\mathcal{T}
    \ri)\; .  } 
    The proof is then concluded by taking the limit $M\to \infty$ of the last expression, that by
    dominated convergence yields the sought result.
\end{proof}

\section{The Microscopic Model}
\label{sec:microscopic-model}

Our aim is to study systems of nonrelativistic particles in interaction with radiation. As discussed
previously, the techniques developed in this paper allow to study some well-known classes of explicit
models (Nelson, polaron, Pauli-Fierz). We carry out here the detailed analysis only for
the simplest example, the Nelson model, in order to convey the general strategy without too many technical
details. The main adaptations needed for the polaron and Pauli-Fierz systems are outlined in
\cref{sec:techn-modif-polar}.

Let $ \mathscr{H} \otimes \mathscr{K}_{\eps} = L^2(\mathbb{R}^{dN})\otimes \GG_{\eps}\bigl(L^2 (\mathbb{R}^d )\bigr)$ be the Hilbert space of the theory, then the Nelson Hamiltonian $H_{\varepsilon}$ is explicitly given by
\begin{equation}
  \label{eq:14}
  H_{\varepsilon}= K_0 + \nu(\varepsilon)\mathrm{d}\Gamma_{\varepsilon}\bigl(\omega\bigr) +\sum_{j=1}^N a^{\dagger}_{\varepsilon}\bigl(\lambda(\xv_j)\bigr) +a_{\varepsilon}\bigl(\lambda(\xv_j)\bigr)\; ,
\end{equation}
where $K_0=\mathcal{K}_0\otimes 1$ is the part of the Hamiltonian acting on the particles alone, being such that
$\mathcal{K}_0$ is self-adjoint on $\dom(\mathcal{K}_0) \subset L^2 (\mathbb{R}^{dN} )$, $\nu(\varepsilon)>0$ is a quasi-classical scaling factor to be
discussed in detail below, $\omega$ is the operator on $L^2(\mathbb{R}^d)$ acting as the multiplication by the positive
dispersion relation of the field $\omega(\kv)$, and $\lambda\in L^{\infty}\bigl(\mathbb{R}^d; L^2 (\mathbb{R}^d )\bigr)=: L^{\infty}_\xv L^2_\kv$ is the
interaction's form factor. In addition, let us define the set of vectors with a finite number of field's
excitations $C_0^{\infty}\bigl(\mathrm{d}\GG_{\varepsilon}(1)\bigr)$:
\bml{
      C_0^{\infty}\bigl(\mathrm{d}\GG_{\varepsilon}(1)\bigr)= \lf\{\psi\in L^2(\R^{Nd}) \otimes\GG_{\eps}(\mathfrak{h}) \; \big| \; \exists M\in \mathbb{N} \text{ s.t. } a_{\varepsilon}(f_1) \dotsm a_{\varepsilon}(f_{M'})\psi=0 \;, \ri. \\ \lf. \forall M'> M, \forall {\{f_j\}}_{j=1}^{M'}\subset L^2 (\mathbb{R}^d  ) \ri\}\; .
 }
 
 The question of self-adjointness of $ H_{\eps} $ has already been addressed in the literature and indeed the following proposition holds:

\begin{proposition}[Self-adjointness of $ H_{\eps} $ \mbox{\cite[Theorem 3.1]{falconi2015mpag}}]
  \label{prop:11}
  \mbox{}	\\
  The operator $H_{\varepsilon}$ is essentially self-adjoint on $D(K_0)\cap D(\mathrm{d}\GG_{\varepsilon}(\omega))\cap
  C_0^{\infty}\bigl(\mathrm{d}\GG_{\varepsilon}(1)\bigr)$.
\end{proposition}

Therefore, there exists a unitary evolution generated by $H_{\varepsilon}$,
\begin{equation}
  U_{\varepsilon}(t)=e^{-itH_{\varepsilon}}\; .
\end{equation}
Now for any  normalized density matrix $\Gamma_{\varepsilon}\in \mathscr{L}^1_{+,1}(\mathscr{H} \otimes \mathscr{K}_{\eps} )$, we denote by $\Gamma_{\varepsilon}(t)$ its unitary
evolution by means of $U_{\varepsilon}(t)$, {\it i.e.},
\begin{equation}
  \Gamma_{\varepsilon}(t)=U_{\varepsilon}(t)\Gamma_{\varepsilon}U^{\dagger}_{\varepsilon}(t)\; .
\end{equation}
The main aim of this paper
is to characterize the asymptotic behavior as $\varepsilon\to 0$ of
\begin{equation}
  \gamma_{\varepsilon}(t):=\tr_{\mathscr{K}_{\eps}}\bigl(\Gamma_{\varepsilon}(t)\bigr) = \tr_{\GG_{\eps}(L^2 (\mathbb{R}^d ))}\bigl(\Gamma_{\varepsilon}(t)\bigr)\; .
\end{equation}
As stated in \cref{def: qc convergence} and characterized in \cref{prop:10}, the quasi-classical limit of a sufficiently
regular state is determined by the weak convergence of its vector-valued noncommutative Fourier
transform $\widehat{\Gamma}_{\varepsilon}(t):\mathfrak{h}\to \mathscr{L}^1(\mathscr{H})$ defined in \eqref{eq: nc fourier transform}:
\begin{equation*}
  \eta\longmapsto \lf[\widehat{\Gamma}_{\varepsilon}(t)\ri](\eta):=\tr_{\mathscr{K}_{\eps}} \lf(\Gamma_{\varepsilon}(t) W_{\varepsilon}(\eta) \ri)=\tr_{\mathscr{K}_{\eps}}\lf(\Gamma_{\varepsilon}(t) e^{i(a^{\dagger}_{\varepsilon}(\eta)+a_{\varepsilon}(\eta))} \ri).
\end{equation*}
Note that consequently $\gamma_{\varepsilon}(t)=\bigl[\widehat{\Gamma}_{\varepsilon}(t)\bigr](0)$. 

The regularity of the state is given by~\eqref{eq:15}, that should thus be satisfied at any time. It is therefore necessary
to ensure a proper propagation in time of such a regularity. An estimate of that kind is however readily available for the Nelson
model with cutoff:

	\begin{proposition}[Regularity propagation \mbox{\cite[Proposition 4.2]{falconi2013jmp}}]
  		\label{prop:12}
  		\mbox{}	\\
  		For any $\varepsilon>0$, $t\in \mathbb{R}$ and $\delta\in \mathbb{R}$:
  		\beqn
      		\Tr \lf(\Gamma_{\varepsilon}(t){(\mathrm{d}\GG_{\varepsilon}(1)+N^2+\varepsilon)}^{\delta} \ri) \leq & e^{c_{\delta/2}(\varepsilon) \sqrt{\varepsilon}\lvert \delta  \rvert_{}^{}\lvert t  \rvert_{}^{} \lf\| \lambda \ri\|_{ L^{\infty}_\xv L^2_\kv}^{}}\,  \Tr \lf(\Gamma_{\varepsilon} {(\mathrm{d}\GG_{\varepsilon}(1)+N^2+\varepsilon)}^{\delta} \ri)\; ,		\label{eq: propagation 1}\\
      		\Tr \lf| \Gamma_{\varepsilon}(t){(\mathrm{d}\GG_{\varepsilon}(1)+N^2+\varepsilon)}^{\delta} \ri|  \leq & e^{c_{\delta}(\varepsilon)\sqrt{\varepsilon}\lvert \delta  \rvert_{}^{}\lvert t  \rvert_{}^{}\lf\| \lambda \ri\|_{ L^{\infty}_\xv L^2_\kv}^{}}\,  \Tr \lf| \Gamma_{\varepsilon} {(\mathrm{d}\GG_{\varepsilon}(1)+N^2+\varepsilon)}^{\delta} \ri|\; ,
    		\eeqn
  		where $ c_{\delta}(\varepsilon):=\max \lf\{ 2+\varepsilon,1+{(1+\varepsilon)}^{\delta} \ri\}$.
	\end{proposition}

Since the exponential in the above inequality is bounded uniformly with respect to $\varepsilon\in (0,1)$, it follows
that the bound~\eqref{eq:15} is satisfied by the state at any time with a suitable time-dependent
constant, provided it is satisfied by the state at $t=0$: using that, for any $ \delta \in \R^+ $,
\bdm	
	\lf(\mathrm{d}\GG_{\varepsilon}(1)+1 \ri)^{\delta} \leq \lf(\mathrm{d}\GG_{\varepsilon}(1)+N^2 + \eps \ri)^{\delta} \leq \lf(\mathrm{d}\GG_{\varepsilon}(1)+N^2 + 1 \ri)^{\delta} \leq (N^2+1)^{\delta} \lf(\mathrm{d}\GG_{\varepsilon}(1)+ 1 \ri)^{\delta},
\edm
we can use \eqref{eq: propagation 1} to obtain
\bml{
	\label{eq: propagation}
  \Tr\lf(\Gamma_{\varepsilon}{(\mathrm{d}\GG_{\varepsilon}(1)+1)}^{\delta} \ri) \leq C_{\delta}\; \Longrightarrow \; \Tr \lf(\Gamma_{\varepsilon}(t){(\mathrm{d}\GG_{\varepsilon}(1)+1)}^{\delta} \ri)\leq C_{\delta} \, (N^2 + 1)^{\delta} e^{c_{\delta/2}(1)\lvert \delta  \rvert_{}^{}\lvert t  \rvert_{}^{}\lVert \lambda \rVert_{ L^{\infty}_\xv L^2_\kv}^{}}	\\
  \leq C(\delta, t), \; 
}
which guarantees that the a priori bound \eqref{eq:15} is preserved by the time evolution.

In analogy with the dynamical semiclassical limit for bosonic field theories (see,
\emph{e.g.}, \cite{ammari2014jsp,ammari2015asns,ammari2017sima}), the quasi-classical dynamics is
characterized studying the limit $\varepsilon\to 0$ of the integral equation of evolution for the microscopic
system. Let us sketch the main ideas: consider the family of states
\begin{equation*}
  \lf\{\Gamma_{\varepsilon}(t) \ri\}_{\varepsilon\in (0,1), \: t\in \mathbb{R}}\;
\end{equation*}
at time $  t = 0 $, satisfying the bound~\eqref{eq:15}. Then, we know that for each fixed $t\in \mathbb{R}$, there exists a
subsequence $\varepsilon_n\to 0$ such that $\Gamma_{\varepsilon_n}(t) \xrightarrow[n \to + \infty]{} \mm_t$ in the sense of \cref{def: qc convergence} by \cref{prop:12,prop:8}. In the next section we prove
that it is actually possible to extract a \emph{common subsequence} $\varepsilon_{n_k} \to 0$ such that \emph{for all} $t\in \mathbb{R}$,
\begin{equation*}
  \Gamma_{\varepsilon_{n_k}}(t) \xrightarrow[k \to +\infty]{} \mm_t\; .
\end{equation*}

Hence one only needs to characterize the map $t\mapsto \mm_t$, and this is done by studying the associated transport equation, that is obtained passing to the limit in the microscopic integral equation of
evolution. Let us provide some intuition on such strategy. For later convenience let us pass to the interaction representation and set
\beq
	\widetilde{\Gamma}_{\eps}(t) : = e^{-i \nu(\eps) t \diff \GG_{\eps}(\omega)} \Gamma_{\eps}(t) e^{i \nu(\eps) t \diff \GG_{\eps}(\omega)}.
\eeq
Then, the microscopic evolution can be rewritten as an integral equation, using Duhamel's formula:
\begin{equation}
  \label{eq:31}
  \widetilde{\Gamma}_{\varepsilon}(t)=\Gamma_{\varepsilon}-i \int_0^t \mathrm{d}\tau \: \lf[\widetilde{H}_{\varepsilon}(\tau),\widetilde\Gamma_{\varepsilon}(\tau) \ri] \; ,
\end{equation}
where accordingly
\beq
	\widetilde{H}_{\varepsilon}(t) : = e^{-i \nu(\eps) t \diff \GG_{\eps}(\omega)} \lf( H_{\eps} - \nu(\eps) \diff \GG_{\eps}(\omega) \ri) e^{i \nu(\eps) t \diff \GG_{\eps}(\omega)}.
\eeq
In addition, $H_{\varepsilon}- \nu(\eps) \diff \GG_{\eps}(\omega) $ is the Wick quantization of an operator-valued symbol
$\mathcal{K}_0+\mathcal{V}(z)$. Therefore, the quasi-classical analysis developed in
\cref{sec:quasi-class-analys} suggests that the integral equation \eqref{eq:31} shall converge, in the
limit $\varepsilon\to 0$, to an equation for the measure $\widetilde\mm_t$, obtained by replacing 
\bdm
	\widetilde{\Gamma}_{\eps}(t) \rightsquigarrow \widetilde\mm_t,	\qquad		H_{\eps}- \nu(\eps) \diff \GG_{\eps}(\omega) \rightsquigarrow \mathcal{K}_0+\mathcal{V}(z),
\edm
and substituting the quantum flow $  e^{- i \nu(\eps) t \diff \GG_{\eps}(\omega)} $ in the phase space $ \mathfrak{h} = L^2(\R^3) $ by its classical counterpart, {\it i.e.},
\beq
	\label{eq: flow}
	z \mapsto e^{- i \nu t \omega} z,	\qquad \forall z \in L^2(\R^3).
\eeq
In conclusion, we get the equation
\begin{equation}
  \label{eq:4}
  \mathrm{d}\widetilde{\mathfrak{m}}_t(z)= \mathrm{d}\widetilde{\mathfrak{m}}(z)-i\int_0^t\mathrm{d}\tau \: \mathrm{d}\mu_{\widetilde{\mm}_{\tau}}(z) \: \lf[\mathcal{K}_0+\mathcal{V}\lf(e^{-i \nu t \omega} z\ri),\gamma_{\widetilde{\mm}_{\tau}}(z)\ri]  \; ,
\end{equation}
and the classical measure $ \mm_t $ associated with the original state $ \Gamma_{\eps}(t) $ is simply given by the {\it push-forward} of $ \widetilde{\mm}_t $ through the flow \eqref{eq: flow}, {\it i.e.},
\beq
	{\Gamma}_{\eps}(t) \rightsquigarrow \mm_t = e^{- i \nu t \omega} \, _{\star} \,\widetilde\mm_t.
\eeq

Such an equation is the integral form of a Liouville-type equation. Once the convergence of the microscopic
to the quasi-classical integral equation has been established (see \cref{sec:quasi-class-limit}), the
crucial point is to prove that equation \eqref{eq:4} has a unique solution that satisfies some properties,
given by the \emph{a priori} information that we have on the quasi-classical measure (see
\cref{sec:uniq-quasi-class}).  As a final step (\cref{sec:putting-it-all}), we show that the convergence is in fact at any time $ t \geq 0 $ along the {\it same} subsequence $ \lf\{ \eps_n \ri\}_{n \in \N} $. Let us remark that, in order to make this heuristic strategy rigorous, some
technical modifications are necessary, in particular it is necessary to pass to the full interaction
representation.

We conclude the section with the rigorous derivation of the microscopic integral
evolution equation for the Fourier transform of $ \Gamma_{\varepsilon}(t)$. By definition, for any $\eta\in \mathfrak{h}$, the Fourier transform
$\Bigl[\widehat{\Gamma}_{\varepsilon}(t)\Bigr](\eta)$ is a reduced microscopic complex state for the particles, and therefore, if $\Gamma_{\varepsilon}$ is regular enough, its time evolution can be described by means of the microscopic generator $H_{\varepsilon}$. It is technically convenient to use the evolved state in interaction picture, {\it i.e.},
\begin{equation}
	\label{eq: interaction picture}
 	\gint(t) := e^{it (K_0+\nu(\varepsilon)\mathrm{d}\GG_{\varepsilon}(\omega))}\Gamma_{\varepsilon}(t)e^{-it (K_0+\nu(\varepsilon)\mathrm{d}\GG_{\varepsilon}(\omega))}\; ,
\end{equation}
in place of $\Gamma_{\varepsilon}(t)$, and therefore study the integral equation for $\hgint(t)$. 

	\begin{remark}[Regularity propagation for $ \gint $]
		\label{rem: regularity gint}
		\mbox{}	\\
		Thanks to the commutativity of $e^{it(K_0+\nu(\varepsilon)\mathrm{d}\GG_{\varepsilon}(\omega))}$ with $\mathrm{d}\GG_{\varepsilon}(1)$, one can easily realize that the results stated in \cref{prop:12} and, consequently the bound propagation in \eqref{eq: propagation}, hold true also for the density matrix $ \gint(t) $ in the interaction picture with the same constants.
	\end{remark}

	\begin{lemma}
  		\label{lemma:5}
  		\mbox{}	\\
  		Let $\Gamma_{\varepsilon} \in \mathscr{L}^1_{+,1}(\mathscr{H} \otimes \mathscr{K}_{\eps}) $ be such that
  		\bdm
      		\Tr \lf( \Gamma_{\varepsilon}{(\mathrm{d}\GG_{\varepsilon}(1)+1)}^{\frac{1}{2}} \ri) \leq C\; ,
    		\edm
  		then, for any $ s,t \in \R $, 
  		\bml{
  			\label{eq: micro transport}
    			\lf[\hgint(t)\ri](\eta)=\lf[\hgint(s)\ri](\eta) \\
    			-i\sum_{j=1}^N\int_s^t \mathrm{d}\tau \: e^{i\tau \mathcal{K}_0}\,\tr_{\mathscr{K}_{\varepsilon}}\lf( \lf[ \varphi_{\varepsilon}\lf(e^{-i\tau\varepsilon\nu(\varepsilon)\omega}\lambda(\xv_j) \ri),\gint(\tau) \ri]\,W_{\varepsilon}(\eta) \ri)\,e^{-i\tau \mathcal{K}_0} \; ;
  		}
 		 weakly in $\mathscr{L}^1(L^2(\R^{Nd})$, where $\varphi_{\varepsilon}(\: \cdot \: )=a^{\dagger}_{\varepsilon}(\: \cdot \:)+a_{\varepsilon}(\: \cdot \:)$ is the  Segal field.
	\end{lemma}
\begin{proof}
  The proof is obtained adapting \cite[Proposition 3.5]{ammari2014jsp}. The differences here are only
  the presence of an arbitrary bounded particle observable, and that the Weyl operator acts only on the
  field's degrees of freedom. Therefore, we omit the details.
\end{proof}

\section{The Quasi-Classical Limit of Time Evolved States}
\label{sec:quasi-class-limit}

	In this section we focus on the quasi-classical limit $\varepsilon\to 0$ of the Fourier transform $ \hgint(t)$ of
time evolved states in the interaction picture. The first and most relevant step is the proof that it is possible to extract a common subsequence for the convergence of
$ \hgint(t)$ at any time (\cref{prop:14}), which in turn follows from the uniform equicontinuity of $ \hgint $ (\cref{prop:13}). Finally, we show that the limit measure satisfies the transport equation (\cref{prop:15}) of \cref{lemma:5}.

	Let us start with a preparatory lemma.

	\begin{lemma}[\mbox{\cite[Lemma 3.1]{ammari2008ahp}}]
  		\label{lemma:6}
  		\mbox{}	\\
  		For any $0<\delta\leq 1/2$, there exists a finite constant $c_{\delta} $ such that for all $\eta,\xi\in \mathfrak{h}$,
  		\begin{equation}
    			\lf\| \bigl(W_{\varepsilon}(\eta) -W_{\varepsilon}(\xi) \bigr){(\mathrm{d}\GG_{\varepsilon}(1)+1)}^{-\delta}  \ri\|_{\mathscr{B}(\mathscr{K}_{\eps})}^{} \leq c_{\delta} \lf( \min \lf\{ \lVert \eta  \rVert_{\mathfrak{h}}^{2\delta}, \lVert \xi  \rVert_{\mathfrak{h}}^{2\delta} \ri\}+1 \ri) \lf\| \eta-\xi  \ri\|_{\mathfrak{h}}^{2\delta}\; .
  		\end{equation}
	\end{lemma}

We are now able to prove uniform equicontinuity of $\lf[\hgint(\: \cdot \:) \ri](\: \cdot \:)$.

	\begin{proposition}[Equicontinuity of $ \hgint $]
  		\label{prop:13}
  		\mbox{}	\\
	  		Let $\Gamma_{\varepsilon}\in \mathscr{L}^1_{+,1}(\mathscr{H} \otimes \mathscr{K}_{\eps})$ be such that 
 	 	\begin{equation*}
    			\Tr \lf( \Gamma_{\varepsilon}{(\mathrm{d}\GG_{\varepsilon}(1)+1)}^{\frac{1}{2}} \ri) \leq C\;.
  		\end{equation*}
  		Then, $ \lf[ \hgint(\: \cdot \:)\ri](\: \cdot \:): \mathbb{R}\times \mathfrak{h}\to \mathscr{L}^1(\mathscr{H})$ is uniformly
  equicontinuous w.r.t. $\varepsilon\in (0,1)$ on bounded sets of $\mathbb{R}\times \mathfrak{h}$, if we endow
  $ \mathscr{L}^1(\mathscr{H})$ with the weak-$*$ topology.
	\end{proposition}

	\begin{proof}
  		Let us fix
  $\mathcal{B}\in \mathscr{L}^{\infty}(\mathscr{H})$, and $(t,\eta),(s,\xi)\in \mathbb{R}\times \mathfrak{h}$, with $ 0 \leq s \leq t $. Then,
  		\bml{
      		\lf| \tr_{\mathscr{H}} \lf[ \lf( \lf[ \hgint(t) \ri](\eta) - \lf[ \hgint(s) \ri](\xi) \ri) \mathcal{B} \ri] \ri| \leq \lf| \tr_{\mathscr{H}} \lf[ \lf( \lf[ \hgint(t) \ri](\eta) - \lf[ \hgint(s) \ri](\eta) \ri) \mathcal{B} \ri] \ri| \\
      		+ \lf| \tr_{\mathscr{H}} \lf[ \lf( \lf[ \hgint(s) \ri](\eta) - \lf[ \hgint(s) \ri](\xi) \ri) \mathcal{B} \ri] \ri| 
      		=:   (I) + (II) \; .
    		}
    		
    		Let us consider the two terms separately. Making use of \cref{lemma:5}, we obtain
  		\begin{equation*}
   			(I) \leq \sum_{j=1}^N  \int_s^t \mathrm{d}\tau \:  \lf| \Tr \lf( \lf[ \varphi_{\varepsilon}\bigl(\lambda(\xv_j)\bigr) , \Gamma_{\varepsilon}(\tau) \ri] \Bigl(\widetilde{\mathcal{B}}(\tau)\otimes W_{\varepsilon}\lf(e^{-i\tau\varepsilon\nu(\varepsilon)\omega}\eta \ri)\Bigr) \ri)  \ri| \; ,
  		\end{equation*}
  		where $\widetilde{\mathcal{B}}(\tau) : =e^{-i\tau \mathcal{K}_0}\mathcal{B} e^{i\tau \mathcal{K}_0}$. Therefore, 
  		\bmln{
      		(I) \leq 2N \lf\lVert \mathcal{B}  \ri\rVert \lf\| {(\mathrm{d}\GG_{\varepsilon}(1)+1)}^{-\frac{1}{4}}\varphi_{\varepsilon}\bigl(\lambda(\: \cdot \:)\bigr){(\mathrm{d}\GG_{\varepsilon}(1)+1)}^{-\frac{1}{4}}  \ri\|_{\mathscr{B}(\mathscr{H} \otimes \mathscr{K}_{\eps})}^{} \times	\\
      		\times  \int_s^t\mathrm{d}\tau \: \lf\| {(\mathrm{d}\GG_{\varepsilon}(1)+1)}^{\frac{1}{4}} W_{\varepsilon}\lf(e^{-i\tau\varepsilon\nu(\varepsilon)\omega}\eta \ri){(\mathrm{d}\GG_{\varepsilon}(1)+1)}^{-\frac{1}{4}}  \ri\|_{\mathscr{B}(\mathscr{K}_{\eps})} \Tr \lf( \Gamma_{\varepsilon}(\tau){(\mathrm{d}\GG_{\varepsilon}(1)+1)}^{\frac{1}{2}}  \ri) 
		}
		where we have used the identity
		\beq
			\lf\| (\mathrm{d}\GG_{\varepsilon}(1)+1)^{-\frac{1}{4}} \Gamma_{\eps} {(\mathrm{d}\GG_{\varepsilon}(1)+1)}^{-\frac{1}{4}} \ri\|_{\mathscr{L}^1(\mathscr{H})} = \Tr \lf( \Gamma_{\varepsilon}(\tau){(\mathrm{d}\GG_{\varepsilon}(1)+1)}^{\frac{1}{2}}  \ri).
		\eeq
		Next, we apply \cite[Corollary 6.2 (ii)]{ammari2014jsp} and \eqref{eq: propagation}, which follows from \cref{prop:12}, to deduce
		\bml{
			(I) \leq C C_{\frac{1}{2}} N{(N^2+1)}^{\frac{1}{2}} \lf\| \lambda  \ri\|_{L^{\infty}_\xv L^2_\kv}^{} \lf\| \mathcal{B} \ri\| \int_s^t \mathrm{d}\tau \: \exp \lf\{ 
      		\tx\frac{3}{2} \lf\| \lambda  \ri\|_{L^{\infty}_\xv L^2_\kv}^{} |\tau| \ri\} \\
      \leq C \lf[ \exp \lf\{ \tx\frac{3}{2} \lf\| \lambda  \ri\|_{L^{\infty}_\xv L^2_\kv}^{} t \ri\} - \exp \lf\{ \tx\frac{3}{2} \lf\| \lambda  \ri\|_{L^{\infty}_\xv L^2_\kv}^{} s \ri\} \ri| \; .
    		}
    	
  		The second term $(II) $ is bounded using again \cref{prop:12} and \cref{lemma:6}, and the fact that
  $e^{it(K_0+\nu(\varepsilon)\mathrm{d}\GG_{\varepsilon}(\omega))}$ commutes with $\mathrm{d}\GG_{\varepsilon}(1)$:
  		\bmln{
      		(II) \leq \lVert \mathcal{B}  \rVert  \lf\| {(\mathrm{d}\GG_{\varepsilon}(1)+1)}^{-\frac{1}{4}}\bigl(W_{\varepsilon}(\eta) -W_{\varepsilon}(\xi) \bigr){(\mathrm{d}\GG_{\varepsilon}(1)+1)}^{-\frac{1}{4}} \ri\| \Tr \lf( \gint(s){(\mathrm{d}\GG_{\varepsilon}(1)+1)}^{\frac{1}{2}} \ri) \\
      			\leq  \lVert \mathcal{B}  \rVert  \lf\| {(\mathrm{d}\GG_{\varepsilon}(1)+1)}^{-\frac{1}{4}}\bigl(W_{\varepsilon}(\eta) -W_{\varepsilon}(\xi) \bigr){(\mathrm{d}\GG_{\varepsilon}(1)+1)}^{-\frac{1}{4}} \ri\| \Tr \lf( \Gamma_{\eps}(s){(\mathrm{d}\GG_{\varepsilon}(1)+1)}^{\frac{1}{2}} \ri) \\
      			\leq c_{\frac{1}{4}}C_{\frac{1}{2}}{(N^2+1)}^{\frac{1}{2}}\lVert \mathcal{B}  \rVert \lf(\min \lf\{ \lVert \eta  \rVert^{\frac{1}{2}}_{\mathfrak{h}},\lVert \xi  \rVert^{\frac{1}{2}}_{\mathfrak{h}} \ri\}+1 \ri)  e^{\frac{3}{2}\lVert \lambda  \rVert_{L^{\infty}_\xv L^2_\kv}^{}\lvert s  \rvert_{}^{}} \lf\| \eta-\xi  \ri\|_{\mathfrak{h}}^{^{\frac{1}{2}}}\; .
 		}
  		This concludes the proof.
\end{proof}

By means of \cref{prop:13}, we are now in a position to prove the existence of a common subsequence,
convergent for all times.

	\begin{proposition}[Existence of a converging subsequence]
  		\label{prop:14}
  		\mbox{}	\\
  		Let $\Gamma_{\varepsilon}\in \mathscr{L}^1_{+,1}(\mathscr{H} \otimes \mathscr{K}_{\eps})$ be such that there exists $\delta\geq \frac{1}{2}$ so
  that
  		\begin{equation*}
   			 \Tr \lf( \Gamma_{\varepsilon}{(\mathrm{d}\GG_{\varepsilon}(1)+1)}^{\delta}  \ri)_{}^{}\leq C\;.
  		\end{equation*}
 		Then, for any sequence $ \lf\{ \varepsilon_n \ri\}_{n \in \N} $, with $\varepsilon_n \to 0$, there exists a subsequence $ \lf\{ \varepsilon_{n_k} \ri\}_{k \in \N}$, with $ \varepsilon_{n_k} \to 0$, and a family of state-valued
  probability measures $ \lf\{ \nn_t \ri\}_{t \in \R} $ indexed by time, such that for all $t\in \mathbb{R}$,
  		\begin{equation}
  			\label{eq: existence subsequence}
    			\Upsilon_{\varepsilon_{n_k}}(t) \xrightarrow[k \to + \infty]{} \nn_t \;.
  		\end{equation}
 		Furthermore, for any $T>0$, there exists $ C(T)>0$, such that, for any $t\in [-T,T]$ and for any $\delta' \leq \delta$,
  		\begin{equation}
  			\label{eq: a priori bound common}
    			\int_{\mathfrak{h}}^{} \mathrm{d} \mu_{\nn_t}(z) \: {\bigl(\lVert z  \rVert_{\mathfrak{h}}^2+1\bigr)}^{\delta'} \leq C(T)\; .
  		\end{equation}
	\end{proposition}
	
	\begin{proof}
  		Let $E:={\{t_j\}}_{j\in \mathbb{N}}\subset \mathbb{R}$ be a dense countable subset of $ \R $, and let $\varepsilon_n\to 0$. Using a
  diagonal extraction argument, and \cref{prop:12,prop:8} (see also \cref{rem: regularity gint} and \eqref{eq: propagation}), there exists a subsequence $\varepsilon_{n_k}\to 0$ such that for all
  $t_j\in E$:
  		\begin{equation*}
    			\Upsilon_{\varepsilon_{n_k}}(t_j) \xrightarrow[k \to + \infty]{} \nn_{t_j}\; .
  		\end{equation*}
  		In addition, since $ \bigl\lVert \bigl[\widehat{\Upsilon}_{\varepsilon}(t_j)\bigr](\eta) \bigr\rVert_{\mathscr{L}^1(\mathscr{H})}^{}\leq
  1$, for any $\eta\in \mathfrak{h}$ and $t_j\in E$, it follows that $ \lVert \widehat{\nn}_{t_j}(\eta)
  \rVert_{\mathscr{L}^1(\mathscr{H})}^{}\leq 1$, by Banach-Alaoglu's theorem. Furthermore, by \cref{prop:13}, for any $t_j,t_\ell \in E $  and for any $\mathcal{B}\in \mathscr{L}^{\infty}(\mathscr{H})$,
  		\begin{equation*}
    			\lf| \tr_{\mathscr{H}} \lf[ \lf( \lf[\widehat{\Upsilon}_{\varepsilon_{n_k}}(t_j)-\widehat{\Upsilon}_{\varepsilon_{n_k}}(t_\ell) \ri](\eta) \ri) \mathcal{B} \ri]  \ri|_{}^{}\leq C \lf| e^{c \lvert t_j  \rvert_{}^{}}-e^{c \lvert t_\ell  \rvert_{}^{}}  \ri|\; ,
  		\end{equation*}
  		where the constants on the r.h.s. are independent of $ k $. Therefore, we can take the limit $k \to + \infty $ of the
  above inequality obtaining that, for all $\mathcal{B}\in \mathscr{L}^{\infty}(\mathscr{H})$,
  		\begin{equation}
    			\lf| \tr_{\mathscr{H}} \lf[ \lf( \widehat{\nn}_{t_j}(\eta) -\widehat{\nn}_{t_\ell}(\eta) \ri) \mathcal{B} \ri]  \ri|_{}^{}\leq C \lf| e^{c \lvert t_j  \rvert_{}^{}}-e^{c \lvert t_\ell  \rvert_{}^{}}  \ri|\;.
  		\end{equation}
  
  		Now, let $t\in \mathbb{R}$ be arbitrary. By density of $E\subset \mathbb{R}$, there exists a sequence $ \lf\{t_j \ri\}_{j\in \mathbb{N}}$ of times in
  $E$, such that $t_j\to t$. It follows that, for any $\eta\in \mathfrak{h}$,
  $ \lf\{\widehat{\nn}_{t_j}(\eta) \ri\}_{j\in \mathbb{N}}$ is a weak-$*$ Cauchy sequence in the ultraweakly
  compact unit ball of the uniform space $\mathscr{L}^1(\mathscr{H})$. Thus, it converges when $t_j\to t$. Hence, we define
  		\begin{equation}
    			\widehat{\nn}_t(\eta):= \underset{j \to + \infty}{\mathrm{w} \lim} \: \widehat{\nn}_{t_j}(\eta)\; ,
  		\end{equation}
  		where the limit is meant in the weak-$*$ topology. For any $t\in \mathbb{R}$, $ \eta \mapsto \widehat{\nn}_t(\eta)$ is
  an ultraweakly continuous function such that:
 		\begin{itemize}
  			\item $\lf\| \widehat{\nn}_t(0)  \ri\|_{\mathscr{L}^1(\mathscr{H})}^{}=1$;
 			 \item $\eta\mapsto \widehat{\nn}_t(\eta)$ is a function of completely positive type (see, {\it e.g.}, \cite[Definition
    A.7]{falconi2017arxiv}).
  		\end{itemize}
  		Therefore, by Bochner's theorem for cylindrical vector measures \cite[Theorem
  A.17]{falconi2017arxiv}, $\widehat{\nn}_t$ is the Fourier transform of a unique state-valued
  cylindrical probability measure $ \nn_t$.

  Furthermore, by approximating $ \Upsilon_{\varepsilon_{n_k}}(t)$ with $ \Upsilon_{\varepsilon_{n_k}}(t_j)$ and using the uniform
  equicontinuity of the noncommutative Fourier transform, one can prove that
  		\begin{equation*}
    			\Upsilon_{\varepsilon_{n_k}}(t) \xrightarrow[k \to +\infty]{} \nn_t\; .
  		\end{equation*}
  		Here, we have used \cref{prop:10} to lift the convergence from the weak-$*$ to the weak
                topology. This in particular implies that $ \nn_t$ is a probability Radon measure on
                $\mathfrak{h}$, because it is a Wigner measure of $ \gint(t)$, satisfying the hypotheses
                of \cref{prop:8}, thanks to \cref{prop:12}.

                To summarize, we have defined the common subsequence, and the family of state-valued
                probability measures obtained in the limit at any arbitrary time. The last inequality
                \eqref{eq: a priori bound common} is finally proved again combining \cref{prop:8,prop:12}.
	\end{proof}
	
Once rewritten for the density matrix $ \Gamma_{\eps}(t) $, the result of \cref{prop:14} reads as follows:
		
	\begin{corollary}
  		\label{cor:4}
  		\mbox{}	\\
  		If $\lim_{\varepsilon\to 0}\varepsilon\nu(\varepsilon) = \nu \in \mathbb{R} $, then, under
  the same hypotheses of \cref{prop:14}, there exists a common subsequence $ \lf\{ n_k \ri\}_{k \in \N} $, such that, for any $t\in \mathbb{R}$,
  		\begin{equation}
    			\Gamma_{\varepsilon_{n_k}}(t) \xrightarrow[k \to + \infty]{} \mm_t := e^{-it \mathcal{K}_0} \, \lf( e^{-it\nu\omega} \: _{\star} \: \nn_t \ri)\, e^{it \mathcal{K}_0}\; ,
  		\end{equation}
  		where $e^{-it\nu\omega}\: _{\star} \: \nn_t $ is the measure obtained pushing forward $ \nn_t$ by means of the unitary map $e^{-it\nu\omega}:\mathfrak{h}\to \mathfrak{h}$. Furthermore, for any $T>0$, any $t\in [-T,T]$ and any $\delta' \leq \delta$,
  		\begin{equation}
    			\int_{\mathfrak{h}}^{} \mathrm{d} \mu_{\mm_t}(z) \: {\bigl(\lVert z  \rVert_{\mathfrak{h}}^2+1\bigr)}^{\delta'} \leq C(T)\; .
  		\end{equation}
  		where $C(T)$ is the same as in \eqref{eq: a priori bound common}.
	\end{corollary}
	
	\begin{proof}	
          The result trivially follows from \cref{prop:14} by identifying $e^{it \mathcal{K}_0}\mathcal{B}
          e^{-it \mathcal{K}_0}$, with $\mathcal{B}\in \mathscr{B}(\mathscr{H}) $, as the bounded operator
          for the weak convergence, and using a very general result for linear symplectic maps, and their
          quantization as maps on algebras of canonical commutation relations \cite[Proposition
          6.1]{falconi2017arxiv}.
	\end{proof}

Therefore, we have obtained a common convergent subsequence, and a map $t\mapsto \nn_t$ of
quasi-classical Wigner measures. The next step is to characterize such dynamical map explicitly by means
of a transport equation, and study the uniqueness properties of the latter. In order to do that, we study
the convergence of the integral equation provided in \cref{lemma:5}.

	\begin{proposition}[Transport equation for $ \nn(t) $]
  		\label{prop:15}
  		\mbox{}	\\
  		Under the same assumptions of \cref{prop:14}, the family of state-valued
  probability measures $ \lf\{ \nn_{t} \ri\}_{t \in \R} $ as in \eqref{eq: existence subsequence} satisfies in weak sense, \emph{i.e.},  when tested against any $\mathcal{B}\in \mathscr{B}(\mathscr{H}) $, the integral equation
  		\begin{equation}
  			\label{eq: fourier equation}
    			\widehat{\nn}_t(\eta)= \widehat{\nn}_s(\eta)-i\int_s^t  \mathrm{d}\tau  \int_{\mathfrak{h}}^{} \mathrm{d}\mu_{\nn_\tau}(z) \: \lf[ \widetilde{\mathcal{V}}_{\tau}(e^{-i\tau\nu\omega}z) , \gamma_{\nn_\tau}(z) \ri] e^{2i\Re \braket{\eta}{z}_{\mathfrak{h}}}   \; ,
  		\end{equation}
  		indexed by $\eta\in
  \mathfrak{h}$, where 
  		\beq
  			\widetilde{\mathcal{V}}_{\tau}(\: \cdot \:) := e^{i\tau \mathcal{K}_0}\mathcal{V}(\: \cdot \:)e^{-i\tau \mathcal{K}_0} =\sum_{j=1}^N
  e^{i\tau \mathcal{K}_0} 2\Re \braket{\lambda(\xv_j)}{\: \cdot \:}_{\mathfrak{h}} e^{-i\tau \mathcal{K}_0}
  		\eeq
  		is meant as a map  from $\mathfrak{h}$ to $ \mathscr{B}(\mathscr{H})$.
	\end{proposition}
	
	\begin{proof}
  		The existence of a common subsequence $ \lf\{ \eps_{n_k} \ri\}_{k \in \N} $, $ \eps_{n_k} \to 0 $, such that \eqref{eq: existence subsequence} holds true is guaranteed by \cref{prop:14}. Let us now fix $s,t\in \mathbb{R}$: given the convergence along the subsequence at any time, it is possible to take the limit $k\to \infty$ separately in all terms of the
  microscopic integral equation of evolution given in \cref{lemma:5}, traced against an arbitrary operator
  $\mathcal{B}\in \mathscr{B}(\mathscr{H}) $. 
  
  		For the integral term (second term on the r.h.s. of \eqref{eq: micro transport}), we make use of \cref{prop:10,prop:12}, where the
  latter is used to prove that $ \gint(\tau)$ satisfies the hypotheses of the former for all $ \tau \in
  [s,t]$, using $e^{-i\tau \mathcal{K}_0}\mathcal{B} e^{i\tau \mathcal{K}_0}$ as test operators. In order to do that, it is necessary to take the limit within the time integral. That is
  possible thanks to a dominated convergence argument, that makes use of the regularity assumption on
  $\Gamma_{\varepsilon}$: for any bounded operator $\mathcal{B}$, consider the integrand function
  		\begin{equation*}
  			I(\tau):= \sum_{j=1}^N\Tr \lf( \lf[\varphi_{\varepsilon}\bigl(e^{-i\tau\varepsilon\nu(\varepsilon)\omega}\lambda(\xv_j) \bigr), \gint(\tau) \ri] \bigl(e^{-i\tau \mathcal{K}_0}\mathcal{B} e^{i\tau \mathcal{K}_0}\otimes W_{\varepsilon}(\eta)\bigr) \ri)\; .
  	\end{equation*}
  	Its absolute value is bounded, using standard Fock space estimates, as 
  	\begin{equation*}
    		\lvert I(\tau)  \rvert_{}^{}\leq 2N \lf\| \lambda  \ri\|_{L^{\infty}(\mathbb{R}^d;,\mathfrak{h})}^{}\lVert \mathcal{B}  \rVert \Tr \lf( \gint(\tau){(\mathrm{d}\GG_{\varepsilon}(1)+1)}^{\frac{1}{2}} \ri) \; .
  	\end{equation*}
  Using \cref{prop:12} (see \eqref{eq: propagation} and \cref{rem: regularity gint}) and the regularity assumption on $\Gamma_{\varepsilon}$, it follows that the r.h.s.\ of the above expression is uniformly bounded by a finite constant. Hence, $ I(\tau) $ is integrable on any finite interval $[s,t]$, uniformly in $\varepsilon$.
	\end{proof}

\section{Uniqueness for the Quasi-Classical Equation of Transport}
\label{sec:uniq-quasi-class}

In this section we study the properties of the transport equation for state-valued measures obtained in
\cref{prop:15} as the quasi-classical limit of the microscopic evolution of states. 

The first technical point is discussed in \cref{lemma:9}
below, where it is proven that it is possible to exchange freely the two integrals of the aforementioned equation, which reads
\begin{equation}
	\label{eq: equation}
    	\mathrm{d}\nn_t(z)= \mathrm{d}\nn(z)-i\int_s^t \mathrm{d}\tau \mathrm{d} \mu_{\nn_\tau}(z) \: \lf[\widetilde{\mathcal{V}}_{\tau}(e^{-i\tau\nu \omega}z),\gamma_{\nn_\tau}(z) \ri]
 \end{equation}
 or, equivalently, using the Radon-Nikod\'{y}m decomposition $ \nn_t = \lf( \mu_{\nn_t}, \: \gamma_{\nn_t}(z) \ri) $,
\begin{equation}
  \label{eq:17}
  \gamma_{\nn_t} (z)\mathrm{d}\mu_{\nn_t}(z)= \gamma_{\nn_s}(z) \mathrm{d}\mu_{\nn_s}(z)-i\int_s^t\mathrm{d}\tau \:\mathrm{d}\mu_{\nn_{\tau}}(z) \lf[\widetilde{\mathcal{V}}_{\tau} (e^{-i\tau\nu\omega}z) , \gamma_{\nn_\tau}(z) \ri]  \; .
\end{equation}
Let us discuss the Bochner integrability of $ \bigl[\widetilde{\mathcal{V}}_{\tau} (e^{-i\tau\nu\omega}z) , \gamma_{\nn_\tau}(z)
\bigr]$ and justify the above statement.

	\begin{lemma}
  		\label{lemma:9}
  		\mbox{}	\\
  		Let $ \lf\{ \nn_{t} \ri\}_{t \in \R} $ be the family of state-valued measures as in \cref{prop:14}, then $\bigl[\widetilde{\mathcal{V}}_{t} (e^{-it\nu\omega}z) , \gamma_{\nn_t}(z) \bigr]$ is Bochner $\mu_{\nn_t}$-integrable
  for any $t\in \mathbb{R}$, and the norm of the integral is uniformly bounded w.r.t. $t$ on compact sets.
	\end{lemma}
	
	\begin{proof}
		By \eqref{eq: a priori bound common}, we immediately have that
		\begin{equation}
    			\label{eq:23}
    			\int_{\mathfrak{h}}^{} \mathrm{d} \mu_{\nn_t}(z) \lVert z  \rVert_{\mathfrak{h}} \leq C(t)\; 
 		\end{equation}
 		for all $t\in \mathbb{R}$ and for some $C(t) < + \infty $. Moreover, for $\mu_{\nn_t}$-almost all $z\in \mathfrak{h}$, $ \lf\|  \gamma_{\nn_t}(z)
  \ri\|_{\mathscr{L}^1(\mathscr{H})}^{}=1$, so that
  		\bdm
  		    \lf\| \lf[\widetilde{\mathcal{V}}_t (e^{-it\nu\omega}z) , \gamma_{\nn_t}(z) \ri]  \ri\|_{\mathscr{L}^1(\mathscr{H})}^{} \leq 2 \lf\| \widetilde{\mathcal{V}}_t (e^{-it\nu\omega}z)  \ri\| \leq 4N \lf\| \lambda  \ri\|_{L^{\infty}(\mathbb{R}^{d};\mathfrak{h})} \lf\| z  \ri\|_{\mathfrak{h}}^{}\;,
  		\edm
  		which implies the result via \eqref{eq:23}.
	\end{proof}

From now on, we assume that we are considering a solution of $t\mapsto \nn_t$ that
satisfies~\eqref{eq: fourier equation}. Let us introduce some terminology: a family of measures $t\mapsto
\nn_t$ solving \eqref{eq: fourier equation} in \cref{prop:15} for all $\eta\in \mathfrak{h}$ is called {\it weak} or {\it weak-$*$} \emph{Fourier solution}, if \eqref{eq: fourier equation} holds true when tested against bounded or compact operators, respectively. Note that every weak or weak-$*$ Fourier
solution is also a \emph{weak} or {\it weak-$*$ solution} of~\eqref{eq:17}, respectively, where the latter denote
solutions of the equation obtained testing with smooth cylindrical scalar functions instead of Fourier
characters. Let us specify further these last features. We first have to properly define the set of test cylindrical functions.

	\begin{definition}[Cylindrical functions]
  		\label{def:10}
  		\mbox{}	\\
  		A function $f: \mathfrak{h}\to \mathbb{C}$ is a smooth and compactly supported cylindrical function over
  		$\mathbb{P}\mathfrak{h}$, where $\mathbb{P}$ is an orthogonal projector and $\dim \mathbb{P} \mathfrak{h}<\infty$, iff there exists $g\in
  		C^{\infty}_0(\mathbb{P}\mathfrak{h})$ such that for all $z\in \mathfrak{h}$,
  		\begin{equation*}
   			 f(z)=g(\mathbb{P}z)\; .
  		\end{equation*}
  		We denote by $C_{0, \mathrm{cyl}}^{\infty}(\mathfrak{h})$ the set of all smooth cylindrical functions.
	\end{definition}
	
Now, let $f\in C_{0, \mathrm{cyl}}^{\infty}(\mathfrak{h})$, cylindrical over $\mathbb{P}\mathfrak{h}$, and
$\widehat{f}:\mathfrak{h}\to \mathbb{C}$ its Fourier transform, also cylindrical over $\mathbb{P}\mathfrak{h}$, defined as
\begin{equation}
  	\widehat{f}(\xi) =\int_{\mathbb{P}\mathfrak{h}}^{} \mathrm{d}\mathbb{P}z \: e^{-2\pi i \Re \braket{\xi}{\mathbb{P}z}} f(z) = \int_{\mathbb{P}\mathfrak{h}}^{} \mathrm{d}\mathbb{P}z  \: e^{-2\pi i \Re \braket{\xi}{\mathbb{P}z}} g(\mathbb{P}z)  =\widehat{g}(\mathbb{P}\xi)\; ,
\end{equation}
where $\mathrm{d} \mathbb{P}z $ stands for the Lebesgue measure on $\mathbb{P}\mathfrak{h}$. By testing \eqref{eq: fourier equation} against a cylindrical function $ \widehat{g}(\mathbb{P}\xi) $, we get 
\bmln{
      \int_{\mathbb{P}\mathfrak{h}}^{}\mathrm{d} \mathbb{P}\xi \: \widehat{g}(\mathbb{P}\xi)  \widehat{\nn}_t(\mathbb{P}\xi) \,= \int_{\mathbb{P}\mathfrak{h}}^{} \mathrm{d} \mathbb{P}\xi \: \widehat{g}(\mathbb{P}\xi) \widehat{\nn}_s(\mathbb{P}\xi) \\ -i \int_{\mathbb{P}\mathfrak{h}}^{}\mathrm{d} \mathbb{P}\xi \: \widehat{g}(\mathbb{P}\xi)  \int_s^t  \mathrm{d}\tau \int_{\mathfrak{h}}^{}   \mathrm{d}\nn_{\tau}(z)  \: \lf[\widetilde{\mathcal{V}}_{\tau}(e^{-i\tau\nu\omega}z), \gamma_{\nn_\tau}(z) \ri] e^{2\pi i\Re \braket{\mathbb{P}\xi}{z}_{\mathfrak{h}}} \,\; .
}
Hence, it follows that, for any  $f\in C_{0, \mathrm{cyl}}^{\infty}(\mathfrak{h})$,
\begin{equation}
  \label{eq:16}
  \int_{\mathfrak{h}}^{}  \mathrm{d}\nn_t(z) \: f(z)=\int_{\mathfrak{h}}^{}\mathrm{d}\nn_s(z) \: f(z)  -i\int_s^t \diff \tau \int_{\mathfrak{h}}^{}\mathrm{d}\nn_{\tau}(z) \: f(z)  \lf[\widetilde{\mathcal{V}}_{\tau}(e^{-i\tau\nu\omega}z), \gamma_{\nn_\tau}(z) \ri] \; .
\end{equation}

Now, let us fix $s\in \mathbb{R}$ as the initial time, and the corresponding $\nn_s\equiv \nn$ as
the initial datum. Then, the following $t\mapsto \nn_t$ is easily checked to be both a weak and weak-$*$
solution of~\eqref{eq:17}:
\begin{equation}
  \label{eq:3}
  t\mapsto \lf(\mu_{\nn_t}, \gamma_{\nn_t}(z) \ri) : = \lf( \mu_{\nn}, \: \widetilde{\mathcal{U}}_{t,s}(z) \gamma_{\nn}(z) \widetilde{\mathcal{U}}^{\dagger}_{t,s}(z) \ri)\;,
\end{equation}
where $\widetilde{\mathcal{U}}_{t,s}(z)$ is the two-parameter unitary group on $\mathscr{H} $ generated by
the time-dependent generator $ \widetilde{\mathcal{V}}_{\tau}(e^{-i t\nu\omega}z)\in \mathcal{L}(L^2)$. Note that such
an evolution two-parameter group exists for all $z\in \mathfrak{h}$ and $ t \in \mathbb{R}$, since the
$\widetilde{\mathcal{V}}_{t}(e^{-it\nu\omega}z)$ is bounded operators on $ \mathscr{H} $ (see, \emph{e.g.},
\cite{reed1975II}). Furthermore, the solution given by \eqref{eq:3} satisfies~\eqref{eq:23} at all times,
provided the inequality is satisfied by the initial datum.

It just remains to prove the solution in \eqref{eq:3} is actually \emph{unique}. This of course might depend on the notion of solution we adopt, but proving weak-$*$ uniqueness, we get also uniqueness for stronger solutions (weak, Fourier weak-$*$, and Fourier weak). As a matter of fact, the proof of uniqueness is actually independent of the notion of solution considered.

	\begin{proposition}[Uniqueness for the transport equation for $ \nn_t $]
  		\label{prop:16}
  		\mbox{}	\\
  		Let $s\in \mathbb{R}$ be the fixed initial time, and let $ \nn_s \equiv \nn \in
  \mathscr{M}(\mathfrak{h};\mathscr{H})$ be a Borel state-valued measure such that 
 		\begin{equation*}
    			\int_{\mathfrak{h}}^{} \mathrm{d}\mu_{\nn}(z)\lVert z  \rVert_{\mathfrak{h}} < C\; .
  		\end{equation*}
  		Then, the integral transport equation \eqref{eq: equation} admits a unique weak-$*$ solution $ \nn_t$, which satisfies~\eqref{eq:23}, defined by its norm Radon-Nikod\'{y}m decomposition
  		\begin{equation}
    			\lf( \mu_{\nn_t}, \: \gamma_{\nn_t}(z) \ri) = \lf( \mu_{\nn}, \:\widetilde{\mathcal{U}}_{t,s}(z) \: \gamma_{\nn}(z) \: \widetilde{\mathcal{U}}^{\dagger}_{t,s}(z) \ri)\; .
  		\end{equation}
  		Such solution is continuous and differentiable on every Borel set in the strong topology of
  $\mathscr{L}^1(\mathscr{H})$ and its the derivative $\partial_t \nn_t$ is a self-adjoint but in general not positive
  state-valued measure.
	\end{proposition}

	\begin{proof}
  		Any weak solution $ \nn_t$ of the transport equation \eqref{eq: equation} or \eqref{eq:17}
  satisfying~\eqref{eq:23} is continuous and can be weakly differentiated w.r.t. time on Borel sets. However, given the structure of equation \eqref{eq:17}, it is easy to realize that such a derivative actually exists in the strong topology of $\mathscr{L}^1(\mathscr{H})$ and reads
  		\begin{equation}
    			\label{eq:18}
    			\mathrm{d}\,\partial_t \nn_t(z)= -i \lf [\widetilde{\mathcal{V}}_t(e^{-it\nu \omega}z),\gamma_{\nn_t}(z) \ri] \mathrm{d}\mu_{\nn_t}(z)\; .
  		\end{equation}
  
  		To prove uniqueness, suppose that $\nn_t$ is a solution satisfying~\eqref{eq:23}. Since we already know that \eqref{eq:3} solves the equation, it is sufficient to prove that $ \nn_t$ admits the
  Radon-Nikod\'{y}m decomposition \eqref{eq:3} (recall \cref{prop:7}). In order to do that, let us set
  		\begin{equation*}
    			\mathrm{d}{\widetilde{\nn}}_t(z):=\widetilde{\mathcal{U}}_{t,s}^{\dagger}(z) \gamma_{\nn_t}(z) \widetilde{\mathcal{U}}_{t,s}(z) \mathrm{d} \mu_{\nn_t}(z)\;,
  		\end{equation*}
  		so that, using \eqref{eq:17} once more, we get
  		\bmln{
      		i\mathrm{d}\partial_t\tilde{\nn}_t(z)=\widetilde{\mathcal{U}}^{\dagger}_{t,s}(z) \lf[\widetilde{\mathcal{V}}_t(e^{-it\nu \omega}z) , \gamma_{\nn_t}(z) \ri] \widetilde{\mathcal{U}}_{t,s}(z) \mathrm{d} \mu_{\nn_t}(z) \\
      		- \widetilde{\mathcal{U}}_{t,s}^{\dagger}(z) \lf[\widetilde{\mathcal{V}}_t(e^{-it\nu \omega}z), \gamma_{\nn_t}(z) \ri] \widetilde{\mathcal{U}}_{t,s}(z) \mathrm{d} \mu_{\nn_t}(z)  =0.
    		}
 		 Hence, $\tilde{\nn}_t= \tilde{\nn}_s = \nn $. Therefore,
  $ \nn_t$ has indeed norm Radon-Nikod\'{y}m decomposition \eqref{eq:3}.
\end{proof}

\section{Putting It All Together: Proof of \cref{thm:5}}
\label{sec:putting-it-all}

It is now possible to combine the results obtained in
\cref{sec:quasi-class-analys,sec:microscopic-model,sec:quasi-class-limit,sec:uniq-quasi-class}, and thus
prove \cref{thm:5}. We first state and prove the result for the evolution in the interaction picture and under a stronger assumption on the initial datum, and then complete the proof by relaxing it and going back to the evolution for $ \Gamma_{\eps}(t) $.

	\begin{proposition}[Quasi-classical evolution in the interaction picture]
  		\label{prop:17}
  		\mbox{}	\\
  		Let $\Gamma_{\varepsilon}\in \mathscr{L}^1_{+,1}(\mathscr{H} \otimes \mathscr{K}_{\eps})$ be such that there exists $\delta > \frac{1}{2}$, so
  that
  		\begin{equation*}
   			\Tr \lf( \Gamma_{\varepsilon}{(\mathrm{d}\GG_{\varepsilon}(1)+1)}^{\delta}  \ri)_{}\leq C\;.
  		\end{equation*}
  		Then, if $\Gamma_{\varepsilon_n}\to \mm$, 
  		\beq
  			\Upsilon_{\varepsilon_n}(t) \xrightarrow[n \to + \infty]{} \nn_t,		\qquad		\forall t \in \R,
		\eeq
		where $ \nn_t = \lf( \mu_{\mm}, \: \widetilde{\mathcal{U}}_{t,0}(z) \gamma_{\mm}(z) \widetilde{\mathcal{U}}_{t,0}^{\dagger}(z) \ri) $. 
	\end{proposition}

	\begin{proof}
 		Let us consider $ \Upsilon_{\varepsilon_n}(t)$: by \cref{prop:14}, there exists a
  common subsequence $\varepsilon_{n_k}\to 0$ such that for all $t\in \mathbb{R}$, 					
  		\bdm
  			\Upsilon_{\varepsilon_{n_k}}(t) \xrightarrow[k \to + \infty]{} \nn_t.
		\edm
		Clearly, $\nn_0 =\mm$. Moreover, by \cref{prop:15}, $\nn_t$ is also a weak
  solution of~\eqref{eq:17}, satisfying \cref{lemma:9} and \eqref{eq:23}. The weak solution of~\eqref{eq:17}
  satisfying~\eqref{eq:23} is however unique by \cref{prop:16}, and therefore $ \nn_t $ has the
  Radon-Nikod\'{y}m decomposition
  		\begin{equation*}
    			\bigl(\mu_{\mm}\,,\,\widetilde{\mathcal{U}}_{t,0}(z) \gamma_{\mm}(z) \widetilde{\mathcal{U}}^{\dagger}_{t,0}(z) \bigr)\; .
 		 \end{equation*}
 		 
 		 We now show that the convergence holds at any time along the original subsequence $ \{ \eps_n \}_{n \in \N} $. Let us take a convergent subsequence of $  \Upsilon_{\eps_n}(t) $ at an arbitrary time $ t $, i.e., such that
 		\beq
			\Upsilon_{\varepsilon_{n_j}}(t) \xrightarrow[j \to + \infty]{} \nn'_t.
		\eeq
		Then, by convergence at time $ t = 0 $, we immediately have that
		\bdm
				\Upsilon_{\varepsilon_{n_j}}(0) \xrightarrow[j \to + \infty]{} \mm.
		\edm
		Hence, we can proceed as in the proof of \cref{prop:14} and extract a subsequence $ \lf\{ \eps_{n_{j_k}} \ri\}_{k \in \N} $ of $ \lf\{ \eps_{n_j} \ri\}_{j \in \N} $ such that we have convergence at any time. Furthermore, again by \cref{prop:15}, the limit points $ \nn'_t $ are weak solutions of the transport equation. Therefore, by uniqueness of the solution, $ \nn'_t = \nn_t $. Hence, all the convergent subsequences of $ \Upsilon_{\eps_n}(t)  $ have the same limit point, which implies that $ \Upsilon_{\eps_n}(t) \xrightarrow[n \to+\infty]{} \nn_t $.
	\end{proof}

	The analogue for $ \Gamma_{\eps} $ of \cref{prop:17} is next

	\begin{corollary}
  		\label{cor:6}
  		\mbox{}	\\
  		Under the same hypotheses of \cref{prop:17}, for any $t\in \mathbb{R}$,
		\beq
			\Gamma_{\varepsilon_n}(t) \xrightarrow[n \to + \infty]{} \mm_t = e^{-it\mathcal{K}_0}\bigl(e^{-it\nu \omega}\, _{\star} \,\nn_t\bigr)e^{it \mathcal{K}_0}\; .
		\eeq
	\end{corollary}

	\begin{proof}
		If \cref{prop:17} holds, then \cref{cor:6} is a direct consequence of \cref{cor:4}.
	\end{proof}

The proof of \cref{thm:5} is almost complete, it remains only to extend the result to states $\Gamma_{\varepsilon}\in \mathscr{L}^1_{+,1}(\mathscr{H} \otimes \mathscr{K}_{\eps})$   satisfying the weaker condition that there exists $\delta>0$ and $C< + \infty$, such that
\begin{equation}
  \label{eq:25}
  \Tr \lf(\Gamma_{\varepsilon}{(\mathrm{d}\GG_{\varepsilon}(1)+1)}^{\delta} \ri)\leq C \; .
\end{equation}
This is done by standard approximation techniques, using an argument originally proposed in
\cite[\textsection2]{ammari2011jmpa} (see also \cite[\textsection4.5]{ammari2014jsp}). Let us briefly reproduce the key
ideas here. Let $\Gamma_{\varepsilon}$ satisfy~\eqref{eq:25}, and define
\begin{equation}
 	\Gamma_{\varepsilon}^{(r)} : = \frac{\chi_r(\mathrm{d}\GG_{\varepsilon}(1)+1) \Gamma_{\varepsilon} \chi_r(\mathrm{d}\GG_{\varepsilon}(1)+1) }{ \Tr\bigl( \chi_r(\mathrm{d}\GG_{\varepsilon}(1)+1) \Gamma_{\varepsilon} \chi_r(\mathrm{d}\GG_{\varepsilon}(1)+1) \bigr)}\in \mathscr{L}^1_{+,1}(\mathscr{H} \otimes \mathscr{K}_{\eps}) \;,
\end{equation}
where $r> 0$ and $\chi_r(\cdot )=\chi(\cdot /r)$, $\chi\in C_0^{\infty}(\mathbb{R})$, with $0\leq \chi\leq 1$ and $\chi=1$ in a neighborhood of zero. By
functional calculus and~\eqref{eq:25}, for any $t\in \mathbb{R}$,
\begin{equation*}
  	\lf\| \Gamma_{\varepsilon}(t)-\Gamma_{\varepsilon}^{(r)}(t)  \ri\|_{\mathscr{L}^1_{+,1}(\mathscr{H} \otimes \mathscr{K}_{\eps})}^{} = 
  	\lf\| \Gamma_{\varepsilon}-\Gamma_{\varepsilon}^{(r)}  \ri\|_{\mathscr{L}^1_{+,1}(\mathscr{H} \otimes \mathscr{K}_{\eps})}^{} = o_r(1)\; ,
\end{equation*}
when $r\to \infty$, uniformly w.r.t. $\varepsilon\in (0,1)$. In addition, $\Gamma_{\varepsilon}^{(r)}$ satisfies the assumptions of
\cref{prop:17}. Suppose now that $\Gamma_{\varepsilon_n}\to \mm$, and for all $r>0$, let $\varepsilon_{n_k(r)}\to 0$ be a
subsequence and $\mm^{(r)}$ a state-valued measure such that $\Gamma_{\varepsilon_{n_k(r)}}^{(r)}\to
\mm^{(r)}$. Then, by \cref{cor:6}, $\Gamma_{\varepsilon_{n_k(r)}}^{(r)}(t)\to
\mm^{(r)}_t$, for any $t\in \mathbb{R}$, where the latter is defined by \cref{thm:5}, with $\mm^{(r)}$ in place of $\mm$. Finally, let us extract a subsequence $\varepsilon_{n_{k_\ell}(r,t)}\to 0$ such that
$\Gamma_{\varepsilon_{n_{k_l}(r,t)}}(t)\to \vec{\nu}_t$. By adapting the argument in \cite[Proposition
2.10]{ammari2011jmpa} to state-valued measures, we get that for any fixed $t\in \mathbb{R}$
\begin{equation*}
  	\int_{\mathfrak{h}}^{}  \mathrm{d} \lf| \mu_{\vec{\nu}_t} - \mu_{\mm^{(r)}_t} \ri|    = o_r(1)\; ,
\end{equation*}
where $\lvert\mu_{\vec{\nu}_t} - \mu_{\mm^{(r)}_t} \rvert_{}^{}$ is the scalar measure in the norm Radon-Nikod\'{y}m
decomposition of the total variation of the signed state-valued measure $\vec{\nu}_t - \mm_t^{(r)}$
\cite[\textsection A.3]{falconi2017arxiv}, \emph{i.e.}, the sum of its positive and negative parts. Therefore,
denoting by $\mm_t$ the measure appearing in \cref{thm:5}, we have
\bdm
   	\int_{\mathfrak{h}}^{}  \mathrm{d} \lf| \mu_{\vec{\nu}_t} - \mu_{\mm_t} \ri|  \leq \int_{\mathfrak{h}}^{}  \mathrm{d} \lf| \mu_{\vec{\nu}_t} - \mu_{\mm^{(r)}_t} \ri|  + \int_{\mathfrak{h}}^{}  \mathrm{d} \lf|\mu_{\mm^{(r)}_t} - \mu_{\mm_t} \ri| = o_r(1)\; .
\edm
Therefore, $\vec{\nu}_t \equiv\mm_t$. Since any subsequence extraction yields the same result, it follows that,
for all $t\in \mathbb{R}$, $\Gamma_{\varepsilon_n}(t)\to \mm_t$, thus concluding the proof of \cref{thm:5}. 

\cref{cor:7} is then a trivial application of the definition of w-$*$ convergence and we thus omit the proof. It only remains to prove \cref{cor:2} and \cref{thm:8}.

\begin{proof}[Proof of \cref{cor:2}] 
	We first of all remark that the quasi-classical evolution $t\mapsto \mathfrak{m}_t$ preserves the mass, \emph{i.e.},
  \begin{equation*}
    \forall t,t'\in \mathbb{R}\,, 	\qquad \lVert \mathfrak{m}_t(\mathfrak{h})  \rVert_{\mathscr{L}^1(\mathscr{H})}^{}=\lVert \mathfrak{m}_{t'}(\mathfrak{h})  \rVert_{\mathscr{L}^1(\mathscr{H})}^{}\; .
  \end{equation*}
  Therefore, for the first part of the statement, it suffices to prove that, under assumption \eqref{eq:
      ass1} or \eqref{eq: ass2}, $\lVert \mathfrak{m}(\mathfrak{h})
    \rVert_{\mathscr{L}^1(\mathscr{H})}^{}=1$. 
    
    In the case of assumption \eqref{eq: ass1}, we can actually show  that quasi-classical convergence can be lifted to bounded quasi-classical convergence. In fact, let  $\mathcal{B}\in \mathscr{B}(\mathscr{H})$ and consider
  	\begin{equation*}
   		 \tr_{\mathscr{H}} \lf(\hat{\Gamma}_{\varepsilon}(\eta)\mathcal{B} \ri)=\tr_{\mathscr{H}} \lf( \hat{\Gamma}_{\varepsilon,A}(\eta)\mathcal{A}^{-\frac{1}{2}}\mathcal{B} \mathcal{A}^{-\frac{1}{2}} \ri)\; ,
  	\end{equation*}
 	 where
  	\begin{equation*}
    		\hat{\Gamma}_{\varepsilon,A}(\eta)=\tr_{\mathscr{K}_{\varepsilon}} \lf(A^{\frac{1}{2}}\Gamma_{\varepsilon}A^{\frac{1}{2}}W_{\varepsilon}(\eta) \ri)\; ,
  	\end{equation*}
  	is the Fourier transform of the state $A^{\frac{1}{2}}\Gamma_{\varepsilon}A^{\frac{1}{2}}\in \mathscr{L}^1_+(\mathscr{H}\otimes
  \mathscr{K}_{\varepsilon})$ by \eqref{eq: ass1}. Therefore, on one hand, by \cref{prop:6}, if $\Gamma_{\varepsilon_n}\to \mathfrak{m}$, then
  $A^{\frac{1}{2}}\Gamma_{\varepsilon}A^{\frac{1}{2}}\to \mathcal{A}^{\frac{1}{2}}\mathfrak{m} \mathcal{A}^{\frac{1}{2}}$,
  on the other hand since $\mathcal{A}^{-\frac{1}{2}}\mathcal{B} \mathcal{A}^{-\frac{1}{2}}\in
  \mathscr{L}^{\infty}(\mathscr{H})$, \cref{thm:5} guarantees that
	\begin{multline}
  		\tr_{\mathscr{H}} \lf( \hat{\Gamma}_{\varepsilon_n}(\eta) \mathcal{B} \ri) \\
  		= \tr_{\mathscr{H}} \lf(\hat{\Gamma}_{\varepsilon_n,A}(\eta)\mathcal{A}^{-\frac{1}{2}}\mathcal{B} \mathcal{A}^{-\frac{1}{2}} \ri)   		\xrightarrow[n\to+\infty]{} \int_{\mathfrak{h}}^{}  \mathrm{d}\mu_{\mathfrak{m}} \: e^{2i\Re \lf\langle  \eta \lf. \ri\vert z \ri\rangle_{\mathfrak{h}}} \tr_{\mathscr{H}} \lf(\mathcal{A}^{\frac{1}{2}} \gamma_{\mathfrak{m}}(z)\mathcal{A}^{\frac{1}{2}}\mathcal{A}^{-\frac{1}{2}}\mathcal{B} \mathcal{A}^{-\frac{1}{2}} \ri)	\\
  = \tr_{\mathscr{H}} \lf(\hat{\mathfrak{m}}(\eta)\mathcal{B} \ri)\; .
	\end{multline}
	Choosing $\eta=0$ and $\mathcal{B}=\mathds{1}$ one concludes that $\lVert \mathfrak{m}(\mathfrak{h})  \rVert_{\mathscr{L}^1(\mathscr{H})}^{}=1$. 
	
	The proof assuming \eqref{eq: ass2} uses a dominated convergence
  argument. Let us denote by $ \lf\{ \mathcal{P}_k \ri\}_{k\in \mathbb{N}}\subset \mathscr{L}^{\infty}(\mathscr{H})$ the projections onto
  the eigenspaces of $\mathfrak{m}(\mathfrak{h})$, {\it i.e.}, 
  	\bdm
  		\mm(\mathfrak{h}) = \sum_{k \in \N} m_k \mathcal{P}_k,
	\edm
	$ m_k \in \R^+ $ being the associated eigenvalues. Then, for any $k\in \mathbb{N}$, we can apply \cref{thm:5} to get
	\begin{equation}
		\Tr\bigl(\Gamma_{\varepsilon_n}P_k\bigr) = \tr_{\mathscr{H}}\bigl(\gamma_{\varepsilon_n}\mathcal{P}_k\bigr)\xrightarrow[n\to + \infty]{} \tr_{\mathscr{H}}\bigl(\mathfrak{m}(\mathfrak{h})\mathcal{P}_k\bigr)=\alpha_k m_k\; ,
	\end{equation}
	where $\alpha_k\in \mathbb{N}$ is the multiplicity of the eigenvalue $ m_k$ of $\mathfrak{m}(\mathfrak{h})$. Hence, by \eqref{eq: ass2},
	\begin{equation*}
  		1=\Tr\bigl(\Gamma_{\varepsilon}\bigr) = \tr_{\mathscr{H}}\bigl(\gamma_{\varepsilon}\bigr)=\sum_{k\in \mathbb{N}}^{}\tr_{\mathscr{H}}\bigl(\gamma_{\varepsilon} \mathcal{P}_k\bigr)\leq \sum_{k\in \mathbb{N}}^{}\tr_{\mathscr{H}}\bigl(\gamma \mathcal{P}_k\bigr)=\tr_{\mathscr{H}}\bigl(\gamma\bigr)<+\infty\; .
	\end{equation*}
	Therefore, by dominated convergence,
	\bml{
  		1 = \lim_{n\to + \infty}\Tr\bigl(\Gamma_{\varepsilon_n}\bigr)=\lim_{n\to +\infty}\sum_{k\in \mathbb{N}}^{}\tr_{\mathscr{H}}\bigl(\gamma_{\varepsilon_n} \mathcal{P}_k\bigr)=\sum_{k\in \mathbb{N}}^{}\lim_{n\to + \infty} \tr_{\mathscr{H}}\bigl(\gamma_{\varepsilon_n} \mathcal{P}_k\bigr)	\\
  		=\sum_{k\in \mathbb{N}}^{}\alpha_k m_k=\lVert \mathfrak{m}(\mathfrak{h})  \rVert_{\mathscr{L}^1(\mathscr{H})}^{}\; .
 	}
\end{proof}

\begin{proof}[Proof of \cref{thm:8}]
	First of all, remark assumption \eqref{eq:29} is propagated in time by means of \cref{prop:12}. In addition, the measure $\mathfrak{m}_t$ is characterized by \cref{thm:5}  at any time $t\in \mathbb{R}$. Finally, the convergence of the expectations of Wick quantizations of symbols $\mathcal{F}\in \mathscr{S}_{\ell,m}$, under condition \eqref{eq:29}, is given
  by \cref{prop:10}. Combining the above ingredients \cref{thm:8} is then proved.
\end{proof}

\section{Technical Modifications for Pauli-Fierz and Polaron Models}
\label{sec:techn-modif-polar}

\cref{thm:5} is stated not only for the regularized Nelson model, but also for the Pauli-Fierz and polaron
models as well. The strategy of the proof for these cases is identical to the one followed above for the
Nelson model. However, one shall overcome some technical difficulties related to the fact that such models
are ``more singular''. In particular, the foremost difficulty is given by the presence of terms of type
$\nabla\cdot a_{\varepsilon}^{\#}\bigl(\vec{\lambda}(\xv)\bigr)$ and their adjoints in the microscopic Hamiltonian $H_{\varepsilon}$. In
relation to that, one needs to propagate in time some further regularity of quantum states, in addition to what is done in \cref{prop:12} for the Nelson model. Finally, some care has to be taken in defining the
effective limit dynamics $ \mathcal{U}_{t,s}(z)$. We comment below on the technical adaptations needed to take care of such difficulties.

\subsection{Quasi-Classical Analysis of Gradient Terms}
\label{sec:quasi-class-analys-1}

In order to deal with terms of the form $\nabla\cdot a_{\varepsilon}^{\#}\bigl(\vec{\lambda}(\xv)\bigr)$, with $\vec{\lambda}\in
  L^{\infty}(\mathbb{R}^d; \mathfrak{h}^d)$, one needs to extend the convergence proven in \cref{prop:10} to such kind
  of observables. This is done in two steps: first, it is possible to restrict the set of test observables
  using the set $\mathfrak{K}$ defined in \cref{lemma:7}, for it separates points, and then we prove that
  with such a restriction the expectation values indeed converge (\cref{prop:18}). In particular,
  \cref{lemma:7}, is used below for the convergence of gradient terms, to solve possible domain
  ambiguities whenever the gradient acts on the test operator: we end up with a form of the integral
transport equation for the measure that holds only when tested with particle observables in
$\mathfrak{K}$ (recall \eqref{eq: set K}), setting $\mathcal{T}=\mathcal{K}_0$, where $\mathcal{K}_0$ is the self-adjoint free
particle Hamiltonian. With such testing it still makes sense to study uniqueness of the solution, since
the aforementioned set separates points.

Let us now consider the convergence of the expectation value of the gradient term. Let us recall that
$a^{\#}_{\varepsilon}(f)$ stands for either $a_{\varepsilon}(f)$ or $a^{\dagger}_{\varepsilon}(f)$, and correspondingly $
\braket{f}{z}_{\mathfrak{h}}^{\#}$ stands for either $ \braket{f}{z} _{\mathfrak{h}}$ or $
\braket{z}{f}_{\mathfrak{h}}$. Let us recall that in all the concrete models considered,  $\mathcal{K}_0 \geq p > - \infty $, and 
  		\beq
  			\label{eq: apriori bound nabla}
  			\lvert \nabla \rvert_{}^{}{(\mathcal{K}_0+1-p)}^{-\frac{1}{2}}\in
  \mathscr{B}(\mathscr{H}).
		\eeq

	\begin{proposition}[Convergence of expectation values of gradient terms]
  		\label{prop:18}
  		\mbox{}	\\
 		Let $\Gamma_{\varepsilon}\in \mathscr{L}^1_{+,1}(\mathscr{H} \otimes \mathscr{K}_{\eps})$ be such that there exists $\delta>1$ so that
  		\begin{equation}
  			\label{eq: apriori bound K0}
    			\Tr\bigl( \Gamma_{\varepsilon}(K_0-p+{(\mathrm{d}\GG_{\varepsilon}(1)+1)}^\delta)  \bigr)\leq C \; .
  		\end{equation}
		Then,  if $\Gamma_{\varepsilon_n} \xrightarrow[n \to + \infty]{} \mm$, for any $ \mathcal{B}\in \mathfrak{K} $, any $\alpha,\beta\in \mathbb{R}$ and all $\eta\in \mathfrak{h}$,
  		\bml{
    			\label{eq:21}
      		\lim_{n\to +\infty}\Tr \lf(\Gamma_{\varepsilon_n} \lf(\alpha\nabla\cdot a^{\#}_{\varepsilon_n}(\vec{\lambda}(\xv))+\beta a^{\#}_{\varepsilon_n}(\vec{\lambda}(\xv))\cdot \nabla \ri) \bigl(\mathcal{B}\otimes W_{\varepsilon_n}(\eta)\bigr) \ri) \\
      		=\tr_{\mathscr{H}} \lf( \int_{\mathfrak{h}}^{} \mathrm{d}\mu_{\mm}(z) \: \gamma_{\mm}(z) \bigl(\alpha\nabla\cdot \braket{\vec{\lambda}(\xv)}{z}^{\#}_{\mathfrak{h}}+\beta \braket{\vec{\lambda}(\xv)}{z}^{\#}_{\mathfrak{h}}\cdot \nabla \bigr)e^{2i\Re \braket{\eta}{z}_{\mathfrak{h}}}\mathcal{B}\, \ri)\; ,
   		}
   		\bml{
    			\label{eq:22}
      		\lim_{n\to +\infty}\Tr \lf(\lf(\alpha\nabla\cdot a^{\#}_{\varepsilon_n}(\vec{\lambda}(\xv))+\beta a^{\#}_{\varepsilon_n}(\vec{\lambda}(\xv))\cdot \nabla \ri) \Gamma_{\varepsilon_n} \bigl(\mathcal{B}\otimes W_{\varepsilon_n}(\eta)\bigr) \ri) \\
      		=\tr_{\mathscr{H}} \lf( \int_{\mathfrak{h}}^{} \mathrm{d}\mu_{\mm}(z) \: \bigl(\alpha\nabla\cdot \braket{\vec{\lambda}(\xv)}{z}^{\#}_{\mathfrak{h}}+\beta \braket{\vec{\lambda}(\xv)}{z}^{\#}_{\mathfrak{h}}\cdot \nabla \bigr)e^{2i\Re \braket{\eta}{z}_{\mathfrak{h}}} \gamma_{\mm}(z) \mathcal{B}\, \ri)\;.
   		}	
	\end{proposition}

	\begin{proof}
          We prove the result for $\Gamma_{\varepsilon}\bigl(\nabla\cdot a_{\varepsilon}(\vec{\lambda}(\xv))+a_{\varepsilon}(\vec{\lambda}(\xv))\cdot \nabla\bigr)$, the
          other cases being perfectly analogous.

          First of all, we observe that ${(K_0-p)}^{\frac{1}{2}}\Gamma_{\varepsilon_n}{(K_0-p)}^{\frac{1}{2}}$ is a
          positive operator and we can then consider its quasi-classical convergence as $ n \to + \infty $: by \cref{prop:6}, 
                \begin{equation}
                  	{(K_0-p)}^{\frac{1}{2}}\Gamma_{\varepsilon_n}{(K_0-p)}^{\frac{1}{2}} \xrightarrow[n \to + \infty]{} {(K_0-p)}^{\frac{1}{2}}\mm{(K_0-p)}^{\frac{1}{2}}\; .
                \end{equation}
  		
  		The term $\Gamma_{\varepsilon} a_{\varepsilon}(\vec{\lambda}(\xv))\cdot \nabla$, in which the gradient acts directly on $\mathcal{B} $ converges by
  \cref{prop:10}, since $\partial_j \mathcal{B}\in \mathscr{L}^{\infty}(\mathscr{H})$ for all $j=1,\dotsc,d$ and
  $\mathcal{B}\in \mathscr{L}^{\infty}(\mathscr{H}) $.  

  		It remains to discuss the term $\Gamma_{\varepsilon}\nabla\cdot a_{\varepsilon}(\vec{\lambda}(\xv))$. This term requires suitable
  approximations. First of all, let us approximate each operator-valued symbol 
  		\bdm
  			F_{j}^{(\vec{\lambda})}(z) : = \braketl{\lambda_{j}(\xv)}{z}_{\mathfrak{h}}
		\edm
		$ j=1,\dotsc,d$, by means of \cref{lemma:4}, and let us denote its
  approximation by $ F_{j,M}^{(\vec{\lambda})}$. It follows that, using estimates analogous to the ones used in
  the proof of \cref{prop:10},
  		\bdm
      		\lf| \sum_{j =1}^d\Tr \lf( \Gamma_{\varepsilon}\partial_{j} \lf[ a_{\varepsilon}(\lambda_{\mu}(\xv))- \mathrm{Op}^{\mathrm{Wick}}_{\varepsilon} \lf(F_{j,M}^{(\vec{\lambda})} \ri) \ri] \bigl(\mathcal{B}\otimes W_{\varepsilon}(\eta)\bigr) \ri)  \ri|_{}^{}\leq C \lVert \mathcal{B}  \rVert \sum_{j=1}^d \lf\| F_{j}^{(\vec{\lambda})} - F_{j,M}^{(\vec{\lambda})} \ri\| \; ,
      	\edm
  and the r.h.s. does not depend on $\varepsilon$, and converges to zero when $M\to + \infty$. In addition, let us
  recall that the symbol $ F_{j,M}^{(\vec{\lambda})}$ has the form:
  	\begin{equation*}
    		F_{j,M}^{(\vec{\lambda})}= \sum_{k=1}^{J(M)} \braketl{\varphi_{j,k}}{z}_{\mathfrak{h}} \one_{B_k}(\xv)\;,
  	\end{equation*}
  	where $J(M)\in \mathbb{N}$, $\varphi_{j,k}\in \mathfrak{h}$, and $B_k\subseteq \mathbb{R}^d$ is a Borel set. Let us consider the convergence as
  $\varepsilon_n\to 0$ of each term of the above sums separately, for $M$ fixed. In other words, let us consider the
  convergence of
  	\begin{equation*}
    		\Tr \Bigl( \Gamma_{\varepsilon_n}\partial_{j}a_{\varepsilon_n}(\varphi_{j,k}) \one_{B_k}(\xv)\bigl(\mathcal{B}\otimes W_{\varepsilon_n}(\eta)\bigr) \Bigr) = \tr_{\mathscr{H}}\Bigl\{ \tr_{\mathscr{K}_{\eps}}\bigl[ \Gamma_{\varepsilon_n} a_{\varepsilon_n}(\varphi_{j,k}) W_{\varepsilon_n}(\eta) \bigr] \partial_{j} \one_{B_k}(\xv) \mathcal{B} \Bigr\} \; .
  	\end{equation*}
  	The operator $a_{\varepsilon}(\varphi_{\mu,k}) W_{\varepsilon}(\eta)$ is the product of the Weyl quantizations of
  two cylindrical albeit not compactly supported symbols, over the complex Hilbert subspace spanned by
  $\varphi_{j,k}$ and $\eta$. Therefore, by finite dimensional pseudodifferential calculus, for all $M$ there
  exists a smooth compactly supported scalar symbol $ F_{\sigma}^{(j,k,\eta)}\in C_{0,
    \mathrm{cyl}}^{\infty}(\mathfrak{h})$, such that, for any $\delta'>\frac{1}{2}$,
  	\begin{equation}
    		\lf\| \lf[ a_{\varepsilon}(\varphi_{j,k}) W_{\varepsilon}(\eta) - \mathrm{Op}^{\mathrm{Weyl}}_{\varepsilon} \lf( F_{\sigma}^{(j,k,\eta)} \ri) \ri]{\bigl(\mathrm{d}\GG_{\varepsilon}(1)+1\bigr)}^{-\delta'} \ri\|_{\mathscr{B}(\mathscr{K}_{\eps}
)}^{} = o_{\sigma}(1)+ o_{\varepsilon}(1) \; .
  	\end{equation}
  	Hence, Cauchy-Schwarz inequality, \eqref{eq: apriori bound K0} and \eqref{eq: apriori bound nabla} yield
  	\bml{
          \lf| \Tr \Bigl( \Gamma_{\varepsilon_n}\partial_{j} \lf( a_{\varepsilon_n}(\varphi_{j,k}) - \mathrm{Op}^{\mathrm{Weyl}}_{\varepsilon} \lf(F_{\sigma}^{(j,k,\eta)} \ri) \ri) \one_{B_k}(\xv)\bigl(\mathcal{B}\otimes W_{\varepsilon_n}(\eta)\bigr) \Bigr) \ri| \\
          \leq C \lf\| \mathcal{B} \ri\|  \lf\|  \lf[ a_{\varepsilon}(\varphi_{j,k}) W_{\varepsilon}(\eta) - \mathrm{Op}^{\mathrm{Weyl}}_{\varepsilon} \lf( F_{\sigma}^{(j,k,\eta)} \ri) \ri]{\bigl(\mathrm{d}\GG_{\varepsilon}(1)+1\bigr)}^{-\delta'} \ri\|_{\mathscr{B}(\mathscr{K}_{\eps}\otimes\mathscr{H})} \times	\\
          \times \lf| \tr \lf( (\mathcal{K}_0+1-p)^{\frac{1}{2}} {\bigl(\mathrm{d}\GG_{\varepsilon}(1)+1\bigr)}^{\delta'} \Gamma_{\varepsilon_n} \ri) \ri| \\
          \leq C \lf\|  \lf[ a_{\varepsilon}(\varphi_{j,k}) W_{\varepsilon}(\eta) - \mathrm{Op}^{\mathrm{Weyl}}_{\varepsilon} \lf( F_{\sigma}^{(j,k,\eta)} \ri) \ri]{\bigl(\mathrm{d}\GG_{\varepsilon}(1)+1\bigr)}^{-\delta'} \ri\|_{\mathscr{B}(\mathscr{K}_{\eps}\otimes\mathscr{H})} \times	\\
          \times \lf| \tr \lf[ \lf(\mathcal{K}_0+1-p + \bigl(\mathrm{d}\GG_{\varepsilon}(1)+1\bigr)^{2\delta'} \ri) \Gamma_{\varepsilon_n} \ri] \ri| 
          = o_{\sigma}(1)+ o_{\varepsilon}(1) \; .  } 
     In addition, for all $\delta'>\frac{1}{2}$,
  	\begin{equation}
   		\label{eq:26}
      	\lf( \braketl{\varphi_{j ,k}}{z}_{\mathfrak{h}} e^{2i\Re \braket{\eta}{ z}_{\mathfrak{h}}}- F_{\sigma}^{(j,k,\eta)}(z) \ri) \lf(\lVert z  \rVert_{\mathfrak{h}}^2+1 \ri)^{-\delta'}  = o_{\sigma}(1)\;
 	\end{equation}
 	uniformly in $ z \in \mathfrak{h} $.
  We can now take the limit $\varepsilon_n\to 0$ of the remaining term
  \bmln{
      \Tr\Bigl(\Gamma_{\varepsilon_n}\partial_{j}\mathrm{Op}^{\mathrm{Weyl}}_{\varepsilon}\lf(F_{\sigma}^{(j,k,\eta)} \ri)  \one_{B_k}(\xv)B \Bigr) \\
      = \tr_{\mathscr{H}}\Bigl( \tr_{\Gamma_{\mathrm{s}}}\Bigl({(K_0-p)}^{\frac{1}{2}}\Gamma_{\varepsilon_n} {(K_0-p)}^{\frac{1}{2}} \mathrm{Op}^{\mathrm{Weyl}}_{\varepsilon}\lf(F_{\sigma}^{(j,k,\eta)} \ri)  \Bigr) {(K_0-p)}^{-\frac{1}{2}}\partial_{j} \one_{B_k}(\xv) \mathcal{B}{(K_0-p)}^{-\frac{1}{2}} \Bigr) \; .
  }
  Since the symbol $ F_{\sigma}^{(j,k,\eta)} \in C_{0, \mathrm{cyl}}^{\infty}(\mathfrak{h})$, this converges to (see \cite{falconi2017arxiv} for further details)
  \bmln{
      \tr_{\mathscr{H}} \lf( \int_{\mathfrak{h}}^{} \mathrm{d}\mu_{\mm}(z) \: {(K_0-p)}^{\frac{1}{2}}\gamma_{\mm}(z) {(K_0-p)}^{\frac{1}{2}} F_{\sigma}^{(j,k,\eta)}(z) {(K_0-p)}^{-\frac{1}{2}} \partial_{j} \one_{B_k}(\xv) \mathcal{B}{(K_0-p)}^{-\frac{1}{2}} \ri)\\
      =\tr_{\mathscr{H}} \lf( \int_{\mathfrak{h}}^{} \mathrm{d}\mu_{\mm}(z) \:  \gamma_{\mm}(z) {(K_0-p)}^{\frac{1}{2}} F_{\sigma}^{(j,k,\eta)}(z) {(K_0-p)}^{-\frac{1}{2}} \partial_{j}\one_{B_k}(\xv) \mathcal{B} \ri).
 	}
  The limit $\sigma\to +\infty$ can then be taken by dominated convergence, at fixed $M$, thanks to the uniform bound for $ \delta > 1 $
  \begin{equation}
    \label{eq:27}
    \int_{\mathfrak{h}}^{}\mathrm{d}\mu_{\mm}(z) \: \lf( \tr_{\mathscr{H}}( \gamma_{\mm}(z)(K_0-p))+{\lf(\lVert z  \rVert_{\mathfrak{h}}^2+1\ri)}^{\delta} \ri)  \leq C\;,
  \end{equation}
  which also allows to take the limit $M\to +\infty$.          
\end{proof}

\subsection{Propagation Estimates and the Pull-Through Formula}
\label{sec:pull-thro-form}

In this section we discuss the so-called pull-through formula, needed both to characterize the dynamics in the quasi-classical limit for the polaron model: as we are going to see, the pull-through formula is key to propagate the a priori bounds on the initial state at later times. The
formula holds for the massive Nelson and the polaron model, therefore $H_{\varepsilon}$ in this section stands
for any of such Hamiltonians as defined above, although it is not needed for the Nelson model with
ultraviolet cut off, as we are considering in this paper. Indeed, in that case, one can simply use the
propagation estimates of \cref{eq: propagation 1}, valid also in the massless case.

Before discussing the formula, let us remark that the Pauli-Fierz and polaron Hamiltonians are
self-adjoint and bounded from below. There is an extensive literature concerning the self-adjointness of
the Pauli-Fierz Hamiltonian (see, \emph{e.g.},
\cite{hiroshima2000cmp,hiroshima2002ahp,spohn2004dcp,hasler2008rmp,falconi2015mpag,matte2017mpag} and
references therein), which, under our assumptions, is self-adjoint on $ \dom(K_0)\cap
\dom(\mathrm{d}\GG_{\varepsilon}(\omega))$. The polaron Hamiltonian is also self-adjoint
\cite{frank2014lmp,griesemer2015arxiv}, but its domain of self-adjointness is not explicitly
characterized. On the other hand, its form domain is known, and it coincides with the form domain of
$K_0+\nu(\varepsilon)\mathrm{d}\GG_{\varepsilon}(1)$.

We do not prove the pull-through formula, since it is discussed in detail for the massive
renormalized Nelson model in \cite{ammari2000mpag}, and its independence of the semiclassical parameter
has been shown in \cite{ammari2017sima}. The models we consider here are ``contained'' in the
massive renormalized Nelson model, namely all the terms in the Hamiltonians contained here are part
or are analogous to some parts of the renormalized Nelson Hamiltonian. Therefore, they have already been
discussed in the aforementioned papers.

	\begin{proposition}[Pull-through formula]
  		\label{prop:22}
  		\mbox{}	\\
  		There exist two finite constants $ a, b $ independent of $ \eps $, such that for any  $\varepsilon\in (0,1)$ and for any $\Psi_{\varepsilon}\in \dom(H_{\varepsilon})$,
  		\beq
    			\lf\lVert\mathrm{d}\GG_{\varepsilon}(1)\Psi_{\varepsilon} \ri\rVert_{\mathscr{H}\otimes \mathscr{K}_{\varepsilon}}^{}\leq \frac{a}{\nu(\varepsilon)} \lf\lVert (H_{\varepsilon}+b)\Psi_{\varepsilon}  \ri\rVert_{\mathscr{H}\otimes \mathscr{K}_{\varepsilon}}^{}\;.
  		\eeq
\end{proposition}

To study the quasi-classical limit of the (massless) Pauli-Fierz model, we cannot
  use the pull-through formula, we use instead the following propagation result (see \cite{ammari2019prep} for a detailed proof).

	\begin{proposition}[Propagation estimate]
  		\label{prop:23}
  		\mbox{}	\\
  		Let $H_{\varepsilon}$ be the Pauli-Fierz Hamiltonian, with either $\nu(\varepsilon)=1$ or $\nu(\varepsilon)=\frac{1}{\varepsilon}$. Then, there
  exist two finite constants $C_1,C_2 $ independent of $ \eps $, such that for any $\varepsilon\in (0,1)$, for any $\Psi_{\varepsilon}\in \dom(K_0)\cap \dom(\mathrm{d}\GG_{\varepsilon}(\omega))\cap
  \dom(\mathrm{d}\GG_{\varepsilon}(1))$ and for any $t\in \mathbb{R}$,
  		\bml{
    			\lf\lVert \mathrm{d}\GG_{\varepsilon}(1) e^{-it H_{\varepsilon}}\Psi_{\varepsilon} \ri\rVert_{\mathscr{H}\otimes \mathscr{K}_{\varepsilon}}^{}\leq C_1 \lf[ \lf\lVert \mathrm{d}\GG_{\varepsilon}(1)\Psi_{\varepsilon}  \ri\rVert_{\mathscr{H}\otimes \mathscr{K}_{\varepsilon}}^{} \ri.	\\
    			\lf. + \lf\lVert \lf(K_0+\mathrm{d}\GG_{\varepsilon}(\omega)+1 \ri)\Psi_{\varepsilon}  \ri\rVert_{\mathscr{H}\otimes \mathscr{K}_{\varepsilon}}^{} \ri] e^{C_2 \lvert t  \rvert_{}^{}}\;.
  		}
  		In addition, there exist two finite constants $c,C>0$ independent of $ \eps $, such that for any $\varepsilon\in (0,1)$ and for any $\Psi_{\varepsilon}\in \dom(H_{\varepsilon})=\dom(K_0)\cap
  \dom(\mathrm{d}\GG_{\varepsilon}(\omega))$,
  		\begin{equation}
   			 c\lf\| (H_{\varepsilon}+1) \Psi_{\varepsilon} \ri\|_{\mathscr{H}\otimes \mathscr{K}_{\varepsilon}}^{}\leq \lf\lVert \lf(K_0+ \nu(\varepsilon) \mathrm{d}\GG_{\varepsilon}(\omega)+1 \ri) \Psi_{\varepsilon} \ri\rVert_{\mathscr{H}\otimes \mathscr{K}_{\varepsilon}}^{} \leq C \lf\lVert (H_{\varepsilon}+1) \Psi_{\varepsilon} \ri\rVert_{\mathscr{H}\otimes \mathscr{K}_{\varepsilon}}^{}\; .
 		 \end{equation}
	\end{proposition}

        Let us now outline in more detail how one can use the pull-through formula in the adaptations of
        the arguments to cover the polaron model. The main technicality
        is the propagation of the a priori bound and regularity of the state. This can be achieved by a direct application of \cref{prop:22}: one can simply restrict the
        proof of \cref{thm:5} to states satisfying \beqn
  	\label{eq:polass1}& \Tr\bigl( \Gamma_{\varepsilon}\bigl((K_0+\mathrm{d}\GG_{\varepsilon}(\omega)+1)^2+\mathrm{d}\GG_{\varepsilon}(1)^2 \bigr) \bigr)\leq C\; ;\\
        \label{eq:polass2}& \Tr\bigl( \Gamma_{\varepsilon}H_{\varepsilon}^2\bigr)\leq C \nu(\varepsilon)^2\;, \eeqn for any $ \eps \in (0,1)
        $. Let us remark that the regularity assumptions above are not propagated in time as
            they are, but they are rather used to control the following expectations at any time $t\in \mathbb{R}$:
        \begin{itemize}
        \item $\Tr \bigl(\Gamma_{\varepsilon}(t) K_0\bigr)$;
          
        \item $\Tr \bigl(\Gamma_{\varepsilon}(t)(\mathrm{d}\mathcal{G}_{\varepsilon}(1)+1)^2 \bigr)$.
        \end{itemize}
        The first expectation is bounded uniformly w.r.t. $\varepsilon$ as in \cite[Lemma
          3.4]{carlone2019arxiv}, using the assumption \eqref{eq:polass1}. The second expectation is
          bounded using \cref{prop:22} and assumption \eqref{eq:polass2}. Once the bounds for the two
          quantities above are established at any time, it is possible to use \cref{prop:18} for the
          quasi-classical convergence of the interaction terms appearing in the integral equation. The
          result is then extended to general states satisfying~\eqref{eq:20} by means of the procedure
          outlined in~\cref{sec:putting-it-all}.

          For the Pauli-Fierz model one proceeds similarly, using \cref{prop:23} instead of the
            pull-through formula.  \cref{thm:5} is first proved for initial states such that
\begin{equation}
  \Tr\bigl(\Gamma_{\varepsilon}(K_0+\mathrm{d}\GG_{\varepsilon}(\omega)+{(\mathrm{d}\GG_{\varepsilon}(1)+1)}^2)\bigr)\leq C\; 
\end{equation}
for $  \eps \in (0,1) $. The needed regularity of the expectation of the number operator at any time is then obtained thanks
  to \cref{prop:23}. To bound the free particle part, one proceeds analogously as it
  was done for the polaron model in the aforementioned result \cite[Lemma 3.4]{carlone2019arxiv}, the
  only difference being that instead of using KLMN-smallness, which would be true only for small values of
  the particles' charge, one uses again the number estimate of \cref{prop:23} to close the argument
  (see \cite{olivieri2019phd} for additional details). Therefore, it is possible to apply \cref{prop:18}
  for the quasi-classical convergence of the gradient terms appearing in the integral equation, and
  \cref{prop:10} for the convergence of quadratic terms in the creation and annihilation operators. The proof can then be completed exactly as for the polaron model.

\subsection{Quasi-Classical Evolution}
\label{sec:quasi-class-evol}

In this section we briefly discuss the well-posedness of the effective evolution equation for the polaron
  and Pauli-Fierz models with $\nu(\varepsilon)=\frac{1}{\varepsilon}$. In fact, when $\nu(\varepsilon)=1$ or $ \nu = 0 $, the generator of
$\mathcal{U}_{t,s}(z)$, $\mathcal{K}_0+\mathcal{V}(z)$, does not depend on time and therefore the
evolution is defined by Stone's theorem, and the interaction picture $\widetilde{\mathcal{U}}_{t,s}(z)$ is
weakly generated by $t\mapsto \widetilde{\mathcal{V}}_{t}(z)$.

If on the opposite $ \nu(\eps) = \frac{1}{\eps} $, we need to prove the existence of a two parameter group of unitary operators
${(\mathcal{U}_{t,s}(z))}_{t,s\in \mathbb{R}}$, satisfying
\begin{equation*}
  \begin{split}
    &i\partial_t\mathcal{U}_{t,s}(z)\psi_s=(\mathcal{K}_0+\mathcal{V}_t(z))\mathcal{U}_{t,s}(z)\psi_s,\\
    &i\partial_s\mathcal{U}_{t,s}(z)\psi_s=-\mathcal{U}_{t,s}(z)(\mathcal{K}_0+\mathcal{V}_s(z))\psi_s,\\
    &\mathcal{U}_{s,s}(z)\psi_s=\psi_s,
  \end{split}
\end{equation*}
for any $\psi\in \dom(\mathcal{K}_0)$, where the latter is also the domain of self-adjointness for
$\mathcal{K}_0+\mathcal{V}_t(z)$, for all $t\in \mathbb{R}$.

For the Pauli-Fierz model, it is easy to prove that the map $ t \mapsto \mathcal{K}_0+\mathcal{V}_t(z)$ is
strongly continuously differentiable on $\dom(\mathcal{K}_0)$. Therefore, a result by \cite{schmid2014mpag}
guarantees the existence of ${(\mathcal{U}_{t,s}(z))}_{t,s\in \mathbb{R}}$. Again, $\tilde{\mathcal{U}}_{t,s}(z)$ is
then defined by $\widetilde{\mathcal{U}}_{t,s}(z)=e^{it\mathcal{K}_0}\mathcal{U}_{t,s}(z)e^{-is
  \mathcal{K}_0}$, and it is weakly generated by $\widetilde{\mathcal{V}}_{t}(z)$. For the polaron
  model, on the other hand, the existence of the quasi-classical dynamics follows from a general result concerning the
  evolution generated by time-dependent closed quadratic forms with a time-independent common core, proved, {\it e.g.}, in \cite[Theorem II.27 \& Corollary II.28]{simon1971psp}.

\end{document}